\documentclass[showpacs,amsmath,amssymb,onecolumn,superscriptaddress,notitlepage,preprintnumbers,aps,pra,10pt,floatfix]{revtex4-2}
\usepackage{hyperref}
\hypersetup{colorlinks=true, linkcolor=blue, citecolor=blue, urlcolor=black }

\usepackage[capitalise,nameinlink]{cleveref}

\usepackage{algorithm}
\usepackage{algpseudocode}
\usepackage[normalem]{ulem}
\usepackage{qcircuit}

\usepackage[dvips]{graphicx}
\usepackage{amsmath,amssymb,amsthm,mathrsfs,amsfonts,dsfont}
\usepackage{subfigure, epsfig}
\usepackage{braket}
\usepackage{bm}
\usepackage{enumerate}
\usepackage{color}
\usepackage{graphicx}
\usepackage{mathrsfs}
\usepackage{multirow}
\usepackage{xcolor}
\definecolor{cream}{RGB}{203, 237, 204}
\newtheorem{theorem}{Theorem}
\newtheorem{lemma}{Lemma}
\newtheorem{corollary}{Corollary}

\DeclareMathOperator{\tr}{\mathrm{Tr}}
\DeclareMathOperator{\support}{\mathrm{supp}}
\DeclareMathOperator*{\rprod}{\overrightarrow{\prod}}
\DeclareMathOperator*{\lprod}{\overleftarrow{\prod}}
\DeclareMathOperator{\Dist}{\mathrm{Dist}}

\newcommand{\id}{\mathbb{I}}

\newcommand{\normH}[1]{{\left\vert\kern-0.25ex\left\vert\kern-0.25ex\left\vert #1
   \right\vert\kern-0.25ex\right\vert\kern-0.25ex\right\vert}}

\renewcommand{\d}{\mathrm{d}}
\newcommand{\mc}{\mathcal}

\let\originalleft\left
\let\originalright\right
\renewcommand{\left}{\mathopen{}\mathclose\bgroup\originalleft}
\renewcommand{\right}{\aftergroup\egroup\originalright}

\newcommand{\re}{\mathrm{re}}
\newtheorem{fact}{Fact}

\begin{document}


\title{Entanglement accelerates quantum simulation}

\author{Qi Zhao}
\email[]{zhaoqi@cs.hku.hk}
\affiliation{
QICI Quantum Information and Computation Initiative, Department of Computer Science,
The University of Hong Kong, Pokfulam Road, Hong Kong}

\author{You Zhou}
\email[]{you\_zhou@fudan.edu.cn}
\affiliation{Key Laboratory for Information Science of Electromagnetic Waves (Ministry of Education), Fudan University, Shanghai 200433, China}

\author{Andrew M. Childs}
\email[]{amchilds@umd.edu}
\affiliation{Department of Computer Science, Institute for Advanced Computer Studies, and Joint Center for Quantum Information and Computer Science, University of Maryland, College Park, Maryland 20742, USA}


\begin{abstract}
Quantum entanglement is an essential feature of many-body systems that impacts both quantum information processing and fundamental physics. The growth of entanglement is a major challenge for classical simulation methods.
In this work, we investigate the relationship between quantum entanglement and quantum simulation, showing that product-formula approximations can perform better for entangled systems. We establish a tighter upper bound for algorithmic error in terms of entanglement entropy and develop an adaptive simulation algorithm incorporating measurement gadgets to estimate the algorithmic error. This shows that entanglement is not only an obstacle to classical simulation, but also a feature that can accelerate quantum algorithms.
\end{abstract}

\maketitle

\section*{Introduction}
Quantum entanglement in many-body systems \cite{Amico2008Entanglement} is a pivotal topic with significant implications in quantum information processing \cite{Horodecki2009entanglement,walter2016multipartite} and the exploration of fundamental quantum physics \cite{Bei2019Meets,Qi2018gravity}. While the dynamics of weakly entangled quantum systems can be efficiently simulated \cite{vidal03efficient}, the precise role of entanglement in quantum algorithms remains unclear---for example, mixed-state computation appears to be capable of quantum speedup even when the system never has much entanglement~\cite{DattaPRL08}.

Understanding the entanglement dynamics stemming from an unentangled initial state is a fundamental concern in many-body physics, which lies at the heart of quantum thermalization (or its failure) \cite{Alessio2016quantum,gogolin2016equilibration,nandkishore2015many,serbyn2021quantum}. The widely studied Eigenstate Thermalization Hypothesis (ETH) suggests that quantum states effectively thermalize under the evolution of typical Hamiltonians even in a closed quantum system \cite{Alessio2016quantum,gogolin2016equilibration}, 
enabled by the linear growth of subsystem entanglement entropy. However, recent advances demonstrate that there are instances where the ETH fails, such as in many-body localized Hamiltonians \cite{nandkishore2015many} with strong disorder, which can slow the entanglement growth to a logarithmic trend over time; and many-body scarred Hamiltonians, which fail to thermalize for some initial states \cite{serbyn2021quantum}.

Tractable simulation methods are needed to investigate these intriguing quantum many-body phenomena. Classical simulation tools like tensor networks have been developed to explore both ground states and dynamics \cite{cirac2021matrix,orus2019tensor}. However, the growth of quantum entanglement poses a challenge for these classical methods. For instance, one-dimensional matrix product states (MPSs) \cite{cirac2021matrix} are typically suitable only for area-law-entangled states.

Engineered quantum platforms such as superconducting-qubit, trapped-ion, and cold-atom simulators offer opportunities to observe and analyze various quantum phases and dynamics \cite{altman2021quantum}. In particular, digital Hamiltonian simulation algorithms offer a robust and efficient means to simulate the dynamics of quantum systems \cite{sethuniversal}.
Although quantum advantage over classical simulation with a digital quantum simulator has yet to be demonstrated, numerous algorithms have been proposed to approximate the time-evolution operator $U_0(t):=e^{-iHt}$ \cite{berry2007efficient,berry2012black,TaylorSeries,PhysRevLett.118.010501,low2019hamiltonian}, and more accurate error analyses have been given, including commutator bounds for product-formula approximations \cite{childs2019nearly,childs2018toward,childs2020theory} and analyses that take advantage of particular input states \cite{sahinoglu2020hamiltonian,An2021timedependent,Zhaorandom22,chen2024average,su2020nearly}.
Previous error analyses have not shown a role for the amount of entanglement in the performance of these algorithms.

In this work, we explore the relationship between quantum entanglement and Hamiltonian simulation. We find that entanglement entropy has a direct impact on the performance of Hamiltonian simulation algorithms, and in particular, that it can accelerate the performance of product-formula simulations. We establish a tighter upper bound for the algorithmic error in Hamiltonian simulation using the entanglement entropy of subsystems. Surprisingly, in contrast to classical simulation, quantum simulation becomes \emph{more} efficient for entangled states. As entanglement grows, the Hamiltonian simulation error, or Trotter error, diminishes and converges to the average-case Trotter error \cite{Zhaorandom22}, so that simulations of sufficiently entangled states perform similarly to simulations with random inputs. In particular, this holds for states that are close to $k$-uniform \cite{scott2004multipartite} for sufficiently large $k$. As a result, entanglement implies more efficient Hamiltonian simulation. Conversely, we also present examples showing that product states can exhibit the worst-case performance of Hamiltonian simulation. 

For this result to be useful, we should know that the state is sufficiently entangled. 
Typical states are highly entangled; indeed, a randomly chosen state is almost surely close to maximally entangled \cite{hayden2006aspects}.
However, understanding when a particular state is entangled may be challenging. To address this issue and leverage the benefit of entanglement in practice, we develop an adaptive simulation algorithm that incorporates measurement gadgets in the middle of the simulation to estimate the upcoming Trotter error. This method can outperform standard Hamiltonian simulation with only negligible measurement overhead. 

Our analysis helps to reveal the potential of quantum entanglement in quantum information processing and fundamental quantum physics.
We hope this sheds light on the intricate relationship between quantum resources, quantum computing, and the development of quantum algorithms that can take advantage of features such as entanglement.

\section*{Results}
\subsection*{Entanglement entropy and Trotter error} Quantum simulation aims to realize the time-evolution operator $U_0(t):=e^{-iHt}$ of a given Hamiltonian $H$. This task is crucial for studying both dynamic and static properties of quantum systems. 
Two major classes of digital Hamiltonian simulation algorithms are product formula (PF) methods~\cite{suzuki1991general,sethuniversal} and algorithms based on linear combinations of unitaries~\cite{berry2007efficient,berry2012black,TaylorSeries,PhysRevLett.118.010501,low2019hamiltonian}. 
Here, we focus on the former, which are simpler and may perform better in practice, at least for some tasks \cite{childs2018toward}.

For a short evolution time $\delta t$, the first-order product formula (PF1) algorithm for a Hamiltonian with decomposition $H=\sum_{l=1}^LH_l$ applies the unitary operation $\mathscr{U}_1(\delta t):=e^{-iH_1\delta t}e^{-iH_2\delta t}\cdots e^{-iH_L\delta t}=\rprod_l e^{-iH_l \delta t}$. 
Here the right arrow indicates the product is in the order of increasing indices.
This provides an effective simulation method when the terms $e^{-iH_l \delta t}$ are easy to implement.
Second-order product formulas (PF$2$)
can be obtained by combining evolutions in both increasing and decreasing orders of indices, with $\mathscr{U}_2(\delta t):=\rprod_l e^{-iH_l\delta t/2} \lprod_l e^{-iH_l\delta t/2}$.
More generally, Suzuki constructed $p$th-order product formulas $\mathscr{U}_{p}$, for even $p$, recursively from the second-order formula~\cite{suzuki1991general}. 
In general, for a $p$th-order product formula method (PF$p$), we can quantify the algorithmic error in terms of the spectral norm of the additive error, $\|U_0-\mathscr{U}_{p}\|=\max_{\ket{\psi}}\|U_0\ket{\psi}-\mathscr{U}_{p}\ket{\psi}\|$, with $\|\cdot\|$ on the right-hand side denoting the vector $l_2$ norm. The spectral norm corresponds to the worst-case input state and thus represents the worst-case Trotter error. 
However, this worst-case error can be unnecessarily pessimistic. In practice, we may only care about some specific initial state $\ket{\psi}$, for which the error $\|{U_0\ket{\psi}-\mathscr{U}_{p}\ket{\psi}}\|$ may be smaller~\cite{sahinoglu2020hamiltonian,An2021timedependent,Zhaorandom22,su2020nearly}.

To simulate the evolution for a long time, one can divide the duration into several segments, each of which is a short-time evolution that can be simulated with small error. The total error is then upper bounded via the triangle inequality as the sum of the Trotter errors for each segment. This analysis is generally reasonably tight (except for PF1, which can exhibit destructive error interference~\cite{Tran_2020,layden2021first}). Thus, in the following, we first focus on a short-time evolution $\delta t$ and relate the Trotter error in one segment to the entanglement entropy of the initial quantum state.

For a quantum evolution $U_0=e^{-iH\delta t}$, consider a PF$p$ approximation $\mathscr{U}_{p}$ with $\|U_0-\mathscr{U}_{p}\|\le \|\sum_j E_j \|\delta t^{p+1}+\|\mathscr{E}_{\mathrm{re}}\|$, where $E=\sum_j E_j$ is the total leading-order error with local terms $E_j$ and $\|\mathscr{E}_{\mathrm{re}}\|=\mathcal O(\delta t^{p+2})$ is the contribution of the higher-order remaining terms \cite{childs2020theory}. 
The \emph{support} of an operator $E_j$, denoted $\support(E_j)$, is the set of qubits on which it acts nontrivially. Let $w(E_j)$ denote the number of qubits of $\support(E_j)$.

We consider the reduced density matrix (RDM) of the state 
$\ket{\psi}$ on $\support(E_j^{\dag}E_{j'})$,  $\rho_{j,j'}:=\tr_{[N]\setminus\support(E_j^{\dag}E_{j'})}(\ket{\psi}\bra{\psi})$ (with $[N]$  the set of all qubits),
with Hilbert space dimension $d_{\rho_{j,j'}}:=\smash{2^{w(E_j^{\dag}E_{j'})}}$.
For a $d$-dimensional operator $A$, let $\|A\|$ denote the spectral norm
and $\|A\|_F:=\sqrt{{\tr (AA^{\dagger})}/{d} }$ the normalized Frobenius norm.
The von Neumann entropy is denoted by $S$.
Then, we can bound the Trotter error via the following theorem. 

\begin{figure}
    \centering
       \includegraphics[width=1\columnwidth]{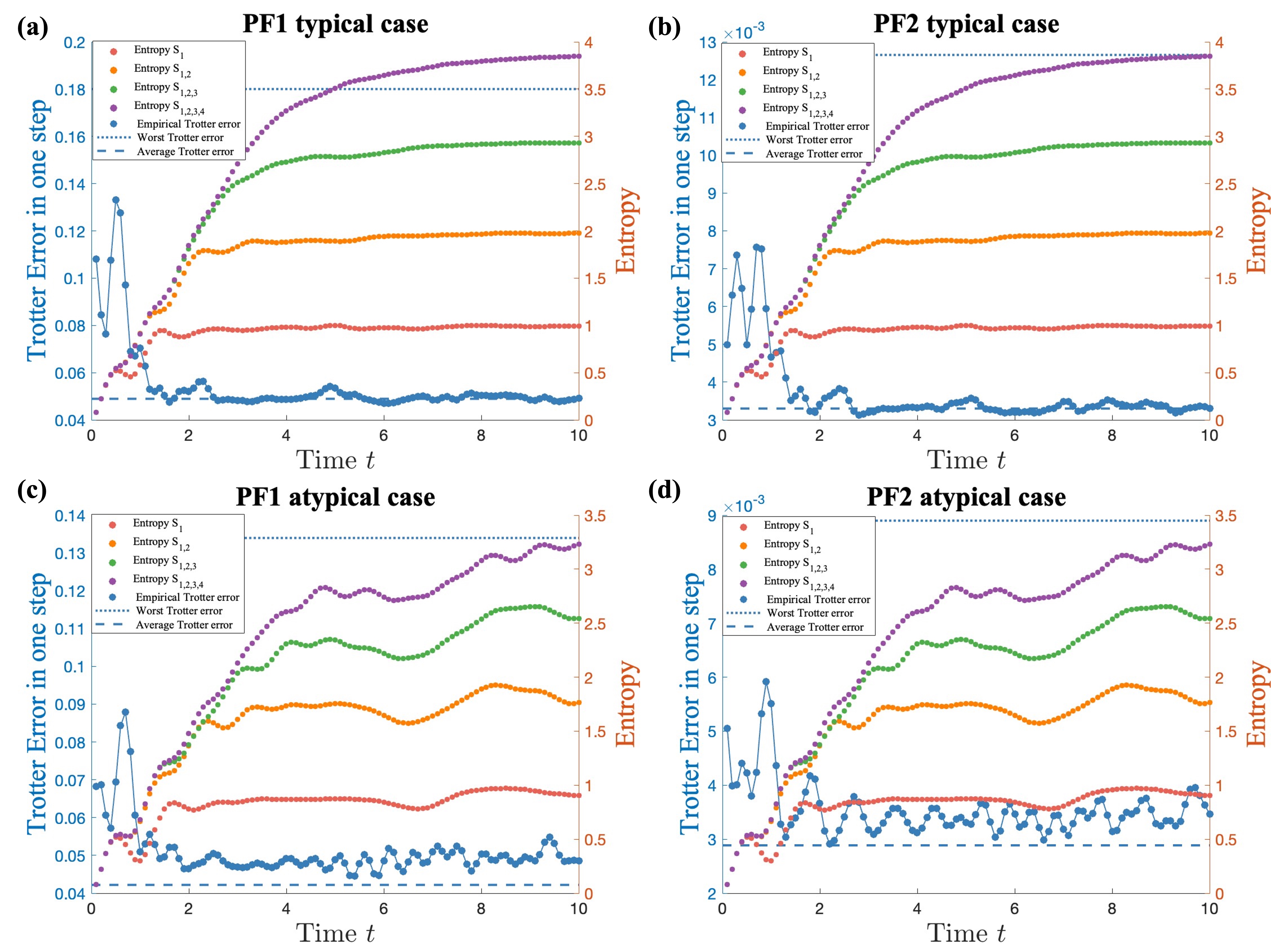}
    \caption{Entanglement entropy and Trotter error in one segment. The Trotter error is $\|(U_0(\delta t)-\mathscr{U}_p(\delta t))\ket{\psi(t)}\|$ for $\delta t=0.1$, where $\ket{\psi(t)}=e^{-iHt}\ket{\psi(0)}$ is the ideal state at time $t$ with $H$ the QIMF model of $N=12$ qubits.
    (a) and (b) show the PF1 and PF2 methods, respectively, with the typical-case Hamiltonian parameter choice $(h_x,h_y,J)= (0.8090, 0.9045, 1)$. The entanglement entropies of various subsystems all increase with time to the corresponding maximal value, as indicated by the right Y-axis. As the entanglement entropy increases, the empirical Trotter error (blue dotted curve) decreases and converges to $\|U_0-\mathscr{U}_{p}\|_F$, the average-case error (blue dashed line). (c) and (d) show analogous results for the atypical-case Hamiltonian parameter choice $(h_x,h_y,J)= (0, 0.9045, 1)$. In this case the entanglement entropy does not increase to its maximal value and there is a clear gap between the empirical Trotter error and the average-case error. 
    }
 \label{fig:thermalH}
\end{figure}

\begin{theorem}\label{thm:main}
For a given pure quantum state $\ket{\psi}$ and quantum evolution $U_0=e^{-iH\delta t}$ with $p$th-order Trotter approximation $\mathscr{U}_{p}$, the error in a Trotter step of duration $\delta t$ has the upper bound
\begin{equation}\label{eq:ErrD1}
\|(U_0-\mathscr{U}_{p})\ket{\psi}\|=\mathcal O\left(\delta t^{p+1}\sqrt{\sum_{j,j'} \|E_j^{\dag}E_{j'}\| \tr|\rho_{j,j'}-\mathbb{I}/d_{\rho_{j,j'}}|}+ \delta t^{p+1}\|E\|_F \right).
\end{equation}
We call this the \emph{distance-based error bound}.
One can further relate the Trotter error to the entanglement entropy of subsystems with the bound
\begin{equation}\label{eq:ErrEnt}
\|(U_0-\mathscr{U}_{p})\ket{\psi}\|=\mathcal O\left(\delta t^{p+1}\sqrt{\sum_{j,j'} \|E_j^{\dag}E_{j'}\|\sqrt{\log (d_{\rho_{j,j'}})-S(\rho_{j,j'})}}+ \delta t^{p+1} \|E\|_F\right).
\end{equation} 
We call this the \emph{entanglement-based error bound}. 
\end{theorem}

We sketch the proof in the Methods section, and present a rigorous proof with an explicit analysis of the higher-order terms in Section I of the Supplemental Material.
If the entanglement entropy of each RDM $\rho_{j,j'}$ is small, then the entanglement-based error bound in \cref{thm:main} is $\mathcal O(\delta t^{p+1}\sum_j\|E_j\|)$, recovering the worst-case Trotter error result~\cite{childs2020theory}. On the other hand, if the entanglement entropies of subsystems $\support(E_j^{\dag}E_{j'})$ satisfy $S(\rho_{j,j'})\ge w(E_j^{\dag}E_{j'})- \mathcal O\bigl(\bigl(\frac{\|E\|_F}{\sum_j\|E_j\|}\bigr)^4\bigr)$, the error is significantly reduced and scales as $\mathcal O(\delta t^{p+1}\|E\|_F)$, recovering the average-case Trotter error bound~\cite{Zhaorandom22}.

If the Hamiltonian is local, the Trotter error terms $E_j$ have low weights, and the  subsystems $\support(E_j^{\dag}E_{j'})$ are also small. For example, consider the $N$-qubit one-dimensional quantum Ising spin model with mixed fields (QIMF). The Hamiltonian is $H = h_x \sum_{j=1}^N X_j + h_y \sum_{j=1}^N Y_j + J \sum_{j=1}^{N-1} X_jX_{j+1}$, with $X_j, Y_j, Z_j$ denoting Pauli operators on qubit $j$. The dynamics $e^{-iH\delta t}$ can be approximated by $\mathscr{U}_{1} = e^{-iA\delta t}e^{-iB\delta t}$ with $A = h_x \sum_{j=1}^N X_j + J \sum_{j=1}^{N-1} X_jX_{j+1}$ and $B = h_y \sum_{j=1}^N Y_j$. Then the leading term of the Trotter error is proportional to the commutator $[A,B]=\sum_j E_j$ with each error term $E_j = 2ih_xh_y Z_j + 2iJh_y(Z_jX_{j+1} + X_jZ_{j+1})$ 
acting on two adjacent qubits, so the weight of $E_j^{\dag}E_{j'}$ is at most 4 (it can be less if the error terms overlap). Since $\|E\|_F = \mathcal O(\sqrt{N})$ \cite{Zhaorandom22} and $\|E\| = \mathcal O(N)$ \cite{childs2020theory}, the entanglement-based error bound is $\|(U_0-\mathscr{U}_{1})\ket{\psi}\|=\mathcal O(\delta t^{2}N \max_{j,j'}(\log d_{\rho_{j,j'}} - S(\rho_{j,j'}))^{1/4}) + \mathcal O(\delta t^{2}\sqrt{N})$. 

The bound of \cref{thm:main} depends on the entanglement of the subsystems $\support(E_j^{\dag}E_{j'})$. When the evolution satisfies the ETH, such as for the QIMF model with the typical parameters $(h_x,h_y, J) = (0.8090, 0.9045, 1)$ \cite{kim2014testing}, the subsystem entanglement entropy increases quickly and the entropies of the $4$-qubit subsystems rapidly approach $4$ (see \cref{fig:thermalH}(a) and (b)). Correspondingly, the Trotter error is only $\mathcal O(\delta t^{2}\sqrt{N})$, which is the scaling of the average-case Trotter error~\cite{Zhaorandom22}. This is a quadratic speedup with respect to the system size $N$, compared to the worst-case error $\mathcal O(\delta t^{2}N)$. We indeed observe this error reduction and convergence phenomenon empirically, as shown by the blue dotted curve in \cref{fig:thermalH}(a) and (b).
Moreover, we compare our theoretical entropy-based error bound on the Trotter error (assuming the distance is known) with the average- and worst-case error bounds. \Cref{fig:theory_error_comparison}(a) and (b) illustrate this comparison for both PF1 and PF2 for a typical QIMF  Hamiltonian with $(h_x, h_y, J) = (0.8, 0.9045, 1)$, for which the entanglement grows rapidly. For atypical cases in which the entanglement grows slowly, the entanglement-based bound does not provide improvement. 

A special property of highly entangled many-body states is $k$-uniformity, which can be used in the construction of quantum error-correcting codes~\cite{PhysRevA.104.032601} and quantum data masking protocols~\cite{scott2004multipartite}. A pure state of $N$ parties is called $k$-uniform if all $k$-party RDMs are maximally mixed. More generally, we say $\ket{\psi}$ is $\Delta$-approximately $k$-uniform if $\|{\tr_{[N]/[k]}(\ket{\psi}\bra{\psi})-\mathbb{I}/2^k}\|_1\le \Delta$ for any $k$-partite RDM. 
For a large system, typical states are close to (approximately) $k$-uniform (for small $k$) and the entanglement entropy of any $k$-local RDM is close to $k$ \cite{hayden2006aspects}. 
From \cref{thm:main}, by considering $k$-uniformity, we can obtain the following sufficient condition for the Trotter error to exhibit average-case scaling. 

\begin{corollary}\label{cor:kuniform}
For a $\Delta$-approximately $k$-uniform pure quantum state $\ket{\psi}$ with $\sqrt{\Delta}\le \|E\|_F/\sum_j\|E_j\|$ and $k\ge 2\max_j w(E_j)$, the Trotter error satisfies
\begin{equation}
\|{(U_0-\mathscr{U}_p)\ket{\psi}}\|= \mathcal O\left(\|E\|_F\delta t^{p+1}\right).
\end{equation} 
\end{corollary}

Note that for the PF1 simulation of the QIMF model, $k\ge 2\max_j w(E_j)=4$ and $\Delta=\mathcal O(1/N)$ is sufficient.
More generally, consider PF1 simulation of any $m$-local $N$-qubit Hamiltonian $H=\sum_{j_1,\dots,j_m}
H_{j_1,\dots,j_m}$ with each term $H_{j_1,\dots,j_m}$
acting nontrivially on (at most) $m$ qubits. If the input state is (approximately) $k$-uniform with $k\geq 4m-2$, then the simulation cost reduces to the average case. 
The parameter $k$, which quantifies how mixed the RDMs are, and thereby controls their entropy, typically increases linearly with the evolution time for a given error $\Delta$ \cite{Alessio2016quantum,gogolin2016equilibration}.  
Since $m$ is constant, the Hamiltonian evolution can make the underlying state (approximately) $k$-uniform in constant time. For comparison, for a classical simulation using MPS \cite{paeckel2019time}, the simulation cost increases exponentially with the entanglement, using $\mathcal O(2^{k/4} N)$ bits to store the state and $\mathcal O(2^{k/2} N)$ operations to simulate each step by contracting local tensors (see Section V of the Supplemental Material for more details). We illustrate this comparison in \cref{fig:compareQC}(a).

\begin{figure}
    \centering
     \includegraphics[width=1\columnwidth]{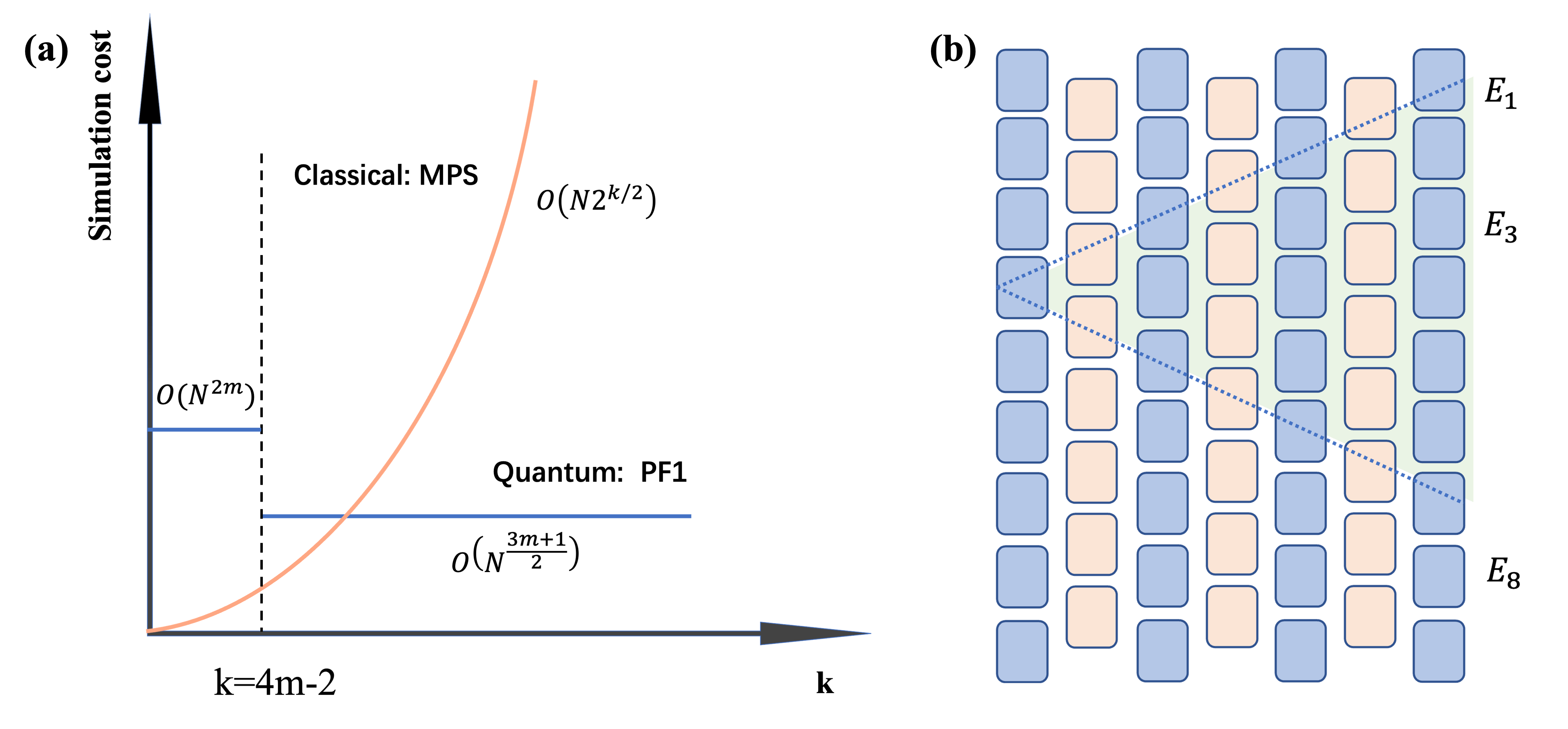}
    \caption{
    (a)
    Comparison of the predicted simulation cost for the quantum and classical algorithms for $k$-uniform input states. Here we consider simulating a general $m$-local lattice Hamiltonian for a short time $t=\mathcal O (1)$. Larger values of $k$ indicate more mixed RDMs, and therefore greater entanglement. The parameter $k$ generally increases linearly with the total evolution time. Here we consider PF1 simulation and quantify the cost using the gate count. For $k\geq 4m-2$, the cost decreases from worst-case to average-case scaling \cite{Zhaorandom22}, which is independent of $k$ thereafter. However, for classical methods based on MPS, the simulation cost is $\mathcal O(2^{k/2} N)$ to contract the local tensors (see Section V of the Supplemental Material for more details).
    (b) Illustration of the light-cone structure.}
 \label{fig:compareQC}
\end{figure}

\Cref{thm:main} and \cref{cor:kuniform} show that for states with sufficient entanglement, the Trotter error can scale similarly to the average case. On the other hand, simulations of unentangled states do not perform as well in general. In fact, one can find product states that achieve the worst-case error, as the following result demonstrates.

\begin{theorem}\label{thm:worstcase}
Consider a Trotter approximation of an $N$-qubit lattice Hamiltonian $H$. Let the leading term of the Trotter error have Pauli decomposition $E =\sum_j a_j P_j + \sum_k b_k Q_k$  with $a_j>0$ and $b_k<0$ (ignoring a global phase). If 
$\sum_j a_j=\Theta(N)$ and $\sum_k |b_k|=o(N)$, there exists a product state $\ket{\psi}$ that achieves the worst-case error scaling, 
\begin{equation}
 \|(U_0-\mathscr{U}_p)\ket{\psi}  \|=\Theta(N\delta t^{p+1}). 
\end{equation} 
\end{theorem}

The condition of this theorem is not overly restrictive and can be satisfied in natural examples, such as the well-studied QIMF model with the PF1 method.
A detailed construction for this case is presented in the Methods section, with further details in Section II of the Supplemental Material.

In addition to the product states identified by \cref{thm:worstcase}, states generated by non-ergodic Hamiltonian evolution may be insufficiently entangled to enjoy average-case Trotter error scaling.
For example, for the QIMF model in the atypical case $(h_x, h_y, J) = (0, 0.9045, 1)$, the Hamiltonian is integrable and can be transformed to free fermions via the Jordan-Wigner transformation. Starting from the initial state $\ket{0}^{\otimes N}$, the entanglement entropy does not reach its maximal value, which leads to a clear gap between the empirical Trotter error and the average-case error, as shown in \cref{fig:thermalH}(c) and (d).

\subsection*{Light-cone effect with geometrically local evolution}
\Cref{fig:thermalH}(a) shows that the Trotter error decreases before the subsystems are close to 4-uniform. Instead, the error converges to the average-case performance around the point where the RDMs of 1- and 2-qubit subsystems are highly mixed. This is because, starting from a product state, after evolving for a short time with a geometrically local Hamiltonian,
distant subsystems are only weakly correlated \cite{PhysRevLett.97.050401}, which simplifies the entanglement-based error bound in \cref{eq:ErrEnt}. To see this, consider the joint entropy
$S(\rho_{j,j'})=S(\rho_{j})+S(\rho_{j'})- I(j,j')_{\rho}$ where $I(j,j')_{\rho}$
is the mutual information for the distant subsystems on $\support(E_{j}^{\dag})$ and $\support(E_{j'})$. 
As shown in \cref{fig:compareQC}(b), only subsystems that lie inside the light cone (e.g., $E_1$ and $E_3$ in the figure) can be significantly correlated. For two subsystems outside the light cone (e.g., $E_1$ and $E_8$), their joint RDM $\rho_{j,j'}=\rho_{i}\otimes \rho_{j}$ is a product state and $I(j,j')_{\rho}=0$. More concretely, 
consider a pure state $\ket{\psi}$ generated by a $D$-dimensional geometrically local circuit of depth $C_{\mathrm{depth}}$ acting on a product state.
We can refine the entanglement-based error bound in \cref{thm:main} by replacing 
$\sum_{j,j'} \|E_j^{\dagger} E_{j'}\|\sqrt{\log d_{\rho_{j,j'}}-S(\rho_{j,j'})}$ with 
\begin{align}
    \sum_{(j,j')\in \mathcal{L}} \|E_j^{\dag}E_{j'}\|\sqrt{\log d_{\rho_{j,j'}}-S(\rho_{j,j'})}+ \sum_{(j,j')\notin \mathcal{L}} \|E_j^{\dag}E_{j'}\|\sqrt{\log d_{\rho_{j}}+ \log d_{\rho_{j'}} -S(\rho_{j})-S(\rho_{j'})},
    \label{eq:lightconeerror}
\end{align}
where $\mathcal{L}$ denotes the pairs of subsystems at distance less than $C_{\mathrm{depth}}$ (i.e., that are inside the light cone).
If the light cone size for each subsystem $\support(E_j)$ is much smaller than the system size $N$, then
for a fixed $E_{j'}$, there are only $o(N)$ terms $E_{j}^{\dag}$ satisfying $(j,j')\in \mathcal{L}$. In this way, the second term of \cref{eq:lightconeerror} dominates the Trotter error. We formalize this as follows (see Section III of the Supplemental Material for the proof).

\begin{corollary}
\label{cor:shallow}
    For an $N$-qubit lattice Hamiltonian and a pure quantum input state $\ket{\psi}$ that is generated by a $D$-dimensional geometrically local circuit of depth $C_{\mathrm{depth}}=o(N^{1/D})$, the entanglement-based bound of \cref{thm:main} is 
    \begin{align}
  \|{(U_0-\mathscr{U}_p)\ket{\psi}}\|
  = \mathcal O(\delta t^{p+1}N\max_{j}(\log d_{\rho_{j}}-S(\rho_{j}))^{1/4})+ O(\delta t^{p+1}\sqrt{N}). 
    \end{align}
\end{corollary}

In particular, if $\ket{\psi}$ is generated by a shallow and geometrically local circuit, and satisfies geometric $k$-uniformity 
with $k\ge \max_j w(E_j)$, the Trotter error $\|{(U_0-\mathscr{U}_p)\ket{\psi}}\|$ will have the average-case error upper bound. This is a looser requirement than the condition in \cref{cor:kuniform}, $k\ge 2\max_jw(E_j)$.
In the above QIMF example, for shallow circuit input states, entanglement entropy close to $2$ suffices for the Trotter error to have the average-case upper bound, which matches the empirical results in \cref{fig:thermalH}(a) very well.  

\subsection*{Total error of long-time evolution}
So far, we have discussed the error in a single Trotter step for a short time. For a long-time evolution, one can divide the whole evolution into $r$ segments, where in each step, the ideal evolution is approximated by the Trotter formula. Then 
the total simulation error can be upper bounded via the triangle inequality as
\begin{equation}
    \begin{aligned}
\|(U_0^r-\mathscr{U}_p^r)\ket{\psi(0)}\|  \le \sum_{i=0}^{r-1} \|U_0^{r-1-i}(U_0-\mathscr{U}_p)\mathscr{U}_p^i\ket{\psi(0)}\|
= \sum_{i=0}^{r-1} \|(U_0-\mathscr{U}_p)\ket{\phi_i}  \|,
\end{aligned}
\end{equation}
where $\ket{\phi_i}:=\mathscr{U}_p^i\ket{\psi(0)}$ denotes the evolved state under the Trotter approximation after $i$ Trotter steps. 
In this way, we can take the states $\ket{\phi_i}$ as the inputs for the above theorems and corollaries. For a long-time evolution, the entanglement typically increases \cite{Alessio2016quantum,gogolin2016equilibration}, so intuitively, one can use the worst-case analysis for the first $c$ Trotter steps, followed by the average-case analysis for the remaining $r-c$ steps, where $c$ is chosen so that the entanglement entropy after time $c \delta t$ satisfies 
$S(\rho_{j,j'})\ge w(E_j^{\dag}E_{j'})-\mc{O}\bigl(\bigl(\frac{\|E\|_F}{\sum_j\|E_j\|}\bigr)^4\bigr)$. 
Then the total error is approximately
$\mathcal O\left([c \|E\|+ (r-c) \|E\|_F] \delta t^{p+1}\right)$. When $c\lesssim r \|E\|_F/\|E\| $, the total error scales approximately as in the average case, with $\mathcal O\left(r\|E\|_F \delta t^{p+1}\right)= \mathcal O\left(\|E\|_F t^{p+1}/r^p\right)$ \cite{Zhaorandom22}. 

For a given order of product formula, the number of Trotter steps $r$ quantifies the cost of the simulation, since the total circuit complexity is proportional to $r$.
In \cref{fig:trottersteps}, we show the number of Trotter steps $r$ that suffices to ensure empirical error $10^{-5}$ for a long-time evolution of the QIMF model with time $t=N$. For the typical case, where the system becomes significantly entangled, the empirical performance is similar to the average-case performance. In contrast, for an atypical case that does not produce high entanglement, there is a clear gap between the empirical results and the average-case analysis.
The extrapolated scaling of our theoretical bounds aligns well with the theoretical average-case error bound, as depicted in \cref{fig:theory_error_comparison}(c).

\begin{figure}
    \centering
    \includegraphics[width=1\columnwidth]{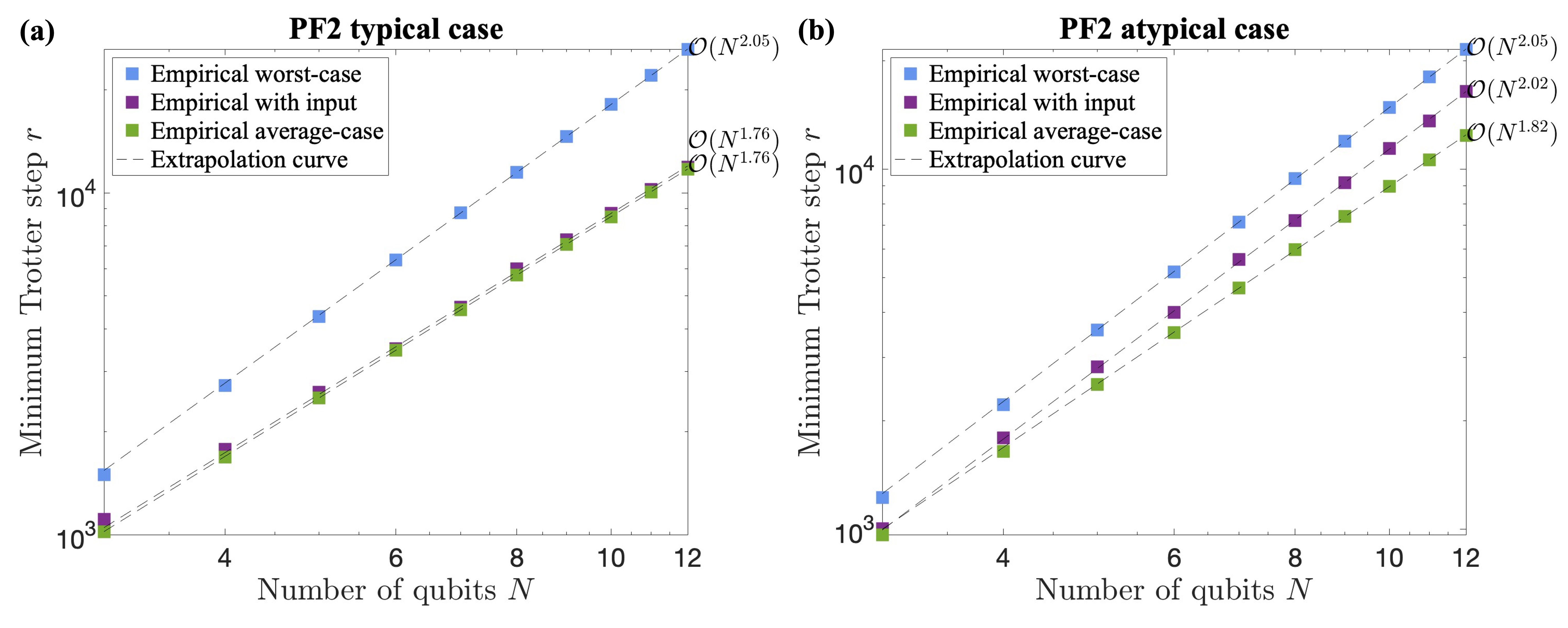}
    \caption{Number of Trotter steps determined with different error bounds for PF2 with $t=N$ and $\varepsilon=10^{-5}$. (a) and (b) show the typical case $(h_x,h_y,J)= (0.8090, 0.9045, 1)$ and the atypical case $(h_x,h_y,J)= (0, 0.9045, 1)$, respectively.
     The empirical spectral norm curve corresponds to the worst-case analysis such that 
     $\|U_0(t)-\mathscr{U}_p^r(t/r)\| \le \varepsilon$. The empirical random input curve corresponds to the average-case analysis such that 
     $\mathbb{E}_{\mathcal{E}}(\|U_0(t)-\mathscr{U}_p^r(t/r)\ket{\psi}\|)\le \varepsilon$, which assumes that $\mathcal{E}=\{\ket{\psi}\}$ is a 1-design ensemble \cite{Zhaorandom22}. The empirical curve with the input $\ket{\psi(0)}=\ket{0}^{\otimes N}$ corresponds to the minimal $r$ such that 
     $\|U_0(t)-\mathscr{U}_p^r(t/r)\ket{\psi(0)}\| \le \varepsilon$.}
    \label{fig:trottersteps}
\end{figure}

\subsection*{Measurement-assisted adaptive Hamiltonian simulation}
To apply the entanglement-based error bound of \cref{thm:main}, we must understand how entangled the state is.
If the entanglement is not known in advance, we can estimate it during the simulation, and thereby adaptively enhance the algorithm's performance.
Let $\ket{\phi_i}=\mathscr{U}_{p}^i\ket{\psi(0)}$ denote the state of the algorithm after $i$ Trotter steps. Using shadow tomography~\cite{aaronson2019shadow,huang2020predicting}, we can estimate the local RDMs of $\ket{\phi_i}$ and their entanglement entropy, giving an estimate of the Trotter error bound in \cref{thm:main}. Let $M$ denote the number of copies of the state used in the estimation. By careful choice of measurements, we can show that $M=\mc O(N^2)$ suffices, as discussed in \cref{Merror:pauli} below.
We describe the measurement process in more detail in the Methods section.

Here we demonstrate the advantage of the adaptive simulation framework compared to a simulation with worst-case error analysis. Consider the standard quantum simulation task: prepare the quantum state obtained by evolving with a given Hamiltonian $H$ for time $t$ within a small error $\epsilon=\mc{O}(1)$. 
Suppose at some checkpoint time $t_c$, we insert the measurement gadget to estimate the Trotter error for the upcoming simulation. 
Then the total gate cost of the simulation is
\begin{align}\label{eq:gatenum}
   \mc{G}= G_1(t_c)M+[G_1(t_c)+G_2(t-t_c)]M_o,
\end{align}
compared to $\mc{G}'= G_1(t)M_o$ without error estimation.
Here $G_{1(2)}$ labels the gate count for simulation of the time before (after) $t_c$, respectively. The parameters $M$ and $M_o$ denote the number of copies of the states used at the checkpoint time $t_c$ and at the final time $t$, respectively.
For a simple final observable, such as a single Pauli operator, $M_o=\mc{O}(1)$;
but for a general measurement $O=\sum O_i$ (for instance, a Hamiltonian with $\mc O (N)$ local terms and estimation accuracy $\epsilon_s=\mc O(1)$ \cite{zhou2023performance}), the final measurement cost could be $M_o=\mc{O}(N^2)$, comparable with $M$. Note that here we assume that the error estimation at checkpoint $t_c$ controls the error for the entire remaining time period of duration $t-t_c$. By \cref{thm:main}, this is implied by the assumption that the local entanglement entropy is non-decreasing, which is the typical case even for MBL systems \cite{Alessio2016quantum,nandkishore2015many}.

Here we use the first-order product formula algorithm to illustrate the underlying advantage. Suppose at time $t_c$, we check that the error for subsequent simulation is close to the average case (that is, the RDMs are almost maximally mixed) \cite{Zhaorandom22}. In this case, the gate counts are $G_1(t,\epsilon)=\mc{O}(N^2t^2\epsilon^{-1})$ and $G_2(t,\epsilon)=\mc{O} (N^{1.5}t^2\epsilon^{-1})$ for the worst and average cases, respectively. 
Here $t_c=\mc{O}(1)$ is typically independent of the system size. Denoting the simulation error before $t_c$ by $\epsilon_c$, 
and assuming the error is proportional to the simulation time (i.e., $\epsilon_c=\epsilon t_c/t$) and $M=M_o=\mc O (N^2)$, we find that $\mc{G}=\mc{O}(N^4t+N^{3.5}t^2)$. This is better than the gate count $\mc{G'}=\mc O(N^{4}t^2)$ that would be obtained by worst-case analysis, without using measurement to improve the Trotter error estimation.

Taking this idea further, we can develop an adaptive Trotter algorithm by inserting a few measurements in the middle of the simulation, as summarized in \cref{al:adaptive}. 

\begin{algorithm}[H]
\caption{Measurement-assisted adaptive Trotter simulation}\label{al:adaptive}
\begin{algorithmic}[1]
\Require
Initial quantum state $\ket{\psi(0)}$, simulation time $t$, simulation error $\epsilon$, Hamiltonian $H$ with the $p$th-order Trotter approximation, and $T$ measurement checkpoints $t_1,t_2,\ldots,t_T$ with $t_0=0\le t_1<t_2\cdots<t_T<t_{T+1}=t$.
\Ensure
Final quantum state $\ket{\phi(t)}$ with simulation error at most $\epsilon$ to $e^{-iHt}\ket{\psi(0)}$.
\For{$i= 1~\text{\textbf{to}}~T$} 
 \For{$j= 1~\text{\textbf{to}}~M$} 
 \State Prepare the state $\ket{\phi(t_{i})}$ from $\ket{\psi{(0)}}$ using Trotter simulation with the previously determined Trotter steps $\{r(\Delta_{i'})\}$ from $i'=0$ to $i-1$.  
  \State  Perform random $N$-qubit Pauli measurements on $\phi(t_{i})$ to collect a one-shot shadow snapshot $\hat{\rho}_{(j)}$.
 \EndFor
 \State Estimate the single-segment Trotter error $\epsilon(i)$ in \cref{Merror:pauli} using $\left\{\hat{\rho}_{(j)}\right\}_{j=1}^M$, with $M=\mc O(N^2)$ for a lattice Hamiltonian.
 \State Use $\epsilon(i)$ to determine the Trotter step $r(\Delta_i)$ for the next interval of duration $\Delta_i=t_{i+1}-t_{i}$, such that the simulation error is at most $\epsilon(\Delta_i)=\epsilon\Delta_i/t$.
\EndFor 
\State Prepare the final state $\ket{\phi(t)}$ from $\ket{\psi(0)}$ using the Trotter algorithm with the previously determined Trotter steps $\{r(\Delta_{i})\}$ from $i=0$ to $T$.
\end{algorithmic}
\end{algorithm}

The total gate count, generalizing \cref{eq:gatenum}, is
\begin{align}\label{eq:gatenum1}
   \mc{G}
   =\sum_{c=1}^{T} \sum_{i=0}^c G_i(\Delta_i,\epsilon(\Delta_i)) M,
\end{align}
where we take $M_o=M$ and let $G_i$ denote the gate count for the $i$th period, which depends on the error estimation at checkpoint $t_i$ in \cref{al:adaptive}.

We make several remarks. First, in \cref{al:adaptive}, we assume that the Trotter error in the period $\Delta_i$ is controlled by the state at the checkpoint $t_i$. 
This assumption is reasonable since the local entanglement entropy tends to increase \cite{Alessio2016quantum,nandkishore2015many}.
Second, the number of measurement points and their locations in time can be adjusted in real time. For example, we can keep track of the simulation error to see whether it reaches nearly the average case, indicating that the local RDMs have already thermalized. Once this occurs, further error estimation is unnecessary under the assumption that the entanglement increases monotonically, though we could choose to verify this by checking the entanglement after long periods of evolution.
Third, the summation in \cref{eq:gatenum1} gives gate overhead due to the measurement.
However, if the quantum simulation task is to track the expectation value of some observable $\langle O(t) \rangle$ as a function of $t$ for several time points, the measurement-assisted simulation may not need to introduce additional measurements. In particular, for a sufficiently complex operator $O$, such as a Hamiltonian with $\mc O (N)$ local terms, one can reuse the randomized measurement data for estimating $O$ to estimate the Trotter error at the checkpoints of the simulation.

In \cref{fig:theory_error_comparison}(d), we show the number of Trotter steps $r$ used by the adaptive PF2 protocol with various numbers of uniformly spaced checkpoints. The adaptive protocol reduces the number $r$ and thus enhances the performance of the simulation with only a few checkpoints. See the Methods section and Section VII of the Supplemental Material for more details.

\begin{figure}
    \centering
    \includegraphics[width=1\columnwidth]{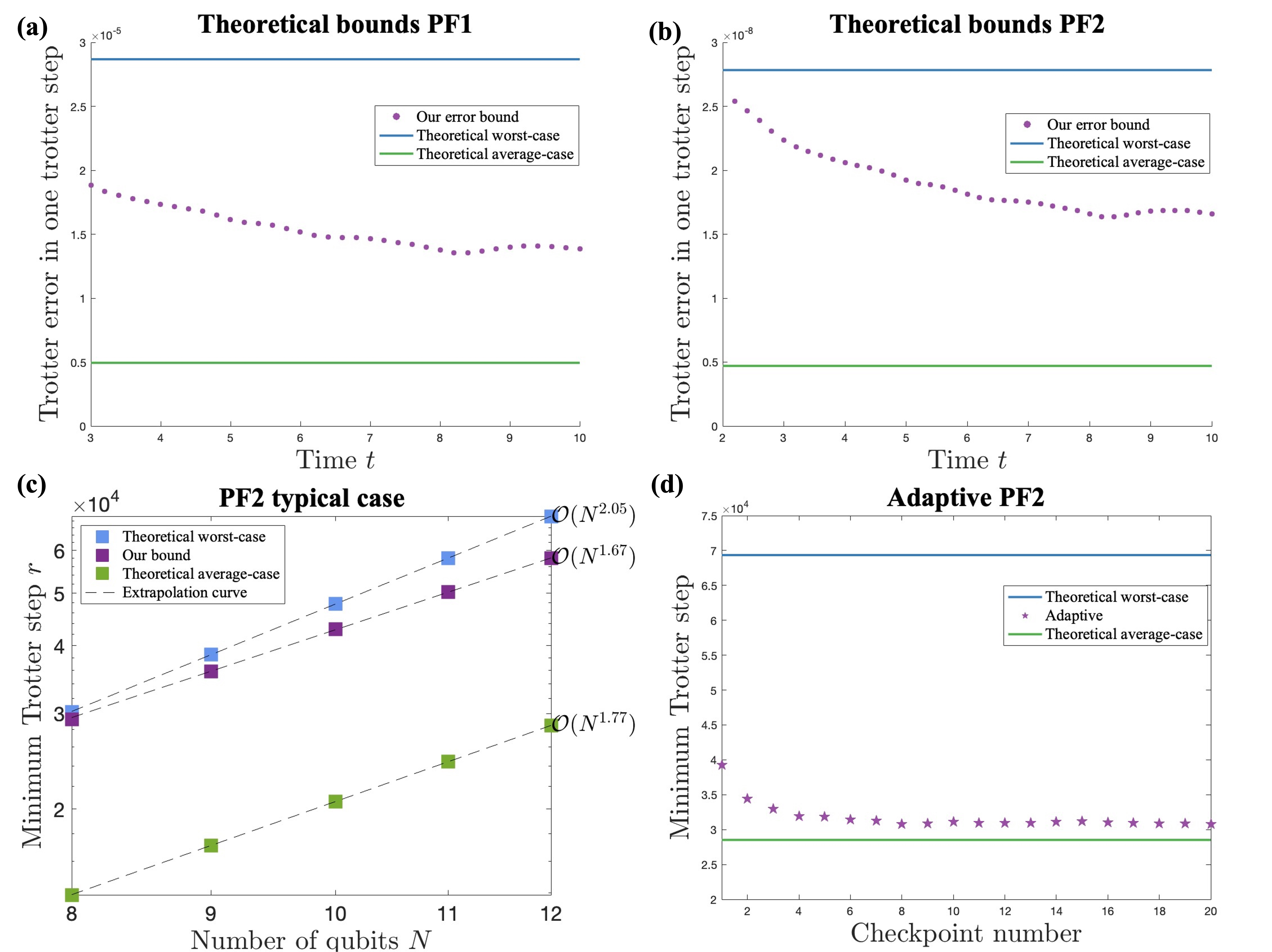}
    \caption{(a--b). Theoretical Trotter error in each Trotter step for PF1 and PF2. We compare the theoretical worst-case Trotter error (purple), distance-based Trotter error (blue), and the theoretical average-case Trotter error (green). See \cref{fig:thermalH} for a comparison to the empirical error.
   (c). Comparison of the minimum required Trotter steps $r$ for different theoretical bounds. For our bound, we use the result from \cref{Lemma:ABPF2}, the explicit formulas for theoretical worst-case and average-case bounds are shown in Section VII of the Supplemental Material. 
    (d). Minimum required Trotter steps $r$ with adaptive protocols. For the purple dots, we apply the adaptive protocol in Algorithm 1 with various numbers of checkpoints by measuring the operators in \cref{Merror:pauli}. }
 \label{fig:theory_error_comparison}
\end{figure}

\section*{Discussion and Outlook}

In this work, we showed that Hamiltonian simulation using product formulas can be more effective when the simulated state is entangled. Our numerical results suggest that the number of Trotter steps to ensure an accurate simulation scales similarly for sufficiently entangled states and for the average case, both empirically (\cref{fig:trottersteps}) and for the theoretical bounds (\cref{fig:theory_error_comparison}).

The acceleration arising from entanglement can be directly applied to other 
algorithms based on quantum simulation, such as quantum adiabatic algorithms \cite{kovalskySelfhealingTrotterError2023} and quantum phase estimation \cite{yi2022spectral}. It is also worth 
exploring whether a similar phenomenon can improve 
product formula-based quantum Monte Carlo algorithms \cite{bravyi2015monte,bravyi2017poly}
and algorithms based on imaginary time evolution \cite{motta2020determining}.

In our numerical study of the adaptive quantum simulation algorithm, we chose uniformly spaced checkpoints. However, an algorithm incorporating some other pattern of checkpoints might yield greater efficiency. In particular, under the assumption that the entanglement tends to increase, we need not estimate the entanglement after it is already determined to be nearly maximal.
Future work might try to make this approach more precise and explore the best way to adaptively estimate the entanglement.

Our results advance the understanding and application of quantum entanglement in quantum information processing, and show how entanglement may prove advantageous for investigating fundamental quantum physics through accelerated quantum simulations.
This work sheds new light on the intricate relationship between quantum resources and quantum algorithms, revealing that quantum resources can serve as a driving force for accelerating quantum simulation, and potentially for other quantum algorithms. Future work might explore whether other quantum resources, such as magic \cite{howard2014contextuality} and coherence \cite{RevModPhys.89.041003}, can also help improve the performance of quantum simulation or other quantum algorithms. It would also be beneficial to perform a systematic study of how our results can reduce the resource requirements for specific practical simulations in areas such as chemistry, condensed matter, and nuclear/particle physics.

\section*{Acknowledgements}
We are grateful to Wenjun Yu, Fei Shi, and Jue Xu for the helpful discussions.
Q.Z.\ acknowledges funding from the HKU Seed Fund for Basic Research for New Staff via Project 2201100596, the Guangdong Natural Science Fund via Project 2023A1515012185, the National Natural Science Foundation of China (NSFC) via Project Nos.\ 12305030 and 12347104, Hong Kong Research Grant Council (RGC) via Project No.\ 27300823, N\_HKU718/23, and R6010-23, Guangdong Provincial Quantum Science Strategic Initiative GDZX2200001.
Y.Z.\ acknowledges the support of the National Natural Science Foundation of China (NSFC) Grant No.\ 12205048, the Innovation Program for Quantum Science and Technology 2021ZD0302000, and the start-up funding of Fudan University.
A.M.C.\ acknowledges support from the United States Department of Energy, Office of Science, Office of Advanced Scientific Computing Research, Accelerated Research in Quantum Computing program (award no.\ DE-SC0020312), and from the National Science Foundation (QLCI grant OMA-2120757).

\section*{Competing financial interests}
The authors declare no competing financial interests.

\section*{Methods}

\subsection*{Proof sketch of \texorpdfstring{\cref{thm:main}}{Theorem 1}}
For a given Hamiltonian $\sum_{l=1}^LH_l$, the first-order product formula
applies the unitary operation $\mathscr{U}_1(t):=e^{-iH_1t}e^{-iH_2t}\cdots e^{-iH_Lt}
=\overrightarrow{\prod_l}e^{-iH_lt}.$
Here the right arrow indicates the product is in the order of increasing indices.
Suzuki's high-order product formulas are defined recursively as
\begin{align}\label{M_Eq:highorder}
\mathscr{U}_2(t)&:=\overrightarrow{\prod_l}e^{-iH_lt/2} \overleftarrow{\prod_l}e^{-iH_lt/2},    \\
\mathscr{U}_{2k}(t)&:= [\mathscr{U}_{2k-2}(p_k t)]^2 \mathscr{U}_{2k-2}((1-4p_k)t) [\mathscr{U}_{2k-2}(p_k t)]^2, \nonumber
\end{align}
where $\overleftarrow{\prod}_l$ denotes a product in decreasing order and $p_k:=\frac{1}{4-4^{1/(2k-1)}}$ with $k>1$ \cite{suzuki1991general}.

For a given input state $\ket{\psi}$,
the error of PF$p$ satisfies $\|(U_0(\delta t)-\mathscr{U}_p(\delta t))\ket{\psi}\|\le  \| \sum_{i} E_j \ket{\psi}\| \delta t^{p+1}  + \|\mathscr{E}_{\mathrm{re}}\|$ with $\|\mathscr{E}_{\mathrm{re}}\|=\mathcal O(\delta t^{p+2})$~\cite{childs2020theory}.
Here we focus on the leading term $\| \sum_{j} E_j \ket{\psi}\|=  \sqrt{\sum_{j,j'} \bra{\psi}E_j^{\dagger} E_{j'}\ket{\psi} } $.

\begin{lemma}\label{Lemma:local}
Let $E=\sum_j E_j$ act on $N$ qubits, where $E_j$ acts nontrivially on the subsystem with $\support(E_j)$. Then
\begin{align}
  | \bra{\psi}E^{\dagger}E \ket{\psi}|\le \|E\|_F^2+    \sum_{j,j'} \|E_j^{\dag} E_j'\| ~\tr|\rho_{j,j'}- \mathbb{I}_{\support(E_j^{\dag}E_{j'})}/d_{\support(E_j^{\dag}E_{j'})} |,
\end{align}
where $\|E\|^2_F:= \tr(E^{\dagger}E) /d$ is the square of the normalized Frobenius norm, and $\rho_{j,j'}:=\tr_{[N]\setminus \support(E_j E_j')}(\ket{\psi}\bra{\psi})$ is the RDM of $\ket{\psi}\bra{\psi}$ on the subsystem on $\support(E_j^{\dag}E_j')$.
\end{lemma}

\begin{proof}
The term $E_j^{\dag}E_{j'}$ in the expression for $E^{\dag}E$ only acts nontrivially on $\support(E_j^{\dag}E_{j'})$. We denote its nontrivial part by $L_{j,j'}:=\tr_{[N]\setminus \support(E_j E_j')} (E_j^{\dag}E_{j'})$. Since $2^{-N}\tr(E_j^{\dagger} E_{j'})=d_{\support(E_j^{\dag}E_{j'})}^{-1}\tr(L_{j,j'})$, we have
\begin{equation}
\begin{aligned}
\bra{\psi}E_j^{\dagger}E_{j'} \ket{\psi}&=\tr(L_{j,j'}\rho_{j,j'})=\tr[L_{j,j'}(\rho_{j,j'}- \mathbb{I}_{\support(E_j^{\dag}E_{j'})}/d_{\support(E_j^{\dag}E_{j'})})]+\tr(L_{j,j'})/d_{\support(E_j^{\dag}E_{j'})}\\
&= \tr[L_{j,j'}(\rho_{j,j'}- \mathbb{I}_{\support(E_j^{\dag}E_{j'})}/d_{\support(E_j^{\dag}E_{j'})})]+\tr(E_j^{\dagger}E_{j'})/2^N\\
&\leq \|L_{j,j'}\|\tr|\rho_{j,j'}- \mathbb{I}_{\support(E_j^{\dag}E_{j'})}/d_{\support(E_j^{\dag}E_{j'})}|+\tr(E_j^{\dagger}E_{j'})/2^N\\
&= \|E_j^{\dag}E_{j'}\| \tr|\rho_{j,j'}- \mathbb{I}_{\support(E_j^{\dag}E_{j'})}/d_{\support(E_j^{\dag}E_{j'})}|+\tr(E_j^{\dagger}E_{j'})/2^N,
\end{aligned}
\end{equation}
and the result follows by summing the indices $j,j'$.
\end{proof}

\cref{Lemma:local} implies the distance-based error bound of \cref{thm:main}. Moreover, the trace distance of $\rho_{j,j'}$ and $\mathbb{I}_{\support(E_j^{\dag}E_{j'})}/d_{\support(E_j^{\dag}E_{j'})}$ can be bounded by the relative entropy as
\begin{align}
    \tr|\rho_{j,j'}- \mathbb{I}_{\support(E_j^{\dag}E_{j'})}/d_{\support(E_j^{\dag}E_{j'})}| \le \sqrt{2S(\rho_{j,j'}\|\mathbb{I}_{\support(E_j^{\dag}E_{j'})}/d_{\support(E_j^{\dag}E_{j'})}) }=
    \sqrt{2\log(d_{\support(E_j^{\dag}E_{j'})})-2S(\rho_{j,j'})},
\end{align}
which leads to the entanglement-based bound.
For the proof, see Section I of the Supplemental Material.

\subsection*{Product state with the worst-case error in \texorpdfstring{\cref{thm:worstcase}}{Theorem 2}}

Here we show an example of a product state achieving the worst-case scaling error.
We consider the QIMF model introduced in the main text with the PF1 method. We find
\begin{equation}
\begin{aligned}
 E= [iA, iB]&= -i \left[2h_xh_y  \sum_{j=1}^N Z_j  +2 Jh_y \sum_{j=1}^{N-1} (Z_jX_{j+1}+ X_jZ_{j+1})\right].\\
\end{aligned}\label{eq:pf1commutator}
\end{equation}
After ignoring the global phase $-i$, we have $\sum_j a_j=\Theta (N)$ and  $\sum_k |b_k|=0$ which satisfies the condition in \cref{thm:worstcase}, since there is no term with a negative coefficient. Thus we can
choose $\ket{\psi}=\ket{0}^{\otimes N}$.
Using the triangle inequality, we can lower bound the overall error as
\begin{equation}
\begin{aligned}
    \|(U_0-\mathscr{U}_{1})\ket{\psi}\|&\ge \sqrt{\tr (E^{\dagger} E\ket{\psi} \bra{\psi}) } -\|\mathscr{E}_{\mathrm{re}}\| =  \Theta (N) \delta t^{p+1} -\|\mathscr{E}_{\mathrm{re}}\|
    =\Omega(N) \delta t^{p+1},
\end{aligned}
\end{equation}
since $\|\mathscr{E}_{\mathrm{re}} \|=\mathcal O (N\delta t^{p+2})$.
Combining this with the worst-case upper bound $\|(U_0-\mathscr{U}_1)\ket{\psi}\|=\mathcal O(N\delta t^{p+1})$,
we conclude that $\|(U_0-\mathscr{U}_{1})\ket{\psi}\|=\Theta(N\delta t^{2})$.
For the general proof and more examples, see Section II of the Supplemental Material.

\subsection*{Concrete error bounds for PF1 and PF2}
Here we also describe the error bounds with concrete prefactors for a two-term Hamiltonian $H=A+B$ with the first- and second-order product formulas (PF1 and PF2).

\begin{lemma}\label{Lemma:AB}
(First-order product formula, PF1)
For a two-term Hamiltonian $H=A+B$, consider the first-order product formula $\mathscr{U}_1(\delta t)=e^{-iAt}e^{-iBt}$ with initial state $\ket{\psi}$. Let $E:=[A,B]=\sum_j E_j$. Then the Trotter error is upper bounded as
\begin{equation}\label{Le2eq}
    \begin{aligned}
\|(\mathscr{U}_1(\delta t)-U_0(\delta t))\ket{\psi}\|&\le \sqrt{\frac{\delta t^4}{4}(\|[A,B]\|_F^2+\Delta_E(\psi)}+ \frac{\delta t^3}{6}\|[A,[A,B]]\|+  \frac{\delta t^3}{3}\|[B,[B,A]]\|,\\
with\ \Delta_E(\psi)&=\sum_{j,j'}\|E_{j'}E_j\| \tr|\rho_{j,j'}-\mathbb{I}_{\support(E_j^{\dag}E_{j'})}/d_{\support(E_j^{\dag}E_{j'})}|.
\end{aligned}
\end{equation}
For sufficiently small $t$, the error is $\mathcal O(\delta t^2\sqrt{\|[A,B]\|_F^2+\Delta_E(\psi}))$.
\end{lemma}

\begin{lemma}\label{Lemma:ABPF2}
(Second-order product formula, PF2)
For a two-term Hamiltonian $H=A+B$, consider the second-order product formula $\mathscr{U}_2( \delta t)=e^{-iA/2 \delta t}e^{-iB \delta t}e^{-iA/2\delta t}$ with initial state $\ket{\psi}$. Let $E_1:=[B,[B,A]]=\sum_j E_{1,j}$ and  $E_2:=[A,[A,B]]=\sum_j E_{2,j}$. Then the Trotter error is upper bounded as
\begin{equation}\label{le3eq}
    \begin{aligned}
\|(\mathscr{U}_2(\delta t)-&U_0(\delta t))\ket{\psi}\|\le \sqrt{\frac{\delta t^6}{144} \|\left[B,\left[B, A\right]\right] \|_F^2 +\Delta_{E_1}(\psi)}
+ \sqrt{\frac{\delta t^6}{576} \|\left[A,\left[A, B\right]\right] \|_F^2 +\Delta_{E_2}(\psi)}
\\
&+ \frac{\delta t^4}{32}\|\left[A,\left[B,\left[B,A\right]\right]\right]\|+  \frac{\delta t^4}{12}\|\left[B,\left[B,\left[B,A\right]\right]\right]\| +\frac{\delta t^4}{32}\|\left[B,\left[A,\left[A,B\right]\right]\right]\|+  \frac{\delta t^4}{48}\|\left[A,\left[A,\left[A,B\right]\right]\right]\|,\\
with\ \Delta_{E_1}(\psi)&=\sum_{j,j'}\|E_{1,j}^{\dag}E_{1,j'}\| \tr|\rho_{j,j'}-\mathbb{I}_{\support(E_{1,j}E_{1,j'})}/d_{\support(E_{1,j}^{\dag}E_{1,j'})}|, \\
\Delta_{E_2}(\psi)&=\sum_{j,j'}\|E_{2,j}^{\dag}E_{2,j'}\| \tr|\rho_{j,j'}-\mathbb{I}_{\support(E_{2,j}E_{2,j'})}/d_{\support(E_{2,j}^{\dag}E_{2,j'})}|.
\end{aligned}
\end{equation}
For sufficiently small $\delta t$, the error is $\mathcal O\left(\delta t^3\left(\sqrt{\|\left[B,\left[B, A\right]\right] \|_F^2 +\Delta_{E_1}(\ket{\psi})}+ \sqrt{\|\left[A,\left[A, B\right]\right] \|_F^2 +\Delta_{E_2}(\ket{\psi}) }    \right)\right)
$.
\end{lemma}

Proofs appear in Section IV of the Supplemental Material.

\subsection*{Numerical details}
In the numerics, we consider the 1D quantum Ising spin system with mixed fields (QIMF) \cite{cotler2021emergent},
\begin{align}
  H= h_x \sum_{j=1}^N X_j + h_y \sum_{j=1}^N Y_j + J \sum_{j=1}^{N-1} X_jX_{j+1},
\end{align}
and take the initial state $\ket{0}^{\otimes n}$. For $h_x=0$, the Hamiltonian can be transformed to an integrable model by the Jordan-Wigner transformation, leading to non-thermalizing dynamics. On the other hand, if $h_x\neq 0$, the Hamiltonian is
expected to be consistent with the ETH prediction \cite{kim2014testing}. Here we consider the parameters $(h_x,h_y,J)= (0.8090, 0.9045, 1)$ (which have been explicitly shown to satisfy the ETH~\cite{kim2014testing}) and $(h_x,h_y,J)= (0, 0.9045, 1)$ as typical and atypical examples, respectively.
For the typical example, the local density matrix thermalizes to the maximally mixed state.
For the product formula decomposition, we take
\begin{align}
  A= h_x \sum_{j=1}^N X_j  + J \sum_{j=1}^{N-1} X_jX_{j+1},~B= h_y \sum_{j=1}^N Y_j.
\end{align}
The PF1 commutator is shown in \cref{eq:pf1commutator}.
The corresponding PF2 commutators are
\begin{equation}
\begin{aligned} \relax
  [A, [A, B]]&= 4h_x^2h_y    \sum_{j=1}^N Y_j+
  4 J^2 h_y \sum_{j=1}^{N-1} Y_j+
   4 J^2 h_y \sum_{j=2}^{N} Y_j+
  8 Jh_xh_y \sum_{j=1}^{N-1} (Y_jX_{j+1}+ X_jY_{j+1}) \\
  &\quad + 8J^2 h_y \sum_{j=1}^{N-2} (X_jY_{j+1}X_{j+2}),\\
  [B, [A, B]]&= -4 h_xh_y^2  \sum_{j=1}^N X_j  + 8Jh_y^2 \sum_{j=1}^{N-1} ( Z_jZ_{j+1}-X_jX_{j+1}).
\end{aligned} \label{eq:commutators}
\end{equation}

The empirical worst-case and average-case Trotter error curves in \cref{fig:thermalH}(a) and \cref{fig:trottersteps} are calculated using the upper bounds from Ref.~\cite{childs2020theory} by directly computing the norms
$\|\mathscr{U}_2(\delta t)-U_0(\delta t)\|$ and $\|\mathscr{U}_2(\delta t)-U_0(\delta t)\|_F$, respectively.
For concreteness, we compare the performance in various cases with $N=12$ qubits.
According to the previous worst-case error analysis, the simulation should use $r\approx 2.62\times 10^4$ Trotter steps for the typical case shown in \cref{fig:trottersteps}(a) and $r\approx 2.15\times 10^4$ steps for the atypical case shown in \cref{fig:trottersteps}(b). Notably, the typical Hamiltonian appears harder to simulate than the atypical one. However, when considering the entanglement of the input state (as estimated empirically), the typical case is actually easier to simulate than the atypical case.
In the typical case, the empirical number of Trotter steps $r\approx 1.19\times 10^4$ is very close to the average-case value $r\approx 1.17\times 10^4$, because the evolved state is highly entangled.
On the other hand, in the atypical case shown in \Cref{fig:trottersteps}(b),
the empirical Trotter number $r\approx 1.65\times 10^4$ is larger than the average-case Trotter number
$r\approx 1.24\times 10^4$.

For the theoretical worst- and average-case bounds in \cref{fig:trottersteps}, we use \cref{eq:commutators} without evaluating their norms numerically. Instead, we upper bound the norms by counting the number of Pauli operators, e.g., $\|[A,B]\|\le 2h_xh_yN+ 4h_y(N-1)$. The concrete theoretical bounds are provided in the Supplemental Material.

We choose $\delta t=10^{-3}$ and apply \cref{Lemma:AB,Lemma:ABPF2} for our entanglement-based theoretical bounds in \cref{fig:theory_error_comparison}(a) and (b). To upper bound both spectral and Frobenius norms of operators, we count Pauli strings. To estimate $\Delta_E$ in \cref{Le2eq} and \cref{le3eq}, we directly compute the trace distances between local density matrices and identity operators. We do not faithfully estimate $\Delta_E$ and calculate the Trotter error for each step in our distance-based bounds. Since the entanglement remains stable for some time, we use the estimated $\Delta_E$ for a few subsequent Trotter steps, reducing the running time, especially for large systems. Details are given in Section VII of the Supplemental Material.

There is a significant constant-factor gap between our entanglement-based theoretical error bounds and average-case error bounds, as illustrated in \cref{fig:theory_error_comparison}.
This discrepancy presumably arises, at least in part, because our numerical simulations only involve quantum systems with at most $N = 12$ qubits and our theory mainly deals with the leading-order Trotter error terms $E_j$. In principle, the theoretical analysis could be extended to include the effect of higher-order error terms.
However, such terms have larger weights and the analysis will work only when there is entanglement in larger subsystems.
Consequently, we anticipate that the gap between our entanglement-based theoretical error and the average-case error will tend to decrease when simulating larger systems.

For the adaptive protocol in \cref{fig:theory_error_comparison}(d), we apply the concrete error bound in \cref{le3eq} from \cref{Lemma:ABPF2}. According to \cref{Merror:pauli}, we can measure the leading part of the Trotter error directly in experiments.
Thus, we replace the terms $\|\left[B,\left[B, A\right]\right] \|_F^2 +\Delta_{E_1}(\phi_i)$ and $ \|\left[A,\left[A, B\right]\right] \|_F^2 +\Delta_{E_2}(\phi_i)$ with measured values of $\sum_{j, j'} \bra{\phi_i}E_{1,j}^{\dag}  E_{1, j'}\ket{\phi_i}$ and
$\sum_{j, j'} \bra{\phi_i}E_{2,j}^{\dag}  E_{2, j'}\ket{\phi_i}$.
We consider the Pauli decompositions of sub-leading contributions to \cref{le3eq} and use the 1-norms of their coefficients to upper bound their spectral norms.
The total numbers of Trotter steps for the adaptive approach with various checkpoints are obtained using this analysis.

\subsection*{Estimating Trotter error with shadow tomography}
Here we develop a measurement gadget that uses shadow tomography \cite{aaronson2019shadow,huang2020predicting} to estimate the entanglement entropy and the local RDMs. This technique can enable simulation that takes advantage of the error bounds in \cref{thm:main} even if the entanglement is not known in advance. The approximate state after $i$ Trotter steps is $\ket{\phi_i}=\mathscr{U}_{p}^i\ket{\psi(0)}$, where $\mathscr{U}_{p}$ is the unitary operation implemented by a single Trotter step. Using the bound $\tr |M|\le d_M\|M\|_F$, the distance-based error bound in \cref{eq:ErrD1} of \cref{thm:main} for the error after the $(i+1)$st segment can be further bounded by
\begin{equation}\label{eq:ErrD2}
\epsilon(i):=\|(U_0-\mathscr{U}_{p})\ket{\phi_i}\|= \mathcal O\left(\left[\max_{j}\|E_j\| \left(\sum_{j,j'} \sqrt{d_{\support(E_j^{\dag}E_{j'})}\tr(\rho_{j,j'}^2)-1}\right)^{1/2}+ \|E\|_F \right]\delta t^{p+1}\right).
\end{equation}
Similarly, the entanglement-based error bound in \cref{eq:ErrEnt} of \cref{thm:main} can also be related to the purity of RDMs, by the fact that the von Neumann entropy can be upper bounded by the R{\'e}nyi-2 entropy, $S_2(\rho_{j,j'})=-\log_2\tr(\rho_{j,j'}^2)$. We provide an explicit formula in Section VI of the Supplemental Material.
As a result, for both error bounds, we would like to estimate the purities $\tr(\rho_{j,j'}^2)$ for all possible RDMs whose supports are determined by the related terms in the commutator, namely $\support(E_j^{\dag}E_{j'})$.

In shadow estimation, one performs randomized measurements on $\rho=\ket{\phi_i}\bra{\phi_i}$ by applying (local) random unitary evolution $U=\bigotimes_{j=1}^N u_j$ and projective measurement in the computational basis. Here the local unitary is sampled independently and uniformly from the local Clifford group. This random measurement scheme is called Pauli measurement, since it is equivalent to measuring each qubit in a randomly selected Pauli $X$, $Y$, or $Z$ basis \cite{huang2020predicting}.
The measurement result is denoted by $\mathbf{b}=(b_1,b_2,\dots, b_N)$.
The shadow snapshot
\begin{equation}
    \begin{aligned}
\hat{\rho}=\bigotimes_{j=1}^N(3u_j^{\dag}\ket{b_j}\bra{b_j}u_j-\id_2),
\end{aligned}
\end{equation}
is an unbiased estimator, i.e., $\mathbb{E} \hat{\rho}=\rho$. One can estimate a linear observable $O$ as $\tr(\hat{\rho}O)$. The purity $\tr(\rho_S^2)$ of some RDM $\rho_S$ is nonlinear, but can be estimated as $\tr{(\hat{\rho}\otimes \hat{\rho}' \; \mathbb{S}_{S}\otimes \id_{[N]/S})}$, where $\hat{\rho}$ and $\hat{\rho}'$ are two distinct shadow snapshots,
and $\mathbb{S}_{S}$ is the swap operator on the 2-copy Hilbert space restricted to subsystem $S$. To reduce the estimation uncertainty, one repeats this process $M$ times to collect the shadow set $\{\hat{\rho}_i\}_{i=1}^M$. For a collection of $L$ observables $\{O_i\}$, the number of samples is $M=\mc{O}(\log(L/\delta)\max_i\|O_i\|_{\mathrm{shadow}}^2/\epsilon_s^2)$ \cite{huang2020predicting}, where the shadow norm is directly related to the variance in a single-shot measurement.

Here, our target observables are purities of many RDMs as shown in \cref{eq:ErrD2},
and the shadow norm of $O=\mathbb{S}_{S}\otimes \id_{N/S}$ satisfies $\|O\|_{\mathrm{shadow}}\leq 4^{2|S|}$ \cite{huang2022provably}.
In \cref{eq:ErrD2}, for the lattice Hamiltonian one has $\max_{j}\|E_j\|=\mathcal{O}(1)$, $\|E\|_F=\mathcal{O} (\sqrt{N})$, and the summation includes $\mathcal{O} (N^2)$ terms. Thus we should choose the shadow estimation error $\epsilon_s=\mathcal{O}(N^{-2})$ to make the estimation sufficiently accurate. The local systems comprise $|\support(E_j^{\dag}E_{j'})|=\mathcal{O}(1)$ qubits.
As a result, the total number of samples for an $N$-qubit system is $M=\mathcal{O}(N^4 \log N)$, as $\|O\|_{\mathrm{shadow}}=\mc{O}(1)$.

We can further reduce the cost of estimating the Trotter error by directly measuring certain local Pauli operators instead of purities. In particular, one can give a more refined error bound beyond the distance- or entanglement-based bound as follows (see Section VI of the Supplemental Material for details):
\begin{align}\label{Merror:pauli}
    \epsilon(i)=
      \mathcal O\left(\sqrt{\sum_{j, j'} \bra{\phi_i}E_j^{\dag}  E_{j'}\ket{\phi_i}} \, \delta t^{p+1}\right).
  \end{align}
As a result, we can directly estimate the error by measuring the terms in the square root, $\sum_w \bra{\phi_i}O_w\ket{\phi_i}$ with $O_w=E_j^{\dag}  E_{j'}+E_{j'}^{\dag} E_j$ for $j\neq j'$ and  $O_w=E_j^{\dag}  E_{j}$ for $j=j'$,
which are each linear combinations of a constant number of Pauli operators.
The shadow estimation error of $O=\sum_w O_w$ should be kept to $\epsilon_s=\mc O (N)$ to make the estimation accurate enough to be comparable with the average performance. Since this observable is a sum Pauli operators $P$ with $\|P\|_{\mathrm{shadow}}\leq 3^{\mathrm{supp}(P)}$, we find that $M=\mc{O}(N^2)$ samples suffice to achieve the desired error \cite{zhou2023performance,huang2020predicting}.

\bibliography{BibEntSim}

\providecommand*\hyphen{-}
\begin{thebibliography}{50}%
\makeatletter
\providecommand \@ifxundefined [1]{%
 \@ifx{#1\undefined}
}%
\providecommand \@ifnum [1]{%
 \ifnum #1\expandafter \@firstoftwo
 \else \expandafter \@secondoftwo
 \fi
}%
\providecommand \@ifx [1]{%
 \ifx #1\expandafter \@firstoftwo
 \else \expandafter \@secondoftwo
 \fi
}%
\providecommand \natexlab [1]{#1}%
\providecommand \enquote  [1]{``#1''}%
\providecommand \bibnamefont  [1]{#1}%
\providecommand \bibfnamefont [1]{#1}%
\providecommand \citenamefont [1]{#1}%
\providecommand \href@noop [0]{\@secondoftwo}%
\providecommand \href [0]{\begingroup \@sanitize@url \@href}%
\providecommand \@href[1]{\@@startlink{#1}\@@href}%
\providecommand \@@href[1]{\endgroup#1\@@endlink}%
\providecommand \@sanitize@url [0]{\catcode `\\12\catcode `\$12\catcode
  `\&12\catcode `\#12\catcode `\^12\catcode `\_12\catcode `\%12\relax}%
\providecommand \@@startlink[1]{}%
\providecommand \@@endlink[0]{}%
\providecommand \url  [0]{\begingroup\@sanitize@url \@url }%
\providecommand \@url [1]{\endgroup\@href {#1}{\urlprefix }}%
\providecommand \urlprefix  [0]{URL }%
\providecommand \Eprint [0]{\href }%
\providecommand \doibase [0]{https://doi.org/}%
\providecommand \selectlanguage [0]{\@gobble}%
\providecommand \bibinfo  [0]{\@secondoftwo}%
\providecommand \bibfield  [0]{\@secondoftwo}%
\providecommand \translation [1]{[#1]}%
\providecommand \BibitemOpen [0]{}%
\providecommand \bibitemStop [0]{}%
\providecommand \bibitemNoStop [0]{.\EOS\space}%
\providecommand \EOS [0]{\spacefactor3000\relax}%
\providecommand \BibitemShut  [1]{\csname bibitem#1\endcsname}%
\let\auto@bib@innerbib\@empty
\bibitem [{\citenamefont {Amico}\ \emph {et~al.}(2008)\citenamefont {Amico},
  \citenamefont {Fazio}, \citenamefont {Osterloh},\ and\ \citenamefont
  {Vedral}}]{Amico2008Entanglement}%
  \BibitemOpen
  \bibfield  {author} {\bibinfo {author} {\bibfnamefont {L.}~\bibnamefont
  {Amico}}, \bibinfo {author} {\bibfnamefont {R.}~\bibnamefont {Fazio}},
  \bibinfo {author} {\bibfnamefont {A.}~\bibnamefont {Osterloh}},\ and\
  \bibinfo {author} {\bibfnamefont {V.}~\bibnamefont {Vedral}},\ }\bibfield
  {title} {\bibinfo {title} {Entanglement in many-body systems},\ }\href
  {https://doi.org/10.1103/RevModPhys.80.517} {\bibfield  {journal} {\bibinfo
  {journal} {Rev. Mod. Phys.}\ }\textbf {\bibinfo {volume} {80}},\ \bibinfo
  {pages} {517} (\bibinfo {year} {2008})}\BibitemShut {NoStop}%
\bibitem [{\citenamefont {Horodecki}\ \emph {et~al.}(2009)\citenamefont
  {Horodecki}, \citenamefont {Horodecki}, \citenamefont {Horodecki},\ and\
  \citenamefont {Horodecki}}]{Horodecki2009entanglement}%
  \BibitemOpen
  \bibfield  {author} {\bibinfo {author} {\bibfnamefont {R.}~\bibnamefont
  {Horodecki}}, \bibinfo {author} {\bibfnamefont {P.}~\bibnamefont
  {Horodecki}}, \bibinfo {author} {\bibfnamefont {M.}~\bibnamefont
  {Horodecki}},\ and\ \bibinfo {author} {\bibfnamefont {K.}~\bibnamefont
  {Horodecki}},\ }\bibfield  {title} {\bibinfo {title} {Quantum entanglement},\
  }\href {https://doi.org/10.1103/RevModPhys.81.865} {\bibfield  {journal}
  {\bibinfo  {journal} {Rev. Mod. Phys.}\ }\textbf {\bibinfo {volume} {81}},\
  \bibinfo {pages} {865} (\bibinfo {year} {2009})}\BibitemShut {NoStop}%
\bibitem [{\citenamefont {Walter}\ \emph {et~al.}(2016)\citenamefont {Walter},
  \citenamefont {Gross},\ and\ \citenamefont
  {Eisert}}]{walter2016multipartite}%
  \BibitemOpen
  \bibfield  {author} {\bibinfo {author} {\bibfnamefont {M.}~\bibnamefont
  {Walter}}, \bibinfo {author} {\bibfnamefont {D.}~\bibnamefont {Gross}},\ and\
  \bibinfo {author} {\bibfnamefont {J.}~\bibnamefont {Eisert}},\ }\bibfield
  {title} {\bibinfo {title} {Multipartite entanglement},\ }\href@noop {}
  {\bibfield  {journal} {\bibinfo  {journal} {Quantum Information: From
  Foundations to Quantum Technology Applications}\ ,\ \bibinfo {pages} {293}}
  (\bibinfo {year} {2016})}\BibitemShut {NoStop}%
\bibitem [{\citenamefont {Zeng}\ \emph {et~al.}(2015)\citenamefont {Zeng},
  \citenamefont {Chen}, \citenamefont {Zhou},\ and\ \citenamefont
  {Wen}}]{Bei2019Meets}%
  \BibitemOpen
  \bibfield  {author} {\bibinfo {author} {\bibfnamefont {B.}~\bibnamefont
  {Zeng}}, \bibinfo {author} {\bibfnamefont {X.}~\bibnamefont {Chen}}, \bibinfo
  {author} {\bibfnamefont {D.-L.}\ \bibnamefont {Zhou}},\ and\ \bibinfo
  {author} {\bibfnamefont {X.-G.}\ \bibnamefont {Wen}},\ }\href
  {https://arxiv.org/abs/1508.02595} {\bibinfo {title} {Quantum information
  meets quantum matter -- from quantum entanglement to topological phase in
  many-body systems}} (\bibinfo {year} {2015}),\ \Eprint
  {https://arxiv.org/abs/1508.02595} {arXiv:1508.02595 [quant-ph]} \BibitemShut
  {NoStop}%
\bibitem [{\citenamefont {Qi}(2018)}]{Qi2018gravity}%
  \BibitemOpen
  \bibfield  {author} {\bibinfo {author} {\bibfnamefont {X.-L.}\ \bibnamefont
  {Qi}},\ }\bibfield  {title} {\bibinfo {title} {Does gravity come from quantum
  information?},\ }\href {https://doi.org/10.1038/s41567-018-0297-3} {\bibfield
   {journal} {\bibinfo  {journal} {Nature Physics}\ }\textbf {\bibinfo {volume}
  {14}},\ \bibinfo {pages} {984} (\bibinfo {year} {2018})}\BibitemShut
  {NoStop}%
\bibitem [{\citenamefont {Vidal}(2003)}]{vidal03efficient}%
  \BibitemOpen
  \bibfield  {author} {\bibinfo {author} {\bibfnamefont {G.}~\bibnamefont
  {Vidal}},\ }\bibfield  {title} {\bibinfo {title} {Efficient classical
  simulation of slightly entangled quantum computations},\ }\href
  {https://doi.org/10.1103/PhysRevLett.91.147902} {\bibfield  {journal}
  {\bibinfo  {journal} {Phys. Rev. Lett.}\ }\textbf {\bibinfo {volume} {91}},\
  \bibinfo {pages} {147902} (\bibinfo {year} {2003})}\BibitemShut {NoStop}%
\bibitem [{\citenamefont {Datta}\ \emph {et~al.}(2008)\citenamefont {Datta},
  \citenamefont {Shaji},\ and\ \citenamefont {Caves}}]{DattaPRL08}%
  \BibitemOpen
  \bibfield  {author} {\bibinfo {author} {\bibfnamefont {A.}~\bibnamefont
  {Datta}}, \bibinfo {author} {\bibfnamefont {A.}~\bibnamefont {Shaji}},\ and\
  \bibinfo {author} {\bibfnamefont {C.~M.}\ \bibnamefont {Caves}},\ }\bibfield
  {title} {\bibinfo {title} {Quantum discord and the power of one qubit},\
  }\href {https://doi.org/10.1103/PhysRevLett.100.050502} {\bibfield  {journal}
  {\bibinfo  {journal} {Phys. Rev. Lett.}\ }\textbf {\bibinfo {volume} {100}},\
  \bibinfo {pages} {050502} (\bibinfo {year} {2008})}\BibitemShut {NoStop}%
\bibitem [{\citenamefont {D'Alessio}\ \emph {et~al.}(2016)\citenamefont
  {D'Alessio}, \citenamefont {Kafri}, \citenamefont {Polkovnikov},\ and\
  \citenamefont {Rigol}}]{Alessio2016quantum}%
  \BibitemOpen
  \bibfield  {author} {\bibinfo {author} {\bibfnamefont {L.}~\bibnamefont
  {D'Alessio}}, \bibinfo {author} {\bibfnamefont {Y.}~\bibnamefont {Kafri}},
  \bibinfo {author} {\bibfnamefont {A.}~\bibnamefont {Polkovnikov}},\ and\
  \bibinfo {author} {\bibfnamefont {M.}~\bibnamefont {Rigol}},\ }\bibfield
  {title} {\bibinfo {title} {From quantum chaos and eigenstate thermalization
  to statistical mechanics and thermodynamics},\ }\href@noop {} {\bibfield
  {journal} {\bibinfo  {journal} {Advances in Physics}\ }\textbf {\bibinfo
  {volume} {65}},\ \bibinfo {pages} {239} (\bibinfo {year} {2016})}\BibitemShut
  {NoStop}%
\bibitem [{\citenamefont {Gogolin}\ and\ \citenamefont
  {Eisert}(2016)}]{gogolin2016equilibration}%
  \BibitemOpen
  \bibfield  {author} {\bibinfo {author} {\bibfnamefont {C.}~\bibnamefont
  {Gogolin}}\ and\ \bibinfo {author} {\bibfnamefont {J.}~\bibnamefont
  {Eisert}},\ }\bibfield  {title} {\bibinfo {title} {Equilibration,
  thermalisation, and the emergence of statistical mechanics in closed quantum
  systems},\ }\href@noop {} {\bibfield  {journal} {\bibinfo  {journal} {Reports
  on Progress in Physics}\ }\textbf {\bibinfo {volume} {79}},\ \bibinfo {pages}
  {056001} (\bibinfo {year} {2016})}\BibitemShut {NoStop}%
\bibitem [{\citenamefont {Nandkishore}\ and\ \citenamefont
  {Huse}(2015)}]{nandkishore2015many}%
  \BibitemOpen
  \bibfield  {author} {\bibinfo {author} {\bibfnamefont {R.}~\bibnamefont
  {Nandkishore}}\ and\ \bibinfo {author} {\bibfnamefont {D.~A.}\ \bibnamefont
  {Huse}},\ }\bibfield  {title} {\bibinfo {title} {Many-body localization and
  thermalization in quantum statistical mechanics},\ }\href@noop {} {\bibfield
  {journal} {\bibinfo  {journal} {Annu. Rev. Condens. Matter Phys.}\ }\textbf
  {\bibinfo {volume} {6}},\ \bibinfo {pages} {15} (\bibinfo {year}
  {2015})}\BibitemShut {NoStop}%
\bibitem [{\citenamefont {Serbyn}\ \emph {et~al.}(2021)\citenamefont {Serbyn},
  \citenamefont {Abanin},\ and\ \citenamefont {Papi{\'c}}}]{serbyn2021quantum}%
  \BibitemOpen
  \bibfield  {author} {\bibinfo {author} {\bibfnamefont {M.}~\bibnamefont
  {Serbyn}}, \bibinfo {author} {\bibfnamefont {D.~A.}\ \bibnamefont {Abanin}},\
  and\ \bibinfo {author} {\bibfnamefont {Z.}~\bibnamefont {Papi{\'c}}},\
  }\bibfield  {title} {\bibinfo {title} {Quantum many-body scars and weak
  breaking of ergodicity},\ }\href@noop {} {\bibfield  {journal} {\bibinfo
  {journal} {Nature Physics}\ }\textbf {\bibinfo {volume} {17}},\ \bibinfo
  {pages} {675} (\bibinfo {year} {2021})}\BibitemShut {NoStop}%
\bibitem [{\citenamefont {Cirac}\ \emph {et~al.}(2021)\citenamefont {Cirac},
  \citenamefont {Perez-Garcia}, \citenamefont {Schuch},\ and\ \citenamefont
  {Verstraete}}]{cirac2021matrix}%
  \BibitemOpen
  \bibfield  {author} {\bibinfo {author} {\bibfnamefont {J.~I.}\ \bibnamefont
  {Cirac}}, \bibinfo {author} {\bibfnamefont {D.}~\bibnamefont {Perez-Garcia}},
  \bibinfo {author} {\bibfnamefont {N.}~\bibnamefont {Schuch}},\ and\ \bibinfo
  {author} {\bibfnamefont {F.}~\bibnamefont {Verstraete}},\ }\bibfield  {title}
  {\bibinfo {title} {Matrix product states and projected entangled pair states:
  Concepts, symmetries, theorems},\ }\href@noop {} {\bibfield  {journal}
  {\bibinfo  {journal} {Reviews of Modern Physics}\ }\textbf {\bibinfo {volume}
  {93}},\ \bibinfo {pages} {045003} (\bibinfo {year} {2021})}\BibitemShut
  {NoStop}%
\bibitem [{\citenamefont {Or{\'u}s}(2019)}]{orus2019tensor}%
  \BibitemOpen
  \bibfield  {author} {\bibinfo {author} {\bibfnamefont {R.}~\bibnamefont
  {Or{\'u}s}},\ }\bibfield  {title} {\bibinfo {title} {Tensor networks for
  complex quantum systems},\ }\href@noop {} {\bibfield  {journal} {\bibinfo
  {journal} {Nature Reviews Physics}\ }\textbf {\bibinfo {volume} {1}},\
  \bibinfo {pages} {538} (\bibinfo {year} {2019})}\BibitemShut {NoStop}%
\bibitem [{\citenamefont {Altman}\ \emph {et~al.}(2021)\citenamefont {Altman},
  \citenamefont {Brown}, \citenamefont {Carleo}, \citenamefont {Carr},
  \citenamefont {Demler}, \citenamefont {Chin}, \citenamefont {DeMarco},
  \citenamefont {Economou}, \citenamefont {Eriksson}, \citenamefont {Fu} \emph
  {et~al.}}]{altman2021quantum}%
  \BibitemOpen
  \bibfield  {author} {\bibinfo {author} {\bibfnamefont {E.}~\bibnamefont
  {Altman}}, \bibinfo {author} {\bibfnamefont {K.~R.}\ \bibnamefont {Brown}},
  \bibinfo {author} {\bibfnamefont {G.}~\bibnamefont {Carleo}}, \bibinfo
  {author} {\bibfnamefont {L.~D.}\ \bibnamefont {Carr}}, \bibinfo {author}
  {\bibfnamefont {E.}~\bibnamefont {Demler}}, \bibinfo {author} {\bibfnamefont
  {C.}~\bibnamefont {Chin}}, \bibinfo {author} {\bibfnamefont {B.}~\bibnamefont
  {DeMarco}}, \bibinfo {author} {\bibfnamefont {S.~E.}\ \bibnamefont
  {Economou}}, \bibinfo {author} {\bibfnamefont {M.~A.}\ \bibnamefont
  {Eriksson}}, \bibinfo {author} {\bibfnamefont {K.-M.~C.}\ \bibnamefont {Fu}},
  \emph {et~al.},\ }\bibfield  {title} {\bibinfo {title} {Quantum simulators:
  Architectures and opportunities},\ }\href@noop {} {\bibfield  {journal}
  {\bibinfo  {journal} {PRX Quantum}\ }\textbf {\bibinfo {volume} {2}},\
  \bibinfo {pages} {017003} (\bibinfo {year} {2021})}\BibitemShut {NoStop}%
\bibitem [{\citenamefont {Lloyd}(1996)}]{sethuniversal}%
  \BibitemOpen
  \bibfield  {author} {\bibinfo {author} {\bibfnamefont {S.}~\bibnamefont
  {Lloyd}},\ }\bibfield  {title} {\bibinfo {title} {Universal quantum
  simulators},\ }\href
  {https://doi.org/https://doi.org/10.1126/science.273.5278.1073} {\bibfield
  {journal} {\bibinfo  {journal} {Science}\ }\textbf {\bibinfo {volume}
  {273}},\ \bibinfo {pages} {1073} (\bibinfo {year} {1996})}\BibitemShut
  {NoStop}%
\bibitem [{\citenamefont {Berry}\ \emph {et~al.}(2007)\citenamefont {Berry},
  \citenamefont {Ahokas}, \citenamefont {Cleve},\ and\ \citenamefont
  {Sanders}}]{berry2007efficient}%
  \BibitemOpen
  \bibfield  {author} {\bibinfo {author} {\bibfnamefont {D.~W.}\ \bibnamefont
  {Berry}}, \bibinfo {author} {\bibfnamefont {G.}~\bibnamefont {Ahokas}},
  \bibinfo {author} {\bibfnamefont {R.}~\bibnamefont {Cleve}},\ and\ \bibinfo
  {author} {\bibfnamefont {B.~C.}\ \bibnamefont {Sanders}},\ }\bibfield
  {title} {\bibinfo {title} {Efficient quantum algorithms for simulating sparse
  {H}amiltonians},\ }\href
  {https://link.springer.com/article/10.1007/s00220-006-0150-x} {\bibfield
  {journal} {\bibinfo  {journal} {Commun. Math. Phys.}\ }\textbf {\bibinfo
  {volume} {270}},\ \bibinfo {pages} {359} (\bibinfo {year}
  {2007})}\BibitemShut {NoStop}%
\bibitem [{\citenamefont {Berry}\ and\ \citenamefont
  {Childs}(2012)}]{berry2012black}%
  \BibitemOpen
  \bibfield  {author} {\bibinfo {author} {\bibfnamefont {D.~W.}\ \bibnamefont
  {Berry}}\ and\ \bibinfo {author} {\bibfnamefont {A.~M.}\ \bibnamefont
  {Childs}},\ }\bibfield  {title} {\bibinfo {title} {Black-box {H}amiltonian
  simulation and unitary implementation},\ }\href
  {https://dl.acm.org/doi/10.5555/2231036.2231040} {\bibfield  {journal}
  {\bibinfo  {journal} {Quantum Info. Comput.}\ }\textbf {\bibinfo {volume}
  {12}},\ \bibinfo {pages} {29} (\bibinfo {year} {2012})}\BibitemShut {NoStop}%
\bibitem [{\citenamefont {Berry}\ \emph {et~al.}(2015)\citenamefont {Berry},
  \citenamefont {Childs}, \citenamefont {Cleve}, \citenamefont {Kothari},\ and\
  \citenamefont {Somma}}]{TaylorSeries}%
  \BibitemOpen
  \bibfield  {author} {\bibinfo {author} {\bibfnamefont {D.~W.}\ \bibnamefont
  {Berry}}, \bibinfo {author} {\bibfnamefont {A.~M.}\ \bibnamefont {Childs}},
  \bibinfo {author} {\bibfnamefont {R.}~\bibnamefont {Cleve}}, \bibinfo
  {author} {\bibfnamefont {R.}~\bibnamefont {Kothari}},\ and\ \bibinfo {author}
  {\bibfnamefont {R.~D.}\ \bibnamefont {Somma}},\ }\bibfield  {title} {\bibinfo
  {title} {Simulating {H}amiltonian dynamics with a truncated {T}aylor
  series},\ }\href {https://doi.org/10.1103/PhysRevLett.114.090502} {\bibfield
  {journal} {\bibinfo  {journal} {Phys. Rev. Lett.}\ }\textbf {\bibinfo
  {volume} {114}},\ \bibinfo {pages} {090502} (\bibinfo {year}
  {2015})}\BibitemShut {NoStop}%
\bibitem [{\citenamefont {Low}\ and\ \citenamefont
  {Chuang}(2017)}]{PhysRevLett.118.010501}%
  \BibitemOpen
  \bibfield  {author} {\bibinfo {author} {\bibfnamefont {G.~H.}\ \bibnamefont
  {Low}}\ and\ \bibinfo {author} {\bibfnamefont {I.~L.}\ \bibnamefont
  {Chuang}},\ }\bibfield  {title} {\bibinfo {title} {Optimal {H}amiltonian
  simulation by quantum signal processing},\ }\href
  {https://doi.org/10.1103/PhysRevLett.118.010501} {\bibfield  {journal}
  {\bibinfo  {journal} {Phys. Rev. Lett.}\ }\textbf {\bibinfo {volume} {118}},\
  \bibinfo {pages} {010501} (\bibinfo {year} {2017})}\BibitemShut {NoStop}%
\bibitem [{\citenamefont {Low}\ and\ \citenamefont
  {Chuang}(2019)}]{low2019hamiltonian}%
  \BibitemOpen
  \bibfield  {author} {\bibinfo {author} {\bibfnamefont {G.~H.}\ \bibnamefont
  {Low}}\ and\ \bibinfo {author} {\bibfnamefont {I.~L.}\ \bibnamefont
  {Chuang}},\ }\bibfield  {title} {\bibinfo {title} {Hamiltonian simulation by
  qubitization},\ }\href {https://quantum-journal.org/papers/q-2019-07-12-163/}
  {\bibfield  {journal} {\bibinfo  {journal} {Quantum}\ }\textbf {\bibinfo
  {volume} {3}},\ \bibinfo {pages} {163} (\bibinfo {year} {2019})}\BibitemShut
  {NoStop}%
\bibitem [{\citenamefont {Childs}\ and\ \citenamefont
  {Su}(2019{\natexlab{a}})}]{childs2019nearly}%
  \BibitemOpen
  \bibfield  {author} {\bibinfo {author} {\bibfnamefont {A.~M.}\ \bibnamefont
  {Childs}}\ and\ \bibinfo {author} {\bibfnamefont {Y.}~\bibnamefont {Su}},\
  }\bibfield  {title} {\bibinfo {title} {Nearly optimal lattice simulation by
  product formulas},\ }\href {https://doi.org/10.1103/PhysRevLett.123.050503}
  {\bibfield  {journal} {\bibinfo  {journal} {Phys. Rev. Lett.}\ }\textbf
  {\bibinfo {volume} {123}},\ \bibinfo {pages} {050503} (\bibinfo {year}
  {2019}{\natexlab{a}})}\BibitemShut {NoStop}%
\bibitem [{\citenamefont {Childs}\ \emph {et~al.}(2018)\citenamefont {Childs},
  \citenamefont {Maslov}, \citenamefont {Nam}, \citenamefont {Ross},\ and\
  \citenamefont {Su}}]{childs2018toward}%
  \BibitemOpen
  \bibfield  {author} {\bibinfo {author} {\bibfnamefont {A.~M.}\ \bibnamefont
  {Childs}}, \bibinfo {author} {\bibfnamefont {D.}~\bibnamefont {Maslov}},
  \bibinfo {author} {\bibfnamefont {Y.}~\bibnamefont {Nam}}, \bibinfo {author}
  {\bibfnamefont {N.~J.}\ \bibnamefont {Ross}},\ and\ \bibinfo {author}
  {\bibfnamefont {Y.}~\bibnamefont {Su}},\ }\bibfield  {title} {\bibinfo
  {title} {Toward the first quantum simulation with quantum speedup},\ }\href
  {https://www.pnas.org/content/115/38/9456} {\bibfield  {journal} {\bibinfo
  {journal} {Proc. Natl. Acad. Sci. U.S.A.}\ }\textbf {\bibinfo {volume}
  {115}},\ \bibinfo {pages} {9456} (\bibinfo {year} {2018})}\BibitemShut
  {NoStop}%
\bibitem [{\citenamefont {Childs}\ \emph {et~al.}(2021)\citenamefont {Childs},
  \citenamefont {Su}, \citenamefont {Tran}, \citenamefont {Wiebe},\ and\
  \citenamefont {Zhu}}]{childs2020theory}%
  \BibitemOpen
  \bibfield  {author} {\bibinfo {author} {\bibfnamefont {A.~M.}\ \bibnamefont
  {Childs}}, \bibinfo {author} {\bibfnamefont {Y.}~\bibnamefont {Su}}, \bibinfo
  {author} {\bibfnamefont {M.~C.}\ \bibnamefont {Tran}}, \bibinfo {author}
  {\bibfnamefont {N.}~\bibnamefont {Wiebe}},\ and\ \bibinfo {author}
  {\bibfnamefont {S.}~\bibnamefont {Zhu}},\ }\bibfield  {title} {\bibinfo
  {title} {Theory of {T}rotter error with commutator scaling},\ }\href
  {https://doi.org/10.1103/PhysRevX.11.011020} {\bibfield  {journal} {\bibinfo
  {journal} {Phys. Rev. X}\ }\textbf {\bibinfo {volume} {11}},\ \bibinfo
  {pages} {011020} (\bibinfo {year} {2021})}\BibitemShut {NoStop}%
\bibitem [{\citenamefont {{\c{S}}ahino{\u{g}}lu}\ and\ \citenamefont
  {Somma}(2021)}]{sahinoglu2020hamiltonian}%
  \BibitemOpen
  \bibfield  {author} {\bibinfo {author} {\bibfnamefont {B.}~\bibnamefont
  {{\c{S}}ahino{\u{g}}lu}}\ and\ \bibinfo {author} {\bibfnamefont {R.~D.}\
  \bibnamefont {Somma}},\ }\bibfield  {title} {\bibinfo {title} {Hamiltonian
  simulation in the low-energy subspace},\ }\href
  {https://www.nature.com/articles/s41534-021-00451-w} {\bibfield  {journal}
  {\bibinfo  {journal} {npj Quantum Information}\ }\textbf {\bibinfo {volume}
  {7}},\ \bibinfo {pages} {1} (\bibinfo {year} {2021})}\BibitemShut {NoStop}%
\bibitem [{\citenamefont {An}\ \emph {et~al.}(2021)\citenamefont {An},
  \citenamefont {Fang},\ and\ \citenamefont {Lin}}]{An2021timedependent}%
  \BibitemOpen
  \bibfield  {author} {\bibinfo {author} {\bibfnamefont {D.}~\bibnamefont
  {An}}, \bibinfo {author} {\bibfnamefont {D.}~\bibnamefont {Fang}},\ and\
  \bibinfo {author} {\bibfnamefont {L.}~\bibnamefont {Lin}},\ }\bibfield
  {title} {\bibinfo {title} {Time-dependent unbounded {H}amiltonian simulation
  with vector norm scaling},\ }\href
  {https://doi.org/10.22331/q-2021-05-26-459} {\bibfield  {journal} {\bibinfo
  {journal} {{Quantum}}\ }\textbf {\bibinfo {volume} {5}},\ \bibinfo {pages}
  {459} (\bibinfo {year} {2021})}\BibitemShut {NoStop}%
\bibitem [{\citenamefont {Zhao}\ \emph {et~al.}(2022)\citenamefont {Zhao},
  \citenamefont {Zhou}, \citenamefont {Shaw}, \citenamefont {Li},\ and\
  \citenamefont {Childs}}]{Zhaorandom22}%
  \BibitemOpen
  \bibfield  {author} {\bibinfo {author} {\bibfnamefont {Q.}~\bibnamefont
  {Zhao}}, \bibinfo {author} {\bibfnamefont {Y.}~\bibnamefont {Zhou}}, \bibinfo
  {author} {\bibfnamefont {A.~F.}\ \bibnamefont {Shaw}}, \bibinfo {author}
  {\bibfnamefont {T.}~\bibnamefont {Li}},\ and\ \bibinfo {author}
  {\bibfnamefont {A.~M.}\ \bibnamefont {Childs}},\ }\bibfield  {title}
  {\bibinfo {title} {Hamiltonian simulation with random inputs},\ }\href
  {https://doi.org/10.1103/PhysRevLett.129.270502} {\bibfield  {journal}
  {\bibinfo  {journal} {Phys. Rev. Lett.}\ }\textbf {\bibinfo {volume} {129}},\
  \bibinfo {pages} {270502} (\bibinfo {year} {2022})}\BibitemShut {NoStop}%
\bibitem [{\citenamefont {Chen}\ and\ \citenamefont
  {Brand{\~a}o}(2024)}]{chen2024average}%
  \BibitemOpen
  \bibfield  {author} {\bibinfo {author} {\bibfnamefont {C.-F.}\ \bibnamefont
  {Chen}}\ and\ \bibinfo {author} {\bibfnamefont {F.~G.}\ \bibnamefont
  {Brand{\~a}o}},\ }\bibfield  {title} {\bibinfo {title} {Average-case speedup
  for product formulas},\ }\href@noop {} {\bibfield  {journal} {\bibinfo
  {journal} {Communications in Mathematical Physics}\ }\textbf {\bibinfo
  {volume} {405}},\ \bibinfo {pages} {32} (\bibinfo {year} {2024})}\BibitemShut
  {NoStop}%
\bibitem [{\citenamefont {Su}\ \emph {et~al.}(2021)\citenamefont {Su},
  \citenamefont {Huang},\ and\ \citenamefont {Campbell}}]{su2020nearly}%
  \BibitemOpen
  \bibfield  {author} {\bibinfo {author} {\bibfnamefont {Y.}~\bibnamefont
  {Su}}, \bibinfo {author} {\bibfnamefont {H.-Y.}\ \bibnamefont {Huang}},\ and\
  \bibinfo {author} {\bibfnamefont {E.~T.}\ \bibnamefont {Campbell}},\
  }\bibfield  {title} {\bibinfo {title} {Nearly tight {T}rotterization of
  interacting electrons},\ }\href
  {https://quantum-journal.org/papers/q-2021-07-05-495/} {\bibfield  {journal}
  {\bibinfo  {journal} {Quantum}\ }\textbf {\bibinfo {volume} {5}},\ \bibinfo
  {pages} {495} (\bibinfo {year} {2021})}\BibitemShut {NoStop}%
\bibitem [{\citenamefont {Scott}(2004)}]{scott2004multipartite}%
  \BibitemOpen
  \bibfield  {author} {\bibinfo {author} {\bibfnamefont {A.~J.}\ \bibnamefont
  {Scott}},\ }\bibfield  {title} {\bibinfo {title} {Multipartite entanglement,
  quantum-error-correcting codes, and entangling power of quantum evolutions},\
  }\href@noop {} {\bibfield  {journal} {\bibinfo  {journal} {Physical Review
  A}\ }\textbf {\bibinfo {volume} {69}},\ \bibinfo {pages} {052330} (\bibinfo
  {year} {2004})}\BibitemShut {NoStop}%
\bibitem [{\citenamefont {Hayden}\ \emph {et~al.}(2006)\citenamefont {Hayden},
  \citenamefont {Leung},\ and\ \citenamefont {Winter}}]{hayden2006aspects}%
  \BibitemOpen
  \bibfield  {author} {\bibinfo {author} {\bibfnamefont {P.}~\bibnamefont
  {Hayden}}, \bibinfo {author} {\bibfnamefont {D.~W.}\ \bibnamefont {Leung}},\
  and\ \bibinfo {author} {\bibfnamefont {A.}~\bibnamefont {Winter}},\
  }\bibfield  {title} {\bibinfo {title} {Aspects of generic entanglement},\
  }\href@noop {} {\bibfield  {journal} {\bibinfo  {journal} {Communications in
  mathematical physics}\ }\textbf {\bibinfo {volume} {265}},\ \bibinfo {pages}
  {95} (\bibinfo {year} {2006})}\BibitemShut {NoStop}%
\bibitem [{\citenamefont {Suzuki}(1991)}]{suzuki1991general}%
  \BibitemOpen
  \bibfield  {author} {\bibinfo {author} {\bibfnamefont {M.}~\bibnamefont
  {Suzuki}},\ }\bibfield  {title} {\bibinfo {title} {General theory of fractal
  path integrals with applications to many-body theories and statistical
  physics},\ }\href {https://aip.scitation.org/doi/10.1063/1.529425} {\bibfield
   {journal} {\bibinfo  {journal} {J. Math. Phys.}\ }\textbf {\bibinfo {volume}
  {32}},\ \bibinfo {pages} {400} (\bibinfo {year} {1991})}\BibitemShut
  {NoStop}%
\bibitem [{\citenamefont {Tran}\ \emph {et~al.}(2020)\citenamefont {Tran},
  \citenamefont {Chu}, \citenamefont {Su}, \citenamefont {Childs},\ and\
  \citenamefont {Gorshkov}}]{Tran_2020}%
  \BibitemOpen
  \bibfield  {author} {\bibinfo {author} {\bibfnamefont {M.~C.}\ \bibnamefont
  {Tran}}, \bibinfo {author} {\bibfnamefont {S.-K.}\ \bibnamefont {Chu}},
  \bibinfo {author} {\bibfnamefont {Y.}~\bibnamefont {Su}}, \bibinfo {author}
  {\bibfnamefont {A.~M.}\ \bibnamefont {Childs}},\ and\ \bibinfo {author}
  {\bibfnamefont {A.~V.}\ \bibnamefont {Gorshkov}},\ }\bibfield  {title}
  {\bibinfo {title} {Destructive error interference in product-formula lattice
  simulation},\ }\href {https://doi.org/10.1103/PhysRevLett.124.220502}
  {\bibfield  {journal} {\bibinfo  {journal} {Phys. Rev. Lett.}\ }\textbf
  {\bibinfo {volume} {124}},\ \bibinfo {pages} {220502} (\bibinfo {year}
  {2020})}\BibitemShut {NoStop}%
\bibitem [{\citenamefont {Layden}(2021)}]{layden2021first}%
  \BibitemOpen
  \bibfield  {author} {\bibinfo {author} {\bibfnamefont {D.}~\bibnamefont
  {Layden}},\ }\bibfield  {title} {\bibinfo {title} {First-order {T}rotter
  error from a second-order perspective},\ }\href@noop {} {\bibfield  {journal}
  {\bibinfo  {journal} {arXiv preprint arXiv:2107.08032}\ } (\bibinfo {year}
  {2021})}\BibitemShut {NoStop}%
\bibitem [{\citenamefont {Kim}\ \emph {et~al.}(2014)\citenamefont {Kim},
  \citenamefont {Ikeda},\ and\ \citenamefont {Huse}}]{kim2014testing}%
  \BibitemOpen
  \bibfield  {author} {\bibinfo {author} {\bibfnamefont {H.}~\bibnamefont
  {Kim}}, \bibinfo {author} {\bibfnamefont {T.~N.}\ \bibnamefont {Ikeda}},\
  and\ \bibinfo {author} {\bibfnamefont {D.~A.}\ \bibnamefont {Huse}},\
  }\bibfield  {title} {\bibinfo {title} {Testing whether all eigenstates obey
  the eigenstate thermalization hypothesis},\ }\href@noop {} {\bibfield
  {journal} {\bibinfo  {journal} {Physical Review E}\ }\textbf {\bibinfo
  {volume} {90}},\ \bibinfo {pages} {052105} (\bibinfo {year}
  {2014})}\BibitemShut {NoStop}%
\bibitem [{\citenamefont {Shi}\ \emph {et~al.}(2021)\citenamefont {Shi},
  \citenamefont {Li}, \citenamefont {Chen},\ and\ \citenamefont
  {Zhang}}]{PhysRevA.104.032601}%
  \BibitemOpen
  \bibfield  {author} {\bibinfo {author} {\bibfnamefont {F.}~\bibnamefont
  {Shi}}, \bibinfo {author} {\bibfnamefont {M.-S.}\ \bibnamefont {Li}},
  \bibinfo {author} {\bibfnamefont {L.}~\bibnamefont {Chen}},\ and\ \bibinfo
  {author} {\bibfnamefont {X.}~\bibnamefont {Zhang}},\ }\bibfield  {title}
  {\bibinfo {title} {$k$-uniform quantum information masking},\ }\href
  {https://doi.org/10.1103/PhysRevA.104.032601} {\bibfield  {journal} {\bibinfo
   {journal} {Phys. Rev. A}\ }\textbf {\bibinfo {volume} {104}},\ \bibinfo
  {pages} {032601} (\bibinfo {year} {2021})}\BibitemShut {NoStop}%
\bibitem [{\citenamefont {Paeckel}\ \emph {et~al.}(2019)\citenamefont
  {Paeckel}, \citenamefont {K{\"o}hler}, \citenamefont {Swoboda}, \citenamefont
  {Manmana}, \citenamefont {Schollw{\"o}ck},\ and\ \citenamefont
  {Hubig}}]{paeckel2019time}%
  \BibitemOpen
  \bibfield  {author} {\bibinfo {author} {\bibfnamefont {S.}~\bibnamefont
  {Paeckel}}, \bibinfo {author} {\bibfnamefont {T.}~\bibnamefont {K{\"o}hler}},
  \bibinfo {author} {\bibfnamefont {A.}~\bibnamefont {Swoboda}}, \bibinfo
  {author} {\bibfnamefont {S.~R.}\ \bibnamefont {Manmana}}, \bibinfo {author}
  {\bibfnamefont {U.}~\bibnamefont {Schollw{\"o}ck}},\ and\ \bibinfo {author}
  {\bibfnamefont {C.}~\bibnamefont {Hubig}},\ }\bibfield  {title} {\bibinfo
  {title} {Time-evolution methods for matrix-product states},\ }\href@noop {}
  {\bibfield  {journal} {\bibinfo  {journal} {Annals of Physics}\ }\textbf
  {\bibinfo {volume} {411}},\ \bibinfo {pages} {167998} (\bibinfo {year}
  {2019})}\BibitemShut {NoStop}%
\bibitem [{\citenamefont {Bravyi}\ \emph {et~al.}(2006)\citenamefont {Bravyi},
  \citenamefont {Hastings},\ and\ \citenamefont
  {Verstraete}}]{PhysRevLett.97.050401}%
  \BibitemOpen
  \bibfield  {author} {\bibinfo {author} {\bibfnamefont {S.}~\bibnamefont
  {Bravyi}}, \bibinfo {author} {\bibfnamefont {M.~B.}\ \bibnamefont
  {Hastings}},\ and\ \bibinfo {author} {\bibfnamefont {F.}~\bibnamefont
  {Verstraete}},\ }\bibfield  {title} {\bibinfo {title} {Lieb-robinson bounds
  and the generation of correlations and topological quantum order},\ }\href
  {https://doi.org/10.1103/PhysRevLett.97.050401} {\bibfield  {journal}
  {\bibinfo  {journal} {Phys. Rev. Lett.}\ }\textbf {\bibinfo {volume} {97}},\
  \bibinfo {pages} {050401} (\bibinfo {year} {2006})}\BibitemShut {NoStop}%
\bibitem [{\citenamefont {Aaronson}(2019)}]{aaronson2019shadow}%
  \BibitemOpen
  \bibfield  {author} {\bibinfo {author} {\bibfnamefont {S.}~\bibnamefont
  {Aaronson}},\ }\bibfield  {title} {\bibinfo {title} {Shadow tomography of
  quantum states},\ }\href {https://epubs.siam.org/doi/abs/10.1137/18M120275X}
  {\bibfield  {journal} {\bibinfo  {journal} {SIAM Journal on Computing}\
  }\textbf {\bibinfo {volume} {49}},\ \bibinfo {pages} {STOC18\hyphen 368}
  (\bibinfo {year} {2019})}\BibitemShut {NoStop}%
\bibitem [{\citenamefont {Huang}\ \emph {et~al.}(2020)\citenamefont {Huang},
  \citenamefont {Kueng},\ and\ \citenamefont {Preskill}}]{huang2020predicting}%
  \BibitemOpen
  \bibfield  {author} {\bibinfo {author} {\bibfnamefont {H.-Y.}\ \bibnamefont
  {Huang}}, \bibinfo {author} {\bibfnamefont {R.}~\bibnamefont {Kueng}},\ and\
  \bibinfo {author} {\bibfnamefont {J.}~\bibnamefont {Preskill}},\ }\bibfield
  {title} {\bibinfo {title} {Predicting many properties of a quantum system
  from very few measurements},\ }\bibfield  {journal} {\bibinfo  {journal}
  {Nature Physics}\ }\href {https://doi.org/10.1038/s41567-020-0932-7}
  {10.1038/s41567-020-0932-7} (\bibinfo {year} {2020})\BibitemShut {NoStop}%
\bibitem [{\citenamefont {Zhou}\ and\ \citenamefont
  {Liu}(2023)}]{zhou2023performance}%
  \BibitemOpen
  \bibfield  {author} {\bibinfo {author} {\bibfnamefont {Y.}~\bibnamefont
  {Zhou}}\ and\ \bibinfo {author} {\bibfnamefont {Q.}~\bibnamefont {Liu}},\
  }\bibfield  {title} {\bibinfo {title} {Performance analysis of multi-shot
  shadow estimation},\ }\href {https://doi.org/10.22331/q-2023-06-29-1044}
  {\bibfield  {journal} {\bibinfo  {journal} {Quantum}\ }\textbf {\bibinfo
  {volume} {7}},\ \bibinfo {pages} {1044} (\bibinfo {year} {2023})}\BibitemShut
  {NoStop}%
\bibitem [{\citenamefont {Kovalsky}\ \emph {et~al.}(2023)\citenamefont
  {Kovalsky}, \citenamefont {{Calderon-Vargas}}, \citenamefont {Grace},
  \citenamefont {Magann}, \citenamefont {Larsen}, \citenamefont {Baczewski},\
  and\ \citenamefont {Sarovar}}]{kovalskySelfhealingTrotterError2023}%
  \BibitemOpen
  \bibfield  {author} {\bibinfo {author} {\bibfnamefont {L.~K.}\ \bibnamefont
  {Kovalsky}}, \bibinfo {author} {\bibfnamefont {F.~A.}\ \bibnamefont
  {{Calderon-Vargas}}}, \bibinfo {author} {\bibfnamefont {M.~D.}\ \bibnamefont
  {Grace}}, \bibinfo {author} {\bibfnamefont {A.~B.}\ \bibnamefont {Magann}},
  \bibinfo {author} {\bibfnamefont {J.~B.}\ \bibnamefont {Larsen}}, \bibinfo
  {author} {\bibfnamefont {A.~D.}\ \bibnamefont {Baczewski}},\ and\ \bibinfo
  {author} {\bibfnamefont {M.}~\bibnamefont {Sarovar}},\ }\bibfield  {title}
  {\bibinfo {title} {Self-healing of {{Trotter}} error in digital adiabatic
  state preparation},\ }\href {https://doi.org/10.1103/PhysRevLett.131.060602}
  {\bibfield  {journal} {\bibinfo  {journal} {Phys. Rev. Lett.}\ }\textbf
  {\bibinfo {volume} {131}},\ \bibinfo {pages} {060602} (\bibinfo {year}
  {2023})},\ \Eprint {https://arxiv.org/abs/2209.06242} {arxiv:2209.06242
  [quant-ph]} \BibitemShut {NoStop}%
\bibitem [{\citenamefont {Yi}\ and\ \citenamefont
  {Crosson}(2022)}]{yi2022spectral}%
  \BibitemOpen
  \bibfield  {author} {\bibinfo {author} {\bibfnamefont {C.}~\bibnamefont
  {Yi}}\ and\ \bibinfo {author} {\bibfnamefont {E.}~\bibnamefont {Crosson}},\
  }\bibfield  {title} {\bibinfo {title} {Spectral analysis of product formulas
  for quantum simulation},\ }\href@noop {} {\bibfield  {journal} {\bibinfo
  {journal} {npj Quantum Information}\ }\textbf {\bibinfo {volume} {8}},\
  \bibinfo {pages} {37} (\bibinfo {year} {2022})}\BibitemShut {NoStop}%
\bibitem [{\citenamefont {Bravyi}(2015)}]{bravyi2015monte}%
  \BibitemOpen
  \bibfield  {author} {\bibinfo {author} {\bibfnamefont {S.}~\bibnamefont
  {Bravyi}},\ }\bibfield  {title} {\bibinfo {title} {Monte carlo simulation of
  stoquastic {H}amiltonians},\ }\href
  {https://dl.acm.org/doi/10.5555/2871363.2871366} {\bibfield  {journal}
  {\bibinfo  {journal} {Quantum Info. Comput.}\ }\textbf {\bibinfo {volume}
  {15}},\ \bibinfo {pages} {1122–1140} (\bibinfo {year} {2015})}\BibitemShut
  {NoStop}%
\bibitem [{\citenamefont {Bravyi}\ and\ \citenamefont
  {Gosset}(2017)}]{bravyi2017poly}%
  \BibitemOpen
  \bibfield  {author} {\bibinfo {author} {\bibfnamefont {S.}~\bibnamefont
  {Bravyi}}\ and\ \bibinfo {author} {\bibfnamefont {D.}~\bibnamefont
  {Gosset}},\ }\bibfield  {title} {\bibinfo {title} {Polynomial-time classical
  simulation of quantum ferromagnets},\ }\href
  {https://doi.org/10.1103/PhysRevLett.119.100503} {\bibfield  {journal}
  {\bibinfo  {journal} {Phys. Rev. Lett.}\ }\textbf {\bibinfo {volume} {119}},\
  \bibinfo {pages} {100503} (\bibinfo {year} {2017})}\BibitemShut {NoStop}%
\bibitem [{\citenamefont {Motta}\ \emph {et~al.}(2020)\citenamefont {Motta},
  \citenamefont {Sun}, \citenamefont {Tan}, \citenamefont {O’Rourke},
  \citenamefont {Ye}, \citenamefont {Minnich}, \citenamefont {Brand{\~a}o},\
  and\ \citenamefont {Chan}}]{motta2020determining}%
  \BibitemOpen
  \bibfield  {author} {\bibinfo {author} {\bibfnamefont {M.}~\bibnamefont
  {Motta}}, \bibinfo {author} {\bibfnamefont {C.}~\bibnamefont {Sun}}, \bibinfo
  {author} {\bibfnamefont {A.~T.}\ \bibnamefont {Tan}}, \bibinfo {author}
  {\bibfnamefont {M.~J.}\ \bibnamefont {O’Rourke}}, \bibinfo {author}
  {\bibfnamefont {E.}~\bibnamefont {Ye}}, \bibinfo {author} {\bibfnamefont
  {A.~J.}\ \bibnamefont {Minnich}}, \bibinfo {author} {\bibfnamefont {F.~G.}\
  \bibnamefont {Brand{\~a}o}},\ and\ \bibinfo {author} {\bibfnamefont
  {G.~K.-L.}\ \bibnamefont {Chan}},\ }\bibfield  {title} {\bibinfo {title}
  {Determining eigenstates and thermal states on a quantum computer using
  quantum imaginary time evolution},\ }\href
  {https://www.nature.com/articles/s41567-019-0704-4} {\bibfield  {journal}
  {\bibinfo  {journal} {Nature Physics}\ }\textbf {\bibinfo {volume} {16}},\
  \bibinfo {pages} {205} (\bibinfo {year} {2020})}\BibitemShut {NoStop}%
\bibitem [{\citenamefont {Howard}\ \emph {et~al.}(2014)\citenamefont {Howard},
  \citenamefont {Wallman}, \citenamefont {Veitch},\ and\ \citenamefont
  {Emerson}}]{howard2014contextuality}%
  \BibitemOpen
  \bibfield  {author} {\bibinfo {author} {\bibfnamefont {M.}~\bibnamefont
  {Howard}}, \bibinfo {author} {\bibfnamefont {J.}~\bibnamefont {Wallman}},
  \bibinfo {author} {\bibfnamefont {V.}~\bibnamefont {Veitch}},\ and\ \bibinfo
  {author} {\bibfnamefont {J.}~\bibnamefont {Emerson}},\ }\bibfield  {title}
  {\bibinfo {title} {Contextuality supplies the ‘magic’for quantum
  computation},\ }\href@noop {} {\bibfield  {journal} {\bibinfo  {journal}
  {Nature}\ }\textbf {\bibinfo {volume} {510}},\ \bibinfo {pages} {351}
  (\bibinfo {year} {2014})}\BibitemShut {NoStop}%
\bibitem [{\citenamefont {Streltsov}\ \emph {et~al.}(2017)\citenamefont
  {Streltsov}, \citenamefont {Adesso},\ and\ \citenamefont
  {Plenio}}]{RevModPhys.89.041003}%
  \BibitemOpen
  \bibfield  {author} {\bibinfo {author} {\bibfnamefont {A.}~\bibnamefont
  {Streltsov}}, \bibinfo {author} {\bibfnamefont {G.}~\bibnamefont {Adesso}},\
  and\ \bibinfo {author} {\bibfnamefont {M.~B.}\ \bibnamefont {Plenio}},\
  }\bibfield  {title} {\bibinfo {title} {Colloquium: Quantum coherence as a
  resource},\ }\href {https://doi.org/10.1103/RevModPhys.89.041003} {\bibfield
  {journal} {\bibinfo  {journal} {Rev. Mod. Phys.}\ }\textbf {\bibinfo {volume}
  {89}},\ \bibinfo {pages} {041003} (\bibinfo {year} {2017})}\BibitemShut
  {NoStop}%
\bibitem [{\citenamefont {Cotler}\ \emph {et~al.}(2021)\citenamefont {Cotler},
  \citenamefont {Mark}, \citenamefont {Huang}, \citenamefont {Hernandez},
  \citenamefont {Choi}, \citenamefont {Shaw}, \citenamefont {Endres},\ and\
  \citenamefont {Choi}}]{cotler2021emergent}%
  \BibitemOpen
  \bibfield  {author} {\bibinfo {author} {\bibfnamefont {J.~S.}\ \bibnamefont
  {Cotler}}, \bibinfo {author} {\bibfnamefont {D.~K.}\ \bibnamefont {Mark}},
  \bibinfo {author} {\bibfnamefont {H.-Y.}\ \bibnamefont {Huang}}, \bibinfo
  {author} {\bibfnamefont {F.}~\bibnamefont {Hernandez}}, \bibinfo {author}
  {\bibfnamefont {J.}~\bibnamefont {Choi}}, \bibinfo {author} {\bibfnamefont
  {A.~L.}\ \bibnamefont {Shaw}}, \bibinfo {author} {\bibfnamefont
  {M.}~\bibnamefont {Endres}},\ and\ \bibinfo {author} {\bibfnamefont
  {S.}~\bibnamefont {Choi}},\ }\bibfield  {title} {\bibinfo {title} {Emergent
  quantum state designs from individual many-body wavefunctions},\ }\href
  {https://arxiv.org/abs/2103.03536} {\bibfield  {journal} {\bibinfo  {journal}
  {arXiv preprint arXiv:2103.03536}\ } (\bibinfo {year} {2021})}\BibitemShut
  {NoStop}%
\bibitem [{\citenamefont {Huang}\ \emph {et~al.}(2022)\citenamefont {Huang},
  \citenamefont {Kueng}, \citenamefont {Torlai}, \citenamefont {Albert},\ and\
  \citenamefont {Preskill}}]{huang2022provably}%
  \BibitemOpen
  \bibfield  {author} {\bibinfo {author} {\bibfnamefont {H.-Y.}\ \bibnamefont
  {Huang}}, \bibinfo {author} {\bibfnamefont {R.}~\bibnamefont {Kueng}},
  \bibinfo {author} {\bibfnamefont {G.}~\bibnamefont {Torlai}}, \bibinfo
  {author} {\bibfnamefont {V.~V.}\ \bibnamefont {Albert}},\ and\ \bibinfo
  {author} {\bibfnamefont {J.}~\bibnamefont {Preskill}},\ }\bibfield  {title}
  {\bibinfo {title} {Provably efficient machine learning for quantum many-body
  problems},\ }\href@noop {} {\bibfield  {journal} {\bibinfo  {journal}
  {Science}\ }\textbf {\bibinfo {volume} {377}},\ \bibinfo {pages} {eabk3333}
  (\bibinfo {year} {2022})}\BibitemShut {NoStop}%
\bibitem [{\citenamefont {Childs}\ and\ \citenamefont
  {Su}(2019{\natexlab{b}})}]{Childs2019Product}%
  \BibitemOpen
  \bibfield  {author} {\bibinfo {author} {\bibfnamefont {A.~M.}\ \bibnamefont
  {Childs}}\ and\ \bibinfo {author} {\bibfnamefont {Y.}~\bibnamefont {Su}},\
  }\bibfield  {title} {\bibinfo {title} {Nearly optimal lattice simulation by
  product formulas},\ }\href {https://doi.org/10.1103/PhysRevLett.123.050503}
  {\bibfield  {journal} {\bibinfo  {journal} {Phys. Rev. Lett.}\ }\textbf
  {\bibinfo {volume} {123}},\ \bibinfo {pages} {050503} (\bibinfo {year}
  {2019}{\natexlab{b}})}\BibitemShut {NoStop}%
\end{thebibliography}%

\newpage

\section*{Supplemental Material: Entanglement accelerates quantum simulation}


\section{Proof of Theorem 1 and Corollary 1}\label{SM:SecTh1}
In this work, we mainly focus on product formula (PF) algorithms \cite{suzuki1991general,sethuniversal}.
For a short evolution time $\delta t$, the first-order product formula (PF1) algorithm for a Hamiltonian $\sum_{l=1}^LH_l$ applies the unitary operation
\begin{align}
    \mathscr{U}_1(\delta t):=e^{-iH_1\delta t}e^{-iH_2\delta t}\cdots e^{-iH_L\delta t}=\rprod_l e^{-iH_l \delta t}.
\end{align}
Here the right arrow indicates the product is in the order of increasing indices (i.e, for $l=1,2,\ldots,L$). Second-order product formulas (PF$2$)
can be obtained by combining evolutions in both increasing and decreasing orders of indices, with
\begin{align}
  \mathscr{U}_2(\delta t):= \rprod_l e^{-iH_l\delta t} \lprod_l e^{-iH_l\delta t},
\end{align}
where the left arrow indicates the product in decreasing order (i.e., for $l=L,L-1,\ldots,1$).

More generally, Suzuki constructed $p$th-order product formulas (PF$p$ for even $p=2k$) recursively from the second-order formula as~\cite{suzuki1991general}
\begin{align}
    \mathscr{U}_{2k}(\delta t)&= [\mathscr{U}_{2k-2}(p_k \delta t)]^2 \mathscr{U}_{2k-2}((1-4p_k)\delta t) [\mathscr{U}_{2k-2}(p_k \delta t)]^2,
\end{align}
where $p_k:=\frac{1}{4-4^{1/(2k-1)}}$ and $k>1$ \cite{suzuki1991general}.
Overall, we have $S=2\cdot 5^{k-1}$ stages (operators of the form $\mathscr{U}_1=\rprod_l e^{-iH_l\delta t}$ or its reverse ordering $\lprod_l e^{-iH_l\delta t}$) and the evolution can be rewritten as
\begin{equation}\label{eq:pf2k}
  \mathscr{U}_{2k}(\delta t)=  \prod_{s=1}^{S} \prod_{l=1}^{L} e^{-i \delta t a_s H_{\pi_s(l)}},
\end{equation}
where each $\pi_s$ is the identity permutation or reversal permutation and $\delta t a_s$ is the simulation time for stage $s$.

To simulate the evolution of a quantum system for a long time, we divide the entire duration into many smaller time segments. Each segment represents a short-time evolution that can be simulated with small error. The total error can be upper bounded using the triangle inequality as the sum of the Trotter errors for each segment, which we also refer to as Trotter steps.

In this section, we primarily focus on short-time evolution $\delta t$ and establish a connection between the Trotter error in a single segment and the entanglement entropy of the quantum state $\ket{\psi}$ during the evolution. Considering the target quantum evolution $U_0 = e^{-iH\delta t}$,
a $p$th-order Trotter approximation $\mathscr{U}_{p}$, and a pure quantum state $\ket{\psi}$, according to \cref{Lemma:remainder}, the Trotter error bound can be expressed as follows:
\begin{align}
 \| ( U_0-\mathscr{U}_{p})\ket{\psi} \|\le \|\sum_j E_j\ket{\psi} \| \delta t^{p+1}+ \|\mathscr{E}_{\re}\|.
\end{align}
Here $E_j$ are the leading-order local error terms and $\mathscr{E}_{\re}$ is the sum of the higher-order remaining terms, with $\mathscr{E}_{\re}=\mathcal O(\delta t^{p+2})$.
For an operator $E$ of dimension $d$, let $\|E\|$ denote the spectral norm, and let $\|E\|_F=\sqrt{\frac{\tr (EE^{\dagger})}{d} }$ be the normalized Frobenius norm. Let $S$ denote the von Neumann entropy.

\begin{lemma}\label{Lemma:remainder}For the Hamiltonian $H=\sum_l H_l$ and an input pure state $ \ket{\psi}$, the additive Trotter error of the $p$th-order product formula can be represented as
\begin{equation}\label{}
\|(U_0(\delta_t)-\mathscr{U}_p(\delta_t)) \ket{\psi}\|=\| \sum_j E_j \ket{\psi}\| \delta t^{p+1}+\|\mathscr{E}_{\re}\|.
\end{equation}
The spectral norm of the higher-order remaining term is $\|\mathscr{E}_{\re}\|=\mathcal O( \alpha_{p+2} \delta t^{p+2})$, where
\begin{equation}\label{eq:Tp}
\alpha_{p+2} :=\sum_{l_1,\dots,l_{p+2}=1}^L
\left\|[H_{l_1},[H_{l_2},\dots,[H_{l_{p+1}},H_{l_{p+2}}]]] \right\|.
\end{equation}
\end{lemma}
\begin{proof}
First define an order for pairs $(s,l)$ of stage indices $s \in \{1,\ldots,S\}$ and term indices $l \in \{1,\ldots,L\}$.
Let $(s,l)\prec (s',l')$ when $s<s'$ or $s=s',l<l'$. Furthermore, let $(s,l)\preceq (s',l')$ when $s<s'$ or $s=s',l\le l'$.
According to Theorem 3 in Ref.~\cite{childs2020theory}, the additive error $\mathscr{E} (\delta t)$
can be expressed as
\begin{equation}\label{app:aderror}
\mathscr{E}(\delta t):=U_0(\delta_t)-\mathscr{U}_p(\delta_t)=\int_0^{\delta t} \d\tau_1 \, e^{-i( \delta t-\tau_1)H} \mathscr{U}_p(\tau_1) \mathscr{N}(\tau_1),
\end{equation}
where $\mathscr{U}_p$ is the $p$th-order Trotter formula and
\begin{equation}
\begin{aligned}
 \mathscr{N}(\tau_1)&= \sum_{(s,l)}
 \rprod_{(s',l')\prec(s,l)}
e^{i\tau_1 a_s H_{\pi_{s'}(l')}} \left(a_s H_{\pi_s(l)}\right)       \lprod_{(s',l')\prec(s,l)}
e^{-i\tau_1 a_s H_{\pi_{s'}(l')}}  \\
&\quad-  \rprod_{(s',l')}
e^{i\tau_1 a_s H_{\pi_{s'}(l')}} H \lprod_{(s',l')}
e^{-i\tau_1 a_s H_{\pi_{s'}(l')}},
\end{aligned}
\end{equation}
where $\mathscr{N}(\tau_1)=\mathcal O(\tau^p_1)$.
Here we define the vector $\vec{j}_{p+1}=(j_1,j_2,\dots,j_{p+1})$ with $p+1$ entries, $j_1,j_2,\dots,j_{p+1} \in \{(s,l): s\in \{1,\dots,S\},l\in\{1,\dots,L\}\}$, and the corresponding nested commutators as
\begin{equation}\label{SM:Nnest}
 N_{\vec{j}_{p+1}}=[H_{j_1},[H_{j_2},\dots,[H_{j_p},H_{j_{p+1}}]]].
\end{equation}
According to Theorem 5 in Ref.~\cite{childs2020theory}, we can further write
\begin{align}
 \mathscr{N}(\tau_1)&=  \int_0^{\tau_1}\d\tau_2    \sum_{i=1,2} \sum_{\vec{j}_{p+1}\in \Gamma_i}   (\tau_1-\tau_2)^{q(\vec{j}_{p+1})-1}\tau_1^{p-q(\vec{j}_{p+1})} c_{\vec{j}_{p+1}} F_{\vec{j}_{p+1}}^{\dagger} N_{\vec{j}_{p+1}} F_{\vec{j}_{p+1}}, \label{eq:Nint}
 \end{align}
where we use the following definitions. Let
\begin{equation}
    \begin{aligned}
\Gamma_1&:=\{(j_1,j_2,\dots,j_{p+1}) : j_1\preceq j_2\preceq\dots \preceq j_{p+1}\},    \\
\Gamma_2&:=\{(j_1,j_2,\dots,j_{p+1}) : j_1\preceq j_2\dots \preceq j_{p}, \, j_{p+1}=(1,l_{p+1})\} .
    \end{aligned}
\end{equation}
The values $c_{\vec{j}_{p+1}}$ are real coefficients that are functions of $\vec{j}_{p+1}$ and $p$, satisfying $|c_{\vec{j}_{p+1}}|\le 1$. The function
$q(\vec{j}_{p+1})$
is the maximal number $q$ satisfying $j_1=j_2=\dots=j_q$.
The operators $F_{\vec{j}_{p+1}}$ are unitary, constructed as products of terms of the form $e^{-iH_l\tau_1}$.
Plugging \cref{eq:Nint} into \cref{app:aderror}, we have
 \begin{align}
 \mathscr{E}(\delta t)&=\int_0^{\delta t} \d\tau_1 \int_0^{\tau_1}\d\tau_2
 \sum_{i=1,2} \sum_{\vec{j}_{p+1}\in \Gamma_i} (\tau_1-\tau_2)^{q(\vec{j}_{p+1})-1}\tau_1^{p-q(\vec{j}_{p+1})} c_{\vec{j}_{p+1}}  e^{-i(\delta \delta t-\tau_1)H} \mathscr{U}_p(\tau_1)F_{\vec{j}_{p+1}}^{\dagger} N_{\vec{j}_{p+1}} F_{\vec{j}_{p+1}}.\label{Eint}
 \end{align}

Since $F_{\vec{j}_{p+1}}$ is a product of short-time evolutions, we can expand it as
\begin{align}
   F_{\vec{j}_{p+1}}^{\dagger} N_{\vec{j}_{p+1}} F_{\vec{j}_{p+1}}= N_{\vec{j}_{p+1}} + R_{\vec{j}_{p+1}}
\end{align}
where the nested commutator $N_{\vec{j}_{p+1}}$ is the $0$th-order term and $R_{\vec{j}_{p+1}}$ represents the remaining higher-order terms.
To realize this decomposition,
we repeatedly apply the formula
\begin{align}
    e^{iA\tau_1}B e^{-iA\tau_1}= B + \int_0^{\tau_1}  \d\tau_2 e^{iA(\tau_1-\tau_2)} [A, B] e^{-iA(\tau_1-\tau_2)},
\end{align}
satisfying
\begin{align}
    \| e^{iA\tau_1}B e^{-iA\tau_1}&- B\| \le \|[A,B]\|\tau_1,
\end{align}
which accounts for the effect of a short-time evolution generated by $A$ on $B$. For a product of short-time evolutions, each time we apply it to the innermost layer, giving
\begin{align}
 \| e^{iA_1\tau_1}\cdots e^{iA_{s-1}\tau_1}e^{iA_s\tau_1} B  e^{-iA_s\tau_1}\cdots e^{-iA_{s-1}\tau_1}e^{-iA_1\tau_1}- e^{iA_1\tau_1}\cdots e^{iA_{s-1}\tau_1} B e^{-iA_{s-1}\tau_1}e^{-iA_1\tau_1}\|\le \|[A_s,B]\|\tau_1.
\end{align}
Iterating, we find
\begin{align}
\left \| e^{iA_1\tau_1}\cdots e^{iA_{s-1}\tau_1}e^{iA_s\tau_1} B  e^{-iA_s\tau_1}\cdots e^{-iA_{s-1}\tau_1}e^{-iA_1\tau_1}- B \right\|\le \sum_s \|[A_s,B]\| \tau_1.
\end{align}
By replacing $B$ with $N_{\vec{j}_{p+1}} $, we thus can control the remaining term as $\|R_{\vec{j}_{p+1}}\|= \mathcal O(\sum_l \|[H_{l}, N_{\vec{j}_{p+1}}]\|\tau_1)$.

Substituting $F_{\vec{j}_{p+1}}^{\dagger} N_{\vec{j}_{p+1}} F_{\vec{j}_{p+1}}$ for $N_{\vec{j}_{p+1}}$ in $\mathscr{E}(\delta t)$ of \cref{Eint}, we define
\begin{equation}
\begin{aligned}
 \mathscr{E}^*(\delta t):=\int_0^{\delta t} \d\tau_1 \int_0^{\tau_1}\d\tau_2
 \sum_{i=1,2} \sum_{\vec{j}_{p+1}\in \Gamma_i} (\tau_1-\tau_2)^{q(\vec{j}_{p+1})-1}\tau_1^{p-q(\vec{j}_{p+1})} c_{\vec{j}_{p+1}}  e^{-i(\delta \delta t-\tau_1)H} \mathscr{U}_p(\tau_1)N_{\vec{j}_{p+1}}.
 \end{aligned}\label{Eq.additive*}
\end{equation}
The higher-order terms of the product-formula approximation are given by the difference of $\mathscr{E}$ and $\mathscr{E}^*$, which can be bounded as follows:
\begin{equation}
    \begin{aligned}
  \|   \mathscr{E}_{\re}(\delta t)\|&= \|   \mathscr{E}(\delta t)-  \mathscr{E}^*(\delta t)\|\\
  &=\left\|\int_0^{\delta t} \!\d\tau_1 \int_0^{\tau_1} \!\d\tau_2
 \sum_{i=1,2} \sum_{\vec{j}_{p+1}\in \Gamma_i} (\tau_1-\tau_2)^{q(\vec{j}_{p+1})-1}\tau_1^{p-q(\vec{j}_{p+1})} c_{\vec{j}_{p+1}}  e^{-i(\delta \delta t-\tau_1)H} \mathscr{U}_p(\tau_1)(F_{\vec{j}_{p+1}}^{\dagger} \! N_{\vec{j}_{p+1}} \! F_{\vec{j}_{p+1}} \!\! -N_{\vec{j}_{p+1}})\right\|\\
  &=\left\|\int_0^{\delta t} \d\tau_1 \int_0^{\tau_1}\d\tau_2
 \sum_{i=1,2} \sum_{\vec{j}_{p+1}\in \Gamma_i} (\tau_1-\tau_2)^{q(\vec{j}_{p+1})-1}\tau_1^{p-q(\vec{j}_{p+1})} c_{\vec{j}_{p+1}}  e^{-i(\delta \delta t-\tau_1)H} \mathscr{U}_p(\tau_1)R_{\vec{j}_{p+1}}\right\|\\\
 &\leq \int_0^{\delta t} \d\tau_1 \int_0^{\tau_1}\d\tau_2
 \sum_{i=1,2} \sum_{\vec{j}_{p+1}\in \Gamma_i} (\tau_1-\tau_2)^{q(\vec{j}_{p+1})-1}\tau_1^{p-q(\vec{j}_{p+1})} |c_{\vec{j}_{p+1}}| \left\|e^{-i(\delta \delta t-\tau_1)H} \mathscr{U}_p(\tau_1)R_{\vec{j}_{p+1}}\right\|\\
 &= \int_0^{\delta t} \d\tau_1 \int_0^{\tau_1}\d\tau_2
 \sum_{i=1,2} \sum_{\vec{j}_{p+1}\in \Gamma_i} (\tau_1-\tau_2)^{q(\vec{j}_{p+1})-1}\tau_1^{p-q(\vec{j}_{p+1})} |c_{\vec{j}_{p+1}}| \left\|R_{\vec{j}_{p+1}}\right\| \\
  &=\mathcal O\left(\sum_{l,\vec{j}_{p+1}} \|[H_{l} , N_{\vec{j}_{p+1}}]\| \delta t^{p+2}\right)\\
 &= \mathcal O\left(\sum_{l_1,\dots,l_{p+2}=1}^L
\left\|[H_{l_1},[H_{l_2},\dots[H_{l_{p+1}},H_{l_{p+2}}]]] \right\| \delta t^{p+2} \right).
\end{aligned}
\end{equation}
Here the first inequality uses the triangle inequality for the spectral norm and the fourth line uses the fact that $\|VA\|=\|A\|$ for any unitary $V$ and operator $A$.

Finally, we bound the leading term $\mathscr{E}^*(\delta t)$ defined in \cref{Eq.additive*} as follows:
\begin{equation}\label{SMleadingE}
    \begin{aligned}
\| \mathscr{E}^*(\delta t) \ket{\psi}\|&=
\left\|\int_0^{\delta t} \d\tau_1 \int_0^{\tau_1}\d\tau_2
 \sum_{i=1,2} \sum_{\vec{j}_{p+1}\in \Gamma_i} (\tau_1-\tau_2)^{q(\vec{j}_{p+1})-1}\tau_1^{p-q(\vec{j}_{p+1})} c_{\vec{j}_{p+1}}  e^{-i(\delta \delta t-\tau_1)H} \mathscr{U}_p(\tau_1)N_{\vec{j}_{p+1}} \ket{\psi}\right \|\\
& \leq \int_0^{\delta t} \d\tau_1 \left\| \int_0^{\tau_1}\d\tau_2 \sum_{i=1,2} \sum_{\vec{j}_{p+1}\in \Gamma_i}  (\tau_1-\tau_2)^{q(\vec{j}_{p+1})-1}\tau_1^{p-q(\vec{j}_{p+1})} c_{\vec{j}_{p+1}}
 N_{\vec{j}_{p+1}} \ket{\psi}\right \| \\
&=: \int_0^{\delta t} \d\tau_1 \left\|  \sum_{i=1,2} \sum_{\vec{j}_{p+1}\in \Gamma_i}  c'_{\vec{j}_{p+1}}\tau_1^{p}
 N_{\vec{j}_{p+1}} \ket{\psi}\right \|\\
 &= \left\|  \sum_{i=1,2} \sum_{\vec{j}_{p+1}\in \Gamma_i} \frac1{p+1} c'_{\vec{j}_{p+1}} \delta t^{p+1}  N_{\vec{j}_{p+1}} \ket{\psi}\right \|\\
 &= \left\|\sum_j E_j \ket{\psi}\right \|\delta t^{p+1}.
 \end{aligned}
\end{equation}
Here the second line uses the triangle inequality. In particular, for a given $\tau_1$, one can eliminate the unitary $e^{-i(\delta \delta t-\tau_1)H} \mathscr{U}_p(\tau_1)$, since $\|{VA\ket{\psi}}\|=\|{A\ket{\psi}}\|$ for unitary $V$.
For the third line, denote the result of integrating over $\tau_2$ as $F(\tau_1,\vec{j}_{p+1}):=\int_0^{\tau_1}\d\tau_2 (\tau_1-\tau_2)^{q(\vec{j}_{p+1})-1}\tau_1^{p-q(\vec{j}_{p+1})} c_{\vec{j}_{p+1}}$. It is straightforward to check that, for all $\vec{j}_{p+1}$, $F(\tau_1,\vec{j}_{p+1})=c'_{\vec{j}_{p+1}}\tau_1^{p+1}$ for some constant $c'_{\vec{j}_{p+1}}$. In the fourth line, we simply integrate $\tau_1^p$. In the last line, we decompose $\frac1{p} \sum c'_{\vec{j}_{p+1}} N_{\vec{j}_{p+1}} $ into local operators $E_j$, completing the proof.
\end{proof}

\begin{lemma}\label{SMLemma:local}
(Lemma 1 in Methods) Let $E=\sum_j E_j$ act on $N$ qubits, where $E_j$ acts nontrivially on the subsystem with $\support(E_j)$. Then
\begin{align}
  | \bra{\psi}E^{\dagger}E \ket{\psi}|\le \|E\|_F^2+    \sum_{j,j'} \|E_j^{\dag} E_j'\| ~\tr|\rho_{j,j'}- \mathbb{I}_{\support(E_j^{\dag}E_{j'})}/d_{\support(E_j^{\dag}E_{j'})} |,
\end{align}
where $\|E\|^2_F:= \tr(E^{\dagger}E) /d$ is the (square) of the normalized Frobenius norm, and $\rho_{j,j'}:=\tr_{[N]\setminus \support(E_j E_j')}(\ket{\psi}\bra{\psi})$ is the RDM of $\ket{\psi}\bra{\psi}$ on the subsystem of $\support(E_jE_j')$.
\end{lemma}

\begin{proof}
The term $E_j^{\dag}E_{j'}$ in the expression for $E^{\dag}E$ only acts nontrivially on $\support(E_j^{\dag}E_{j'})$. We denote its nontrivial part by $L_{j,j'}:=\tr_{[N]\setminus \support(E_j E_j')} (E_j^{\dag}E_{j'})$. Since $2^{-N}\tr(E_j^{\dagger} E_{j'})=d_{\support(E_j^{\dag}E_{j'})}^{-1}\tr(L_{j,j'})$, we have
\begin{equation}
\begin{aligned}
\bra{\psi}E_j^{\dagger}E_{j'} \ket{\psi}&=\tr(L_{j,j'}\rho_{j,j'})=\tr[L_{j,j'}(\rho_{j,j'}- \mathbb{I}_{\support(E_j^{\dag}E_{j'})}/d_{\support(E_j^{\dag}E_{j'})})]+\tr(L_{j,j'})/d_{\support(E_j^{\dag}E_{j'})}\\
&= \tr[L_{j,j'}(\rho_{j,j'}- \mathbb{I}_{\support(E_j^{\dag}E_{j'})}/d_{\support(E_j^{\dag}E_{j'})})]+\tr(E_j^{\dagger}E_{j'})/2^N\\
&\leq \|L_{j,j'}\|\tr|\rho_{j,j'}- \mathbb{I}_{\support(E_j^{\dag}E_{j'})}/d_{\support(E_j^{\dag}E_{j'})}|+\tr(E_j^{\dagger}E_{j'})/2^N\\
&= \|E_j^{\dag}E_{j'}\| \tr|\rho_{j,j'}- \mathbb{I}_{\support(E_j^{\dag}E_{j'})}/d_{\support(E_j^{\dag}E_{j'})}|+\tr(E_j^{\dagger}E_{j'})/2^N,
\end{aligned}
\end{equation}
and the result follows by summing the indices $j,j'$.
\end{proof}

\begin{theorem}\label{SMthm:main}
(Theorem 1 in the main text) For a given pure quantum state $\ket{\psi}$ and perfect quantum evolution $U_0=e^{-iH\delta t}$ with $p$th-order Trotter approximation $\mathscr{U}_{p}$, the error in a Trotter step of duration $\delta t$ has the upper bound
\begin{equation}\label{SMeq:ErrD1}
\begin{aligned}
    \|(U_0-\mathscr{U}_{p})\ket{\psi}\|
    &=\mathcal O\left(\sqrt{\sum_{j,j'} \|E_j^{\dag}E_{j'}\| \tr|\rho_{j,j'}-\mathbb{I}/d_{\rho_{j,j'}}|}\,\delta t^{p+1}+ \|E\|_F \delta t^{p+1}  +\alpha_{p+2} \delta t^{p+2} \right).
\end{aligned}
\end{equation}
where $\alpha_{p+2}$ is defined in \cref{eq:Tp}.
We call this the \emph{distance-based error bound}.
One can further relate the Trotter error to the entanglement entropy of subsystems with the bound
\begin{equation}\label{SMeq:ErrEnt}
\|(U_0-\mathscr{U}_{p})\ket{\psi}\|=\mathcal O\left(\delta t^{p+1}\sqrt{\sum_{j,j'} \|E_j^{\dag}E_{j'}\|\sqrt{\log (d_{\rho_{j,j'}})-S(\rho_{j,j'})}}+ \delta t^{p+1} \|E\|_F+\alpha_{p+2} \delta t^{p+2}\right).
\end{equation}
We call this the \emph{entanglement-based error bound}.
\end{theorem}

Note that compared to the statement of Theorem 1 in the main text, here we also include an explicit upper bound for the high-order terms $\mathscr{E}_{\re}$.

\begin{proof}
According to \cref{Lemma:remainder}, the higher-order terms can be bounded by $\|\mathscr{E}_{\re}\|=\mathcal O(\alpha_{p+2} \delta t^{p+2})$.
The leading term has the bound
\begin{equation}\label{EqSMTh1-1}
\begin{aligned}
    \|(U_0-\mathscr{U}_p)\ket{\psi}\|
    &\le \left \| \sum_{j} E_j \ket{\psi}\right \| \delta t^{p+1}  + \|\mathscr{E}_{\re}\|\\
    &=  \delta t^{p+1}   \sqrt{\sum_{j,j'} \bra{\psi}E_j^{\dagger} E_{j'}\ket{\psi} }  + \|\mathscr{E}_{\re}\|       \\
    &\le \delta t^{p+1}  \sqrt{\|E\|_F^2+    \sum_{j,j'} \|E_j^{\dagger} E_{j'}\| \,\tr|\rho_{j,j'}- \mathbb{I}/d_{\rho_{j,j'}} | }+  \|\mathscr{E}_{\re}\|       \\
    &\le \delta t^{p+1}  \sqrt{  \sum_{j,j'} \|E_j^{\dagger} E_{j'}\| \,\tr|\rho_{j,j'}- \mathbb{I}/d_{\rho_{j,j'}} | }+
   \delta t^{p+1}   \|E\|_F+  \|\mathscr{E}_{\re}\|.
   \end{aligned}
\end{equation}
Here the third line uses \cref{SMLemma:local} and the fourth line uses $\sqrt{x_1+x_2}\leq \sqrt{x_1}+\sqrt{x_2}$.
Since $\|\mathscr{E}_{\re}\|=\mathcal O(\delta t^{p+2})$, for a sufficiently small $\delta t$, we have
\begin{align}
    \|(U_0-\mathscr{U}_p)\ket{\psi}\|= \mathcal O(\delta t^{p+1}  \sqrt{  \sum_{j,j'} \|E_j^{\dagger} E_{j'}\| \, \tr|\rho_{j,j'}- \mathbb{I}/d_{\rho_{j,j'}} | }+
   \delta t^{p+1}   \|E\|_F+\alpha_{p+2} \delta t^{p+2}),
\end{align}
which gives the distance-based error bound. Moreover, the trace distance of $\rho_{j,j'}$ and $\mathbb{I}/d_{\rho_{j,j'}}$ can be bounded using the relative entropy as
\begin{align}
    \tr|\rho_{j,j'}- \mathbb{I}/d_{\rho_{j,j'}} | \le \sqrt{2S(\rho_{j,j'}\|\mathbb{I}/d_{\rho_{j,j'}}) }=
    \sqrt{2\log(d_{\rho_{j,j'}})-2S(\rho_{j,j'})},
\end{align}
which gives the entanglement-based bound.
\end{proof}

A pure state of $N$ parties is called $\Delta$-approximate $k$-uniform if all $k$-party reduced density matrices are close to maximally mixed, i.e., $\|\tr_{[N]\setminus[k]}(\ket{\psi}\bra{\psi})-\mathbb{I}/2^k\|_1\le \Delta$.
\begin{corollary}\label{SMcor:kuniform}
(Corollary 1 in the main text) For a $\Delta$-approximate $k$-uniform pure quantum state $\ket{\psi}$ with $\sqrt{\Delta}\le \|E\|_F/\sum_j\|E_j\|$ and $k\ge 2\max_j w(E_j)$, the Trotter error satisfies
\begin{equation}
\|(U_0-\mathscr{U}_p)\ket{\psi}\|= \mathcal O\left(\|E\|_F\delta t^{p+1}\right).
\end{equation}
\end{corollary}
\begin{proof}
    According to the proof of \cref{SMthm:main}, we have
\begin{align}
\|(U_0-\mathscr{U}_p)\ket{\psi}\|=\mathcal O   \left(\delta t^{p+1}  \sqrt{  \sum_{j,j'} \|E_j^{\dagger} E_{j'}\| \,\tr|\rho_{j,j'}- \mathbb{I}/d_{\rho_{j,j'}} | }+
   \delta t^{p+1}   \|E\|_F\right).
\end{align}
    For a $\Delta$-approximate $k$-uniform state, the trace distance part is bounded by
\begin{align}
  \sqrt{\sum_{j,j'} \|E_j^{\dagger} E_{j'}\| \,\tr|\rho_{j,j'}- \mathbb{I}/d_{\rho_{j,j'}} |} \le \sqrt{\sum_{j,j'} \|E_j^{\dagger}\| \|E_{j'} \| \Delta}=\sum_{j} \|E_j\|\sqrt{\Delta} \le \|E\|_F.
\end{align}
   Thus,  we have $\|(U_0-\mathscr{U}_p)\ket{\psi}\|= \mathcal O\left(\|E\|_F\delta t^{p+1}\right).$
\end{proof}

Similarly, one can show that when the entanglement entropy of the subsystem $\support(E_j^{\dag} E_{j'})$ satisfies $S(\rho_{j,j'})\ge w(E_jE_{j'})-(\frac{\|E\|_F}{\sum_j\|E_j\|})^4$, the error scales like $\mathcal O(\delta t^{p+1}\|E\|_F)$, recovering the average-case Trotter error bound~\cite{Zhaorandom22}.

\section{Product states with worst-case error scaling}

\subsection{Proof of Theorem 2}

We begin with a simple observation about the signs of the coefficients of nested commutators for product formulas.

\begin{fact}
For a $p$th-order product formula, the leading terms of the Trotter error can be written in the form
\begin{equation}\label{SM:Epauli}
E = e^{i\theta}\left(\sum_j a_j P_j + \sum_k b_k Q_k\right),
\end{equation}
where $\theta \in \{0,\pi/2,\pi,3\pi/2\}$, $P_j$ and $Q_k$ are Pauli operators, and the coefficients $a_j,b_k$ are real numbers satisfying $a_j>0$ and $b_k<0$.
\end{fact}

\begin{proof}
As shown in \cref{SMleadingE} in the proof of \cref{Lemma:remainder}, the error term $E$ is a linear combination of the $(p+1)$st-order nested commutators defined in \cref{SM:Nnest} (ignoring the global phase $i^{p+1}$). We prove by induction that each commutator $[H_{l_1},[H_{l_2},\dots[H_{l_{p}},H_{l_{p+1}}]]]$ is in the form of \cref{SM:Epauli}. Suppose that the $p'$th commutator $N_{p'}=e^{i\theta}\sum_W \alpha_W$ is in this form (where $\alpha_W$, the coefficient of the Pauli operator $W$, can be either positive or negative), and we add one more layer with $H_{l}=\sum_{W'}\beta_{W'}W'$, where $\beta_{W'}$ is the coefficient of the Pauli operator $W'$. Then the new commutator is
\begin{equation}
\begin{aligned}
 N_{p'+1}=[H_{l}, N_{p'}]=e^{i\theta} \sum_W \sum_{W'} \alpha_W\beta_{W'} [W',W]=e^{i\theta} \sum_W \sum_{W'} \alpha_W\beta_{W'} (\pm 2iW'').
 \end{aligned}
\end{equation}
For every pair of Pauli operators $W,W'$ in the sum with $[W',W]\neq 0$, we have $\{W',W\}=0$, so $[W',W]=\pm 2iW''$ for some new Pauli operator $W''$. This introduces a phase $\pm i$ for each term,
so we still extract an overall phase $e^{i\theta}$ with $\theta=\{0,\pi/2,\pi,3\pi/2\}$ for $N_{p'+1}$. It is easy to see that $N_{p=1}=H_l$ satisfies \cref{SM:Epauli}, and thus the claim follows by induction.
\end{proof}

Now we begin the proof of Theorem 2 from main text. Note that the global phase $e^{i\theta}$ cancels out in $E^{\dag}E$ for the error analysis, so we neglect this phase in the following discussion.

\begin{theorem}\label{SMthm:worstcase}
(Theorem 2 in the main text) Consider a Trotter approximation of an $N$-qubit lattice Hamiltonian $H$. Let the leading term of the error have Pauli decomposition $E =\sum_j a_j P_j + \sum_k b_k Q_k$  with $a_j>0$ and $b_k<0$ (ignoring a global phase). If
$\sum_j a_j=\Theta(N)$ and $\sum_k |b_k|=o(N)$, there exists a product state $\ket{\psi}$ that achieves the worst-case error scaling,
\begin{equation}
 \|(U_0-\mathscr{U}_p)\ket{\psi}  \|=\Theta(N\delta t^{p+1}).
\end{equation}
\end{theorem}

\begin{proof}
We choose a maximal index set, denoted by $\mathrm{\mathrm{Stab}}$, such that for any $j, j'\in \mathrm{Stab}$,  $\support (P_j) \cap \support (P_{j'}) =\emptyset$. Because the weight $w(P_j)$ is constant,
$|\mathrm{Stab}|= \Theta (N)$.

Consider each Pauli string $P_j=\sigma_{j_1} \otimes \cdots \otimes \sigma_{j_{w(P_j)}}$ in the decomposition of $E$.
Then we construct the state with a density matrix
    \begin{align}
       \ket{\psi} \bra{\psi}=
       \bigotimes_{j\in \mathrm{Stab}} \bigotimes_{k\in \{1,2,\dots,w(P_j)\}}
        \frac{\mathbb{I}_{2}+ \sigma_{j_k}}{2} \bigotimes_{k'\in [N]\setminus \cup\support (P_j)}  \frac{\mathbb{I}_{2}+ Z_{k'}}{2}.
    \end{align}
This is indeed a product state, and for any $j\in \mathrm{Stab}$, $ P_j\ket{\psi}=\ket{\psi}$. Note that the union of the $\support (P_j)$ for $j\in \mathrm{Stab}$ is a proper subset of the whole system $[N]$, and we take $(\mathbb{I}_{2}+ Z)/2$
for the remaining qubits. In addition,
for any Pauli string $P$, $\tr (P \ket{\psi} \bra{\psi})\in\{0,1\}$, which can be seen as follows. The state $\ket{\psi} \bra{\psi}$ is a linear combination of $2^n-1$ distinct non-identity Pauli operators, and $\tr (P \ket{\psi} \bra{\psi})=1$ if $P$ is one of those Paulis. For any Pauli operator $P$ that does not appear in the sum, $\tr (P \ket{\psi} \bra{\psi})=1$ since $\tr (PQ)=0$ for distinct Paulis $P,Q$. In this way, we write down the leading term of the Trotter error as
\begin{equation}\label{SMeqEE}
\begin{aligned}
\tr (E^{\dagger}E \ket{\psi} \bra{\psi})&=
   \sum_{j,j'} a_j a_{j'}\tr (P_j P_{j'} \ket{\psi} \bra{\psi})  +
   \sum_{k,k'} b_k b_{k'}\tr (Q_k Q_{k'} \ket{\psi} \bra{\psi})\\
  &\quad+ \sum_{j,k'} a_j b_{k'}\tr (P_j Q_{k} \ket{\psi} \bra{\psi})
   + \sum_{k,j'} b_k a_{j'}\tr (Q_k P_{j'}  \ket{\psi} \bra{\psi}).
\end{aligned}
\end{equation}
Here in the first summation, if both $j,j'\in \mathrm{Stab}$, we have $\tr (P_j P_{j'} \ket{\psi} \bra{\psi})=1$ by construction. Therefore $\sum_{j,j' \in \mathrm{Stab}} a_j a_{j'} \tr(P_j P_{j'} \ket{\psi}\bra{\psi}) = \left( \sum_{j\in \mathrm{Stab}} a_j\right)^2=\Theta (|\mathrm{Stab}|^2)=\Theta (N^2)$.

It remains to consider the contributions of the other cases.
For $j\in \mathrm{Stab}$ and $j'\notin \mathrm{Stab}$, we have
\begin{equation}\label{SM:crossStab}
\begin{aligned}
\sum_{j\in \mathrm{Stab},j'\notin \mathrm{Stab}} 2a_j a_{j'}\tr ((P_j P_{j'} +P_{j'} P_{j})\ket{\psi} \bra{\psi}),
\end{aligned}
\end{equation}
whose individual terms are zero unless $[P_j, P_{j'}]=0$.
In that case, we obtain another Pauli operator $P'=P_jP_{j'}$ with a possible minus sign.
A minus occurs only if $\support (P_j) \cap \support (P_{j'}) \neq \emptyset$. Otherwise, the contribution is positive or zero, depending on whether $P'$ stabilizes $\ket{\psi}\bra{\psi}$ or not. As $\support (P_j)$ is a constant, the number of overlapping terms $P_{j'}$ is $\mc O(1)$. As a result, the summation of \cref{SM:crossStab} can reduce the whole summation in \cref{SMeqEE} by only $\mc O(N)$. Similar arguments handle both the case $j,j'\notin \mathrm{Stab}$, and also the second summation in \cref{SMeqEE}, which both give at most a $\mc O(N)$ reduction. For the last two summations in \cref{SMeqEE}, since $a_jb_{k'}$ and $b_{k}a_{j'}$ are negative, and the terms in such summations are $o(N^2)$ by the assumption that $\sum_j a_j=\Theta(N)$ and $\sum_k |b_k|=o(N)$, the reduction is $o(N^2)$.
Overall, this shows that $\tr (E^{\dagger}E \ket{\psi} \bra{\psi})=\Theta (N^2)$.

Using the triangle inequality, we can lower bound the overall error as
\begin{equation}
\begin{aligned}
    \|(U_0-\mathscr{U}_p)\ket{\psi}\|&\ge \sqrt{\tr (E^{\dagger} E\ket{\psi} \bra{\psi}) } -\|\mathscr{E}_{\re}\| \\
    &=  \Theta (N) \delta t^{p+1} -\|\mathscr{E}_{\re}\|
    =\Omega(N) \delta t^{p+1}.
\end{aligned}
\end{equation}
The last line follows because $\|\mathscr{E}_{\re}\|=\mathcal{O} (N\delta t^{p+2})$ according to \cref{Lemma:remainder}.
Combining this with the worst-case upper bound $\|(U_0-\mathscr{U}_p)\ket{\psi}\|=\mathcal O(N\delta t^{p+1})$,
we conclude that $\|(U_0-\mathscr{U}_p)\ket{\psi}\|=\Theta(N\delta t^{p+1})$.
\end{proof}

\subsection{Ising models}
The condition of this theorem is not particularly special and is satisfied by some well-studied models.
For instance, for the QIMF model that we consider in the main text and the PF1 method, we have
\begin{equation}
\begin{aligned}
 E= [iA, iB]&= -i \left[2h_xh_y  \sum_{j=1}^N Z_j  +2 Jh_y \sum_{j=1}^{N-1} (Z_jX_{j+1}+ X_jZ_{j+1})\right].\\
\end{aligned}
\end{equation}
After ignoring the global phase $-i$, we have $\sum_j a_j=\Theta (N)$ and  $\sum_k |b_k|=0$, since there is no term with a negative coefficient. Thus we can choose $P_j=Z_j$ for $j=[N]$ as the $\mathrm{Stab}$ set, and
choose $\ket{\psi}=\ket{0}^{\otimes N}$.

Furthermore, we can find product states that achieve the worst-case error even in some cases that do not satisfy the conditions of \cref{SMthm:worstcase}. For example, consider the PF2 approximation for the QIMF model. By the third line of \cref{Eq:PF2additive} and also \cref{SMeq:commutators}, we have
\begin{equation}
\begin{aligned}
E&=\left[i\frac{A}{2},\left[i\frac{A}{2},iB\right]\right]+\left[-iB,\left[-iB,-i\frac{A}{2}\right]\right]
  \\
&=-i\Bigg[ h_x^2h_y    \sum_{j=1}^N Y_j+
  J^2 h_y \sum_{j=1}^{N-1} Y_j+
   J^2 h_y \sum_{j=2}^{N} Y_j+
  2Jh_xh_y \sum_{j=1}^{N-1} (Y_jX_{j+1}+ X_jY_{j+1}) \\
  &\qquad\quad+ 2J^2 h_y \sum_{j=1}^{N-2} (X_jY_{j+1}X_{j+2})
  -2 h_xh_y^2  \sum_{j=1}^N X_j  +4Jh_y^2 \sum_{j=1}^{N-1} ( Z_jZ_{j+1}-X_jX_{j+1})\Bigg].
\end{aligned}
\end{equation}
Ignoring the global phase $-i$, the summation
\begin{equation}
\begin{aligned}
-2 h_xh_y^2  \sum_{j=1}^N X_j -4Jh_y^2 \sum_{j=1}^{N-1}X_jX_{j+1}
\end{aligned}
\end{equation}
has a negative sign, and in total the summation has $\mc O (N)$ terms. Thus the assumptions of \cref{SMthm:worstcase} do not hold. However, there is still a product state $\ket{\psi}=\ket{0}^{\otimes N}$ that achieves the worst-case error. The term in the third line $\sum_{j=1}^{N-1} Z_jZ_{j+1}$ gives a positive contribution of order $\Theta (N^2)$ in $E^\dag E$.
For the other terms, according to the proof of \cref{SMthm:worstcase}, now we only need to guarantee that the cross terms (between the terms with positive and negative coefficients) can reduce the error by at most $o(N^2)$. Indeed, such terms  always generate Pauli operators $P'$ that have Pauli $X$ or $Y$ operators on some qubits, so that $\tr (P'\ket{\psi} \bra{\psi})=0$.

\subsection{Heisenberg models}
We can also consider the one-dimensional Heisenberg model with a random magnetic field $h_{j}\in [-1,1]$ at each site $j \in \{1,\ldots,N\}$. The Hamiltonian shows
\begin{equation}\label{Eq:heisenberg}
    H = \sum^{N-1}_{j=1} \left(X_j X_{j+1} + Y_jY_{j+1} + Z_j Z_{j+1}\right)+ \sum^{N}_{j=1}h_{j}Z_{j}.
\end{equation}
For convenience, suppose $N$ is even.
The summands of this Hamiltonian can be partitioned into two groups in an even-odd pattern \cite{Childs2019Product}, giving $H=A+B$ with
\begin{equation}
    \begin{aligned}
    A &= \sum^{\frac{N}{2}}_{j=1} \left(X_{2j-1} X_{2j} + Y_{2j-1} Y_{2j} + Z_{2j-1}  Z_{2j}+ h_{2j-1}Z_{2j-1}\right) ,\\
    B &= \sum^{\frac{N}{2}-1}_{j=1} \left(X_{2j} X_{2j+1} + Y_{2j} Y_{2j+1} + Z_{2j} Z_{2j+1}  +h_{2j}Z_{2j}\right) +h_{N}Z_{N}.
\end{aligned}
\end{equation}
For the PF1 method, the error term is
\begin{align}
 iE= i[iA, iB]= &\sum_{j=1}^{N-2} (-1)^j\left(X_jY_{j+1}Z_{j+2}-X_jZ_{j+1}Y_{j+2}-Y_jZ_{j+1}X_{j+2}+Y_jX_{j+1}Z_{j+2}+
 Z_jX_{j+1}Y_{j+2}-Z_jY_{j+1}X_{j+2}\right)\nonumber\\
 &+\sum_{j=1}^{N-1} (-1)^j\left(h_{j+1}(Y_jX_{j+1}-X_jY_{j+1})\right).
 \label{SMEHeisenberg}
\end{align}
It is not hard to see that this commutator does not meet the conditions of \cref{SMthm:worstcase}. However, we can still find a product state that gives the worst-case performance.
We choose the $\mathrm{Stab}$ set containing $X_jY_{j+1}$ for odd $j$.
Then the corresponding product state is stabilized by $\{X_1,Y_2,X_3,Y_4,\ldots,X_{N-1},Y_N\}$, giving $\ket{\psi}=\ket{+}_1\ket{+i}_2\ket{+}_3\ket{+i}_4\cdots\ket{+}_{N-1}\ket{+i}_{N}$. Now consider the contributions of various terms in \cref{SMEHeisenberg} to $\tr (E^{\dagger}E \ket{\psi} \bra{\psi})$. Besides $X_jY_{j+1}$ for odd $j$, the terms $Y_jX_{j+1}$ for even $j$ also stabilize $\ket{\psi}$ (and of course the same holds for products of such terms). Letting $M=\sum_{j~\mathrm{odd}}h_{j+1} X_jY_{j+1}+\sum_{j'~\mathrm{even}} h_{j'+1}Y_{j'}X_{j'+1}$, we have
\begin{equation}\label{SMEHeisenberg2}
\begin{aligned}
 \tr (M^{\dag}M\ket{\psi} \bra{\psi})=\sum_{j,k=1}^{N-1}h_{j+1}h_{k+1}=\left(\sum_{j=1}^{N-1}h_{j+1}\right)^2.
\end{aligned}
\end{equation}
If $h_j=h$ for all $j$, we have $\tr (M^{\dag}M\ket{\psi} \bra{\psi})=(N-1)^2h^2$, which is $\Theta(N^2)$ for constant $h$. Then we just need to check that the contributions of other terms are at most $\mc O(N)$. As mentioned in \cref{SM:crossStab} in the proof of \cref{SMthm:worstcase}, any pair of Pauli operators $P,P'$ in the expression for $E$ in \cref{SM:Epauli} can survive in $E^{\dag}E$ only when they commute. Since $\support(P)=\mc O (1)$ and the number of $P$ in $E$ is $\mc O (N)$, the number of overlapping and also commuting Pauli pairs is $\mc O(N)$. As a result, we only need to focus on Pauli pairs whose supports are disjoint (and of course commuting), since there are $\mc O(N^2)$ such cases.
For all these disjoint Pauli pairs, if both $P,P'$ are in $M$ and stabilize the state $\ket{\psi}$ given by \cref{SMEHeisenberg2}, the result is nonzero. Otherwise, if $P,P'$ are not both in $M$, then
$\tr (PP'\ket{\psi} \bra{\psi})=0$
since $PP'$ does not stabilize $\ket{\psi}$,
so such terms do not contribute to $\tr (E^{\dag}E\ket{\psi} \bra{\psi})$.
For example, for $P=X_1Z_{2}Y_{3}$ and $P'=Y_6X_7$, it is clear that $\tr (PP'\ket{\psi} \bra{\psi})=0$.

Moreover, if the values $h_j$ are independent and identically distributed, with mean zero and variance $\sigma^2$, then the previous main contribution term in \cref{SMEHeisenberg2} now becomes $(N-1)\sigma^2$. As a result, for constant $\sigma$, the contribution is only $\Theta(N)$ in this case. In summary, we show that for the Heisenberg model with a uniform magnetic field, we can construct a product state to reach the worst-case error. However, this construction does not work for the random-field case. We leave it as an open question to construct a product state that achieves worst-case scaling in this case, or to show that no such state exists.

\section{Proof of Corollary 2}

\begin{corollary}
\label{SMcor:shallow}
  (Corollary 2 in the main text)  For an $N$-qubit lattice Hamiltonian and a pure quantum input state $\ket{\psi}$ that is generated by a $D$-dimensional geometrically local circuit of depth $o(N^{1/D})$, the entanglement-based bound of  is
    \begin{align}
  \|{(U_0-\mathscr{U}_p)\ket{\psi}}\|
  = \mathcal O(\delta t^{p+1}N\max_{j}(\log d_{\rho_{j}}-S(\rho_{j}))^{1/4})+ O(\delta t^{p+1}\sqrt{N}).
    \end{align}
\end{corollary}
\begin{proof}
According to the entanglement-based error bound in \cref{SMthm:main}, we need to calculate $S(\rho_{j,j'})$ for each $j$ and $j'$ with
\begin{align}
    S(\rho_{j,j'})=S(\rho_{j})+S(\rho_{j'})- I(j,j')_{\rho_{j,j'}},
\end{align}
where $I(j,j')_{\rho}$ is the mutual information for the subsystems on $\support(E_{j})$ and $\support(E_{j'})$.
Consider a geometrically local quantum circuit where the locality of the interactions is represented by a graph whose vertices are qubits and whose edges represent the available interactions. For two distinct qubits $u$ and $v$,
the distance between two qubits $u$ and $v$ is the length of the shortest path between them in $G$. For two sets of qubits, $\mathscr S_1$ and $\mathscr S_2$, their distance is defined as
\begin{align}
\Dist(\mathscr S_1,\mathscr S_2)= \min_{u\in \mathscr S_1, v\in \mathscr S_2}\Dist(u, v).
\end{align}

Here we define the light-cone region $\mathcal{L}$
according to the distance of the supports as follows:
\begin{align}
  \mathcal{L}=\{(j,j')\mid\Dist\{\support(E_{j}),\support(E_{j'})\}\le  2\,\mathrm{depth}\}.
\end{align}
For a given $E_j$, there are only a limited number of $E_{j'}$ that are close to $\support(E_j)$. We denote the maximal number of terms $E_{j'}$  in the light cone as
\begin{align}
L:=\max_j \{\#\{j'\}\mid(j,j')\in \mathcal{L}\}.
\end{align}
For a given $E_j$, there are at most $\mathcal O(\mathrm{depth}^D )= o(N)$ qubits in the light cone. Because the evolved Hamiltonian $H$ is a lattice model, each $E_j$ in the Trotter error commutators only has constant weight. A fixed qubit could overlap at most a constant number of $E_j$.
As a result, we have $L=o(N)$.

When $(j,j')\notin \mathcal{L}$, i.e., the two subsystems on $\support(E_j)$ and $\support(E_j')$ lie outside the light cone and the RDM $\rho_{j,j'}=\rho_{i}\otimes \rho_{j}$ is a tensor product with $I(j,j')_{\rho}=0$.
In this way, we can refine the first term in the entanglement-based error bound in \cref{SMthm:main} as
\begin{equation}
\begin{aligned}
&\sum_{j,j'} \|E_j^{\dagger} E_{j'}\|\sqrt{\log d_{\rho_{j,j'}}-S(\rho_{j,j'})}\\
&=\sum_{(j,j')\in \mathcal{L}} \|E_j^{\dagger} E_{j'}\|\sqrt{\log d_{\rho_{j,j'}}-S(\rho_{j,j'})}+ \sum_{(j,j')\notin \mathcal{L}} \|E_j^{\dagger} E_{j'}\|\sqrt{\log d_{\rho_{j}}+ \log d_{\rho_{j'}} -S(\rho_{j})-S(\rho_{j'})},\\
&=  \mathcal O(LN) + \mathcal O(N(N-L)) \max_{j}(\log d_{\rho_{j}}-S(\rho_{j}))^{1/2})\\
&= \mathcal O(N^2 \max_{j}(\log d_{\rho_{j}}-S(\rho_{j}))^{1/2}).
\end{aligned}
\end{equation}
For a lattice Hamiltonian, $\|E\|_F=\mathcal O(\sqrt{N})$. Substituting these results in \cref{SMthm:main}, the corollary follows.
\end{proof}

\section{Concrete upper bounds for PF1 and PF2}
\begin{lemma}\label{Lemma:traceproduct}
Let $M,U \in \mathbb{C}^{d\times d}$ with $U$ unitary.
Then
\begin{equation}
|\tr(UMM^{\dag})|\leq |\tr(MM^{\dag})|
\end{equation}
\end{lemma}
\begin{proof}
The Cauchy-Schwarz inequality gives  $|\tr(MN)|\leq \sqrt{\tr(MM^{\dag})}\sqrt{\tr(NN^{\dag})}$.
Using this inequality, we have
\begin{equation}
   |\tr(UMM^{\dag})|\leq \sqrt{\tr(UMM^{\dag}U^{\dag})}\sqrt{\tr(MM^{\dag})}=|\tr(MM^{\dag})|,
\end{equation}
because $U$ is unitary.
\end{proof}

\begin{lemma}\label{SMLemma:AB}
(PF1, Lemma 2 in Methods)
For a two-term Hamiltonian $H=A+B$, consider the first-order product formula $\mathscr{U}_1(\delta t)=e^{-iA\delta t}e^{-iB\delta t}$ with initial state $\ket{\psi}$. Let $E=[A,B]=\sum_j E_j$. Then the Trotter error is upper bounded as
\begin{equation}
    \|(\mathscr{U}_1(\delta t)-U_0(\delta t))\ket{\psi}\|\le \sqrt{\frac{\delta t^4}{4}(\|[A,B]\|_F^2+\Delta_E(\psi)}+ \frac{\delta t^3}{6}\|[A,[A,B]]\|+  \frac{\delta t^3}{3}\|[B,[B,A]]\|,
\end{equation}
with
\begin{equation}
\Delta_E(\psi)=\sum_{j,j'}\|E_{j'}E_j\| \tr|\rho_{j,j'}-\mathbb{I}_{\support(E_jE_{j'})}/d_{\support(E_jE_{j'})}|.
\end{equation}
For sufficiently small $\delta t$, the error is $\mathcal O(\delta t^2\sqrt{\|[A,B]\|_F^2+\Delta_E(\psi}))$.
\end{lemma}

\begin{proof}
By Ref.~\cite{childs2020theory}, the additive error $\mathscr{E}(\delta t)=\mathscr{U}_1(\delta t)-U_0(\delta t)$ is
\begin{equation}
\begin{aligned}\label{Mhigh}
\mathscr{E}(\delta t)&=\int_0^{\delta t} \d\tau_1 \,  e^{-iH(\delta t-\tau_1)}e^{-iA\tau_1}e^{-iB\tau_1}
\left(e^{iB\tau_1}e^{iA\tau_1}iBe^{-iA\tau_1}e^{-iB\tau_1}    -iB\right)\\
&=\int_0^{\delta t} \d\tau_1 \, e^{-iH(\delta t-\tau_1)}e^{-iA\tau_1}e^{-iB\tau_1} [iA,iB]\tau_1 \\
&\quad+
\int_0^{\delta t} \d\tau_1
\int_0^{\tau_1} \d\tau_2 \,  e^{-iH(\delta t-\tau_1)}e^{-iA\tau_1}e^{-iB\tau_1} e^{iB(\tau_1-\tau_2)}[iB,[iA, iB]]e^{-iB(\tau_1-\tau_2)}   \tau_1 \\
&\quad+
\int_0^{\delta t} \d\tau_1  \int_0^{\tau_1} \d\tau_2 \, e^{-iH(\delta t-\tau_1)}e^{-iA\tau_1}e^{-iB\tau_1} e^{iB\tau_1}e^{iA(\tau_1-\tau_2)} [iA,[iA,iB]]e^{-iA(\tau_1-\tau_2)}   e^{-iB\tau_1}  \tau_2\\
&= \mathscr{E}_1+ \mathscr{E}_{\re},
\end{aligned}
\end{equation}
where $\mathscr{E}_1=\int_0^{\delta t} \d\tau_1 \, e^{-iH(\delta t-\tau_1)}e^{-iA\tau_1}e^{-iB\tau_1} [iA,iB]\tau_1$ and $\mathscr{E}_{\re}$ is the sum of the other integrals.
Here the second line is due to the equation
\begin{equation}
\begin{aligned}
 e^{iB\tau_1}e^{iA\tau_1}iBe^{-iA\tau_1}e^{-iB\tau_1}
&=  iB+e^{iB\tau_1}[iA, iB] \tau_1 e^{-iB\tau_1} + \int_0^{\tau_1} \d\tau_2  e^{iB\tau_1}e^{iA(\tau_1-\tau_2)} [iA,[iA,iB]] \tau_2 e^{-iA(\tau_1-\tau_2)}   e^{-iB\tau_1}.\\
&=iB+[iA, iB] \tau_1  +
\int_0^{\tau_1} \d\tau_2  e^{iB(\tau_1-\tau_2)}[iB,[iA, iB]]e^{-iB(\tau_1-\tau_2)}   \tau_1 \\
&\quad+\int_0^{\tau_1} \d\tau_2  e^{iB\tau_1}e^{iA(\tau_1-\tau_2)} [iA,[iA,iB]] \tau_2 e^{-iA(\tau_1-\tau_2)}   e^{-iB\tau_1}.
\end{aligned}
\end{equation}

By the triangle inequality, we have
\begin{align}\label{Eq:PF1triangle}
\|{\mathscr{E}(\delta t)\ket{\psi}}\|\le \|{\mathscr{E}_1(\delta t)\ket{\psi}}\| + \|{\mathscr{E}_{\re}(\delta t)\ket{\psi}}\|\le \sqrt{\bra{\psi}\mathscr{E}_1^{\dagger}\mathscr{E}_1\ket{\psi}}+ \|\mathscr{E}_{\re}\|.
\end{align}
The second term can be bounded as
\begin{align}\label{eq:remain_upper_pf1}
 \|\mathscr{E}_{\re}\|\le \frac{t^3}{6}\|[iA,[iA,iB]]\|+ \frac{t^3}{3}\|[iB,[iA,iB]]\|.
\end{align}
For the first term, we have
\begin{equation}
\begin{aligned}
    |\bra{\psi}\mathscr{E}_1^{\dagger}\mathscr{E}_1\ket{\psi}|&=\left|\int_0^{\delta t} \d\tau_1 \int_0^{\delta t} \d\tau_1' \bra{\psi}
[iB,iA]^{\dagger}U_{m}[iB,iA]\ket{\psi}\tau_1\tau_1'\right|\\
&\le \int_0^{\delta t} \d\tau_1 \int_0^{\delta t} \d\tau_1' \, \tau_1\tau_1' \left|\bra{\psi}
[iB,iA]^{\dagger}U_{m}[iB,iA]\ket{\psi}\right|\\
&\le \int_0^{\delta t} \d\tau_1 \int_0^{\delta t} \d\tau_1' \, \tau_1\tau_1' \left|\bra{\psi}
[iB,iA]^{\dagger}[iB,iA]\ket{\psi}\right|\\
&=\frac{\delta t^4}{4} \bra{\psi}[iB,iA]^2 \ket{\psi},
\end{aligned} \label{eq:e_1^2}
\end{equation}
where $U_{m}=e^{iB\tau_1'}e^{iA\tau_1'} e^{iH(\delta t-\tau_1')}e^{-iH(\delta t-\tau_1)}e^{-iB\tau_1}e^{-iA\tau_1}$.
For the second line, we use \cref{Lemma:traceproduct} with $M=[iB,iA]\ket{\psi}$ and $U=U_{m}$, giving
\begin{equation}
\begin{aligned}
   \left|\bra{\psi}[iB,iA]^{\dagger}U_{m}[iB,iA]\ket{\psi}\right|
   &=|\tr(U_{m}[iB,iA]\ket{\psi}\bra{\psi}[iB,iA]^{\dagger})|\\
    &\le |\tr([iB,iA]\ket{\psi}\bra{\psi}[iB,iA]^{\dagger})| \\
    &= \left|\bra{\psi}[iB,iA]^{\dagger}[iB,iA]\ket{\psi}\right|.
\end{aligned}
\end{equation}

According to \cref{SMLemma:local} with $E=[A,B]$, we find
\begin{align}
\bra{\psi}[iB,iA]^{\dagger}[iB,iA]\ket{\psi}=\bra{\psi}[A,B]^2\ket{\psi}\le \|[A,B]\|_F^2+    \sum_{j,j'} \|E_j^{\dagger} E_{j'}\| \, \tr|\rho_{j,j'}- \mathbb{I}/d_{\rho_{j,j'}} |.
\end{align}
Substituting this result in \cref{eq:e_1^2}, we have
\begin{align}\label{eq:E1bound}
\sqrt{\bra{\psi}\mathscr{E}_1^{\dagger}\mathscr{E}_1\ket{\psi}}\le \sqrt{\frac{\delta t^4}{4} \left(\|[A,B]^2\|_F^2+    \sum_{j,j'} \|E_j^{\dagger} E_{j'}\| \, \tr|\rho_{j,j'}- \mathbb{I}/d_{\rho_{j,j'}} | \right)}.
\end{align}
Finally, using \cref{eq:remain_upper_pf1,eq:E1bound} in \cref{Eq:PF1triangle}, the result follows.
\end{proof}

\begin{lemma}\label{SMLemma:ABPF2}
(PF2, Lemma 3 in Methods)
For a two-term Hamiltonian $H=A+B$, consider the second-order product formula $\mathscr{U}_2( \delta t)=e^{-iA/2 \delta t}e^{-iB \delta t}e^{-iA/2\delta t}$ with initial state $\ket{\psi}$. Let $E_1=[B,[B,A]]=\sum_j E_{1,j}$ and $E_2=[A,[A,B]]=\sum_j E_{2,j}$. Then the Trotter error is upper bounded as
\begin{equation}
\begin{aligned}
\|(\mathscr{U}_2(\delta t)-U_0( t))\ket{\psi}\|
&\le \sqrt{\frac{ {\delta t}^6}{144} \|\left[B,\left[B, A\right]\right] \|_F^2 +\Delta_{E_1}(\psi)}
+ \sqrt{\frac{ {\delta t}^6}{576} \|\left[A,\left[A, B\right]\right] \|_F^2 +\Delta_{E_2}(\psi)}
\\
&\quad+ \frac{\delta t^4}{32}\|\left[A,\left[B,\left[B,A\right]\right]\right]\|+  \frac{\delta t^4}{12}\|\left[B,\left[B,\left[B,A\right]\right]\right]\| \\
&\quad+ \frac{\delta t^4}{32}\|\left[B,\left[A,\left[A,B\right]\right]\right]\|+  \frac{\delta t^4}{48}\|\left[A,\left[A,\left[A,B\right]\right]\right]\|,
\end{aligned}
\end{equation}
with
\begin{align}
\Delta_{E_1}(\psi)&=\sum_{j,j'}\|E_{1,j}^{\dag}E_{1,j'}\| \tr|\rho_{j,j'}-\mathbb{I}_{\support(E_{1,j}E_{1,j'})}/d_{\support(E_{1,j}E_{1,j'})}|,\\
\Delta_{E_2}(\psi)&=\sum_{j,j'}\|E_{2,j}^{\dag}E_{2,j'}\| \tr|\rho_{j,j'}-\mathbb{I}_{\support(E_{2,j}E_{2,j'})}/d_{\support(E_{2,j}E_{2,j'})}|.
\end{align}
For sufficiently small $\delta t$, the error is $\mathcal O\left(\delta t^3\left(\sqrt{\|\left[B,\left[B, A\right]\right] \|_F^2 +\Delta_{E_1}(\ket{\psi})}+ \sqrt{\|\left[A,\left[A, B\right]\right] \|_F^2 +\Delta_{E_2}(\ket{\psi}) }    \right)\right)
$.
\end{lemma}

\begin{proof}
According to Appendix L of  Ref.~\cite{childs2020theory},
\begin{equation}
\begin{aligned}\label{Eq:PF2additive}
&\mathscr{E}(\delta t)=\mathscr{U}_2(\delta t)-e^{-iH\delta t}\\
&=\int_0^{\delta t} \d\tau_1\int_0^{\tau_1}\d\tau_2\int_0^{\tau_2} \d\tau_3 \, e^{-i(\delta t-\tau_1)H}e^{-i\tau_1A/2}
e^{-i\tau_1B}
e^{-i\tau_1A/2}
e^{i\tau_1A/2}e^{i\tau_1B}\\
&\quad\cdot \left( e^{-i\tau_3B}\left[-iB,\left[-iB,-i\frac{A}{2}\right]\right]e^{i\tau_3B}+ e^{i\tau_3A/2}\left[i\frac{A}{2},\left[i\frac{A}{2},iB\right]\right]e^{-i\tau_3A/2}\right) e^{-i\tau_1B} e^{-i\tau_1A/2}\\
&= \int_0^{\delta t} \d\tau_1\int_0^{\tau_1}\d\tau_2\int_0^{\tau_2} \d\tau_3 \, e^{-i(\delta t-\tau_1)H}e^{-i\tau_1A/2}
e^{-i\tau_1B}
e^{-i\tau_1A/2}\left(\left[-iB,\left[-iB,-i\frac{A}{2}\right]\right]+ \left[i\frac{A}{2},\left[i\frac{A}{2},iB\right]\right]+ R_1+R_2 \right)\\
&= \mathscr{E}_{2,1}+\mathscr{E}_{2,2} + \mathscr{E}_{\mathrm{re},1}+ \mathscr{E}_{\mathrm{re},1},
\end{aligned}
\end{equation}
where
\begin{equation}
\begin{aligned}
    R_1&= e^{i\tau_1A/2}e^{i\tau_1B}  e^{-i\tau_3B}\left[-iB,\left[-iB,-i\frac{A}{2}\right]\right]e^{i\tau_3B} e^{-i\tau_1B} e^{-i\tau_1A/2} - \left[-iB,\left[-iB,-i\frac{A}{2}\right]\right],\\
    R_2&=e^{i\tau_1A/2}e^{i\tau_1B} e^{i\tau_3A/2}\left[i\frac{A}{2},\left[i\frac{A}{2},iB\right]\right]e^{-i\tau_3A/2} e^{-i\tau_1B} e^{-i\tau_1A/2}- \left[i\frac{A}{2},\left[i\frac{A}{2},iB\right]\right],
\end{aligned}
\end{equation}
and
\begin{equation}
\begin{aligned}
    &\mathscr{E}_{2,1}= \int_0^{\delta t} \d\tau_1\int_0^{\tau_1}\d\tau_2\int_0^{\tau_2} \d\tau_3 \, e^{-i(\delta t-\tau_1)H}e^{-i\tau_1A/2}
e^{-i\tau_1B}
e^{-i\tau_1A/2}\left[-iB,\left[-iB,-i\frac{A}{2}\right]\right],\\
&\mathscr{E}_{2,2}= \int_0^{\delta t} \d\tau_1\int_0^{\tau_1}\d\tau_2\int_0^{\tau_2} \d\tau_3 \, e^{-i(\delta t-\tau_1)H}e^{-i\tau_1A/2}
e^{-i\tau_1B}
e^{-i\tau_1A/2}  \left[i\frac{A}{2},\left[i\frac{A}{2},iB\right]\right],\\
&\mathscr{E}_{\mathrm{re},1}= \int_0^{\delta t} \d\tau_1\int_0^{\tau_1}\d\tau_2\int_0^{\tau_2} \d\tau_3 \, e^{-i(\delta t-\tau_1)H}e^{-i\tau_1A/2}
e^{-i\tau_1B}e^{-i\tau_1A/2}R_1,\\
&\mathscr{E}_{\mathrm{re},2}=\int_0^{\delta t} \d\tau_1\int_0^{\tau_1}\d\tau_2\int_0^{\tau_2} \d\tau_3 \, e^{-i(\delta t-\tau_1)H}e^{-i\tau_1A/2}
e^{-i\tau_1B}e^{-i\tau_1A/2}R_2 .
\end{aligned}
\end{equation}
We have
\begin{equation}
\begin{aligned}
  R_1&= e^{i\tau_1A/2}e^{i\tau_1B}  e^{-i\tau_3B}\left[-iB,\left[-iB,-i\frac{A}{2}\right]\right]e^{i\tau_3B} e^{-i\tau_1B} e^{-i\tau_1A/2 }-\left[-iB,\left[-iB,-i\frac{A}{2}\right]\right]\\
  & =e^{i\tau_1A/2}e^{i\tau_1B}  \left[-iB,\left[-iB,-i\frac{A}{2}\right]\right]e^{-i\tau_1B} e^{-i\tau_1A/2}-\left[-iB,\left[-iB,-i\frac{A}{2}\right]\right]\\
   &\quad+ \int_0^{\tau_3}\d\tau_4 \,
 e^{i\tau_1A/2}e^{i\tau_1B}  e^{-i(\tau_3-\tau_4)B}\left[-iB,\left[-iB,\left[-iB,-i\frac{A}{2}\right]\right]\right]e^{i(\tau_3-\tau_4)B} e^{-i\tau_1B} e^{-i\tau_1A/2}\\
 &=e^{i\tau_1A/2}\left[-iB,\left[-iB,-i\frac{A}{2}\right]\right] e^{-i\tau_1A/2}-\left[-iB,\left[-iB,-i\frac{A}{2}\right]\right]\\
 &\quad+ \int_0^{\tau_1} \d \tau_2 \, e^{i\tau_1A/2}e^{i(\tau_1-\tau_2)B}\left[iB,\left[-iB,\left[-iB,-i\frac{A}{2}\right]\right]\right] e^{-i(\tau_1-\tau_2)B}e^{-i\tau_1A/2}\\
 &\quad+ \int_0^{\tau_3}\d\tau_4 \,
 e^{i\tau_1A/2}e^{i\tau_1B}  e^{-i(\tau_3-\tau_4)B}\left[-iB,\left[-iB,\left[-iB,-i\frac{A}{2}\right]\right]\right]e^{i(\tau_3-\tau_4)B} e^{-i\tau_1B} e^{-i\tau_1A/2}\\
&=\int_0^{\tau_1} \d \tau_1' \, e^{i(\tau_1-\tau_1')A/2}\left[i\frac{A}{2},\left[-iB,\left[-iB,-i\frac{A}{2}\right]\right]\right] e^{-i(\tau_1-\tau_1')A/2}\\
&\quad+ \int_0^{\tau_1} \d \tau_1' \, e^{i\tau_1A/2}e^{i(\tau_1-\tau_1')B}\left[iB,\left[-iB,\left[-iB,-i\frac{A}{2}\right]\right]\right] e^{-i(\tau_1-\tau_1')B}e^{-i\tau_1A/2}\\
 &\quad+ \int_0^{\tau_3}\d\tau_4 \,
 e^{i\tau_1A/2}e^{i\tau_1B}  e^{-i(\tau_3-\tau_4)B}\left[-iB,\left[-iB,\left[-iB,-i\frac{A}{2}\right]\right]\right]e^{i(\tau_3-\tau_4)B} e^{-i\tau_1B} e^{-i\tau_1A/2}.
\end{aligned}
\end{equation}
Therefore
\begin{equation}
     \|R_1\| \le \tau_1 \left \|\left[i\frac{A}{2},\left[-iB,\left[-iB,-i\frac{A}{2}\right]\right]\right]\right\|+ (\tau_1+ \tau_3) \left \|\left[iB,\left[-iB,\left[-iB,-i\frac{A}{2}\right]\right]\right]\right\|,
\end{equation}
which implies
\begin{equation}      \|\mathscr{E}_{\mathrm{re},1}\| \le \frac{{\delta t}^4}{32}\|\left[A,\left[B,\left[B,A\right]\right]\right]\|+  \frac{{\delta t}^4}{12} \|\left[B,\left[B,\left[B,A\right]\right]\right]\|.
\label{eq:remain_pf2_re1}
\end{equation}
Similarly, we have
\begin{equation}
\begin{aligned}
R_2&=e^{i\tau_1A/2}e^{i\tau_1B} e^{i\tau_3A/2}\left[i\frac{A}{2},\left[i\frac{A}{2},iB\right]\right]e^{-i\tau_3A/2} e^{-i\tau_1B} e^{-i\tau_1A/2}-\left[i\frac{A}{2},\left[i\frac{A}{2},iB\right]\right] \\
&=e^{i\tau_1A/2}e^{i\tau_1B} \left[i\frac{A}{2},\left[i\frac{A}{2},iB\right]\right] e^{-i\tau_1B} e^{-i\tau_1A/2}-\left[i\frac{A}{2},\left[i\frac{A}{2},iB\right]\right]
\\
&\quad+ \int_0^{\tau_3}\d \tau_4 \, e^{i\tau_1A/2}e^{i\tau_1B}  e^{i(\tau_3-\tau_4)A/2}   \left[i\frac{A}{2}, \left[i\frac{A}{2},\left[i\frac{A}{2},iB\right]\right]\right]e^{-i(\tau_3-\tau_4)A/2} e^{-i\tau_1B} e^{-i\tau_1A/2}\\
&=e^{i\tau_1A/2}\left[i\frac{A}{2},\left[i\frac{A}{2},iB\right]\right] e^{-i\tau_1A/2}-\left[i\frac{A}{2},\left[i\frac{A}{2},iB\right]\right] \\
&\quad+ \int_0^{\tau_1} \d \tau_1' \,  e^{i\tau_1A/2}e^{i(\tau_1-\tau_1')B}\left[iB, \left[i\frac{A}{2},\left[i\frac{A}{2},iB\right]\right] \right] e^{-i(\tau_1-\tau_1')B} e^{-i\tau_1A/2}\\
&\quad+ \int_0^{\tau_3}\d \tau_4 \, e^{i\tau_1A/2}e^{i\tau_1B}  e^{i(\tau_3-\tau_4)A/2}  \left[i\frac{A}{2}, \left[i\frac{A}{2},\left[i\frac{A}{2},iB\right]\right]\right]e^{-i(\tau_3-\tau_4)A/2} e^{-i\tau_1B} e^{-i\tau_1A/2}\\
&=\int_0^{\tau_1} \d \tau_1' \, e^{i(\tau_1-\tau_1')A/2} \left[i\frac{A}{2}, \left[i\frac{A}{2},\left[i\frac{A}{2},iB\right]\right]\right] e^{-i(\tau_1-\tau_1')A/2} \\
&\quad+ \int_0^{\tau_1} \d \tau_1' \,  e^{i\tau_1A/2}e^{i(\tau_1-\tau_1')B}\left[iB, \left[i\frac{A}{2},\left[i\frac{A}{2},iB\right]\right] \right] e^{-i(\tau_1-\tau_1')B} e^{-i\tau_1A/2}\\
&\quad+ \int_0^{\tau_3}\d \tau_4 \, e^{i\tau_1A/2}e^{i\tau_1B}  e^{i(\tau_3-\tau_4)A/2}  \left[i\frac{A}{2}, \left[i\frac{A}{2},\left[i\frac{A}{2},iB\right]\right]\right]e^{-i(\tau_3-\tau_4)A/2} e^{-i\tau_1B} e^{-i\tau_1A/2},
\end{aligned}
\end{equation}
so
\begin{equation}
\|R_2\| \le (\tau_1+ \tau_3) \left\|\left[i\frac{A}{2}, \left[i\frac{A}{2},\left[i\frac{A}{2},iB\right]\right]\right] \right\|+ \tau_1\left\|\left[B, \left[i\frac{A}{2},\left[i\frac{A}{2},iB\right]\right] \right]\right\|,
\end{equation}
which implies
\begin{equation}
\|\mathscr{E}_{\mathrm{re},2}\| \le \frac{{\delta t}^4}{32}\|\left[B,\left[A,\left[A,B\right]\right]\right]\|+  \frac{{\delta t}^4}{48}\|\left[A,\left[A,\left[A,B\right]\right]\right]\|.
\label{eq:remain_pf2_re2}
\end{equation}

By the triangle inequality,
\begin{equation}
\begin{aligned}
\|{\mathscr{E}(\delta t)\ket{\psi}}\|&\le \|{\mathscr{E}_{2,1}(\delta t)\ket{\psi}}\| + \|{\mathscr{E}_{2,2}(\delta t)\ket{\psi}}\| +\|\mathscr{E}_{\mathrm{re},1}\|+\|\mathscr{E}_{\mathrm{re},2}\|\\
&=\sqrt{\bra{\psi}\mathscr{E}_{2,1}^{\dagger}\mathscr{E}_{2,1}\ket{\psi}}+ \sqrt{\bra{\psi}\mathscr{E}_{2,2}^{\dagger}\mathscr{E}_{2,2}\ket{\psi}}+ \|\mathscr{E}_{\mathrm{re},1}\|+\|\mathscr{E}_{\mathrm{re},2}\|.
\end{aligned}  \label{eq:triangle_pf2}
\end{equation}
Similarly to the proof of \cref{SMLemma:AB}, we can apply \cref{Lemma:traceproduct} to give
\begin{align}
|\bra{\psi}\mathscr{E}_{2,1}^{\dagger}\mathscr{E}_{2,1}\ket{\psi}|
&\le \frac{{\delta t}^6}{144} \bra{\psi}\left[B,\left[B,A\right]\right]^2 \ket{\psi}, \\
|\bra{\psi}\mathscr{E}_{2,2}^{\dagger}\mathscr{E}_{2,2}\ket{\psi}|
&\le \frac{{\delta t}^6}{576} \bra{\psi}\left[A,\left[A,B\right]\right]^2 \ket{\psi}.
\end{align}
According to \cref{SMLemma:local}, taking $E_1=\left[B,\left[B,A\right]\right]=\sum_j E_{1,j}$ and $E_2=\left[A,\left[A,B\right]\right]=\sum_j E_{2,j}$, we have
\begin{align}  |\bra{\psi}\mathscr{E}_{2,1}^{\dagger}\mathscr{E}_{2,1}\ket{\psi}|\leq \frac{{\delta t}^6}{144} \bra{\psi}\left[B,\left[B,A\right]\right]^2 \ket{\psi}
  &\le  \frac{{\delta t}^6}{144} \left(\|\left[B,\left[B,A\right]\right]\|_F^2+\Delta_{E_1}(\psi) \right), \label{eq:e21}\\  |\bra{\psi}\mathscr{E}_{2,2}^{\dagger}\mathscr{E}_{2,2}\ket{\psi}|
\le \frac{{\delta t}^6}{576} \bra{\psi}\left[A,\left[A,B\right]\right]^2 \ket{\psi}
  &\le  \frac{{\delta t}^6}{576} \left(\|\left[A,\left[A,B\right]\right]\|_F^2+\Delta_{E_2}(\psi) \right). \label{eq:e22}
\end{align}
Substituting the upper bounds of \cref{eq:remain_pf2_re1,eq:remain_pf2_re2,eq:e21,eq:e22} into \cref{eq:triangle_pf2}, the proof follows.
\end{proof}

\section{Classical simulation with MPS}

\begin{figure}[h]
    \centering
    \includegraphics[scale=0.7]{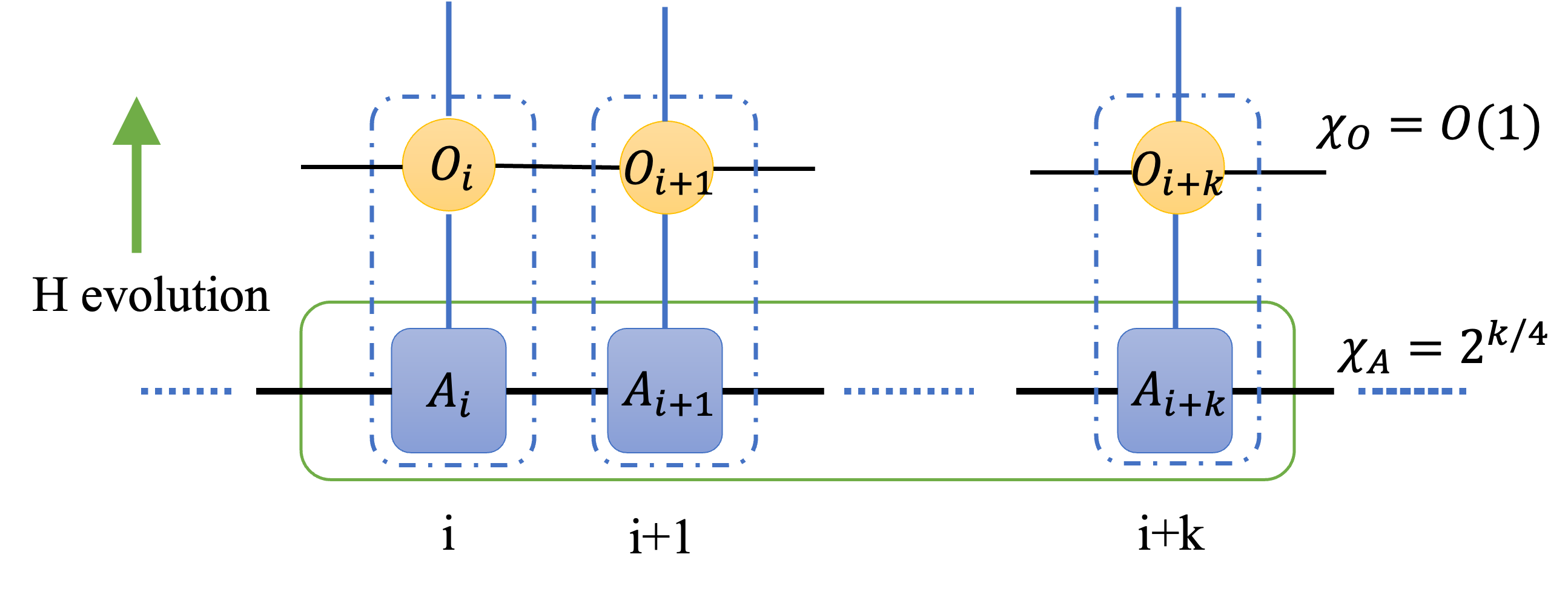}
    \caption{Classical simulation with MPS. The blue box labeled $A_i$ represents an elementary tensor, each of which has bond dimension $\chi_A=2^{k/4}$ for a $k$-uniform state. The short-time evolution can be mapped to an MPO \cite{paeckel2019time} whose elementary tensors are represented by the orange circles labeled $O_i$. To get the updated state after the evolution, we contract the local tensors $A_i$ and $O_i$ together in the dot-dashed blue boxes.
    }
    \label{fig:MPS}
\end{figure}

As discussed in the main text around Fig.~3(a), the cost of quantum simulation reduces to the average case if the underlying state is $k$-uniform with $k$ at least some constant. On the other hand, the cost of classical simulation with matrix-product states (MPS) scales exponentially with $k$. We illustrate this as follows.
As the underlying state is $k$-uniform---that is, the $k$-body RDM is the maximally mixed state---the bond dimension of the MPS should be around $\chi_A=2^{k/4}$, as illustrated in \cref{fig:MPS}. For a short-time evolution, we can use the matrix-product operator (MPO), shown in orange in \cref{fig:MPS} \cite{paeckel2019time}. The evolved state is obtained by contracting the local tensors $A_i$ and $O_i$ in each block to get the evolved state. The computational cost of this process is $4\chi_A^2\chi_O^2N=\mathcal{O}(2^{k/2})$.

\section{Modified bounds for (single-segment) Trotter error}\label{sec:measured_error}

In this section, we give modified versions of the distance- and entanglement-based error bounds in terms of more readily determined quantities. Both of these bounds relate to the purities of the local RDMs. In addition, we give a refined bound at the end of this section that reduces the measurement budget further.

Trace distance is hard to measure, but one can use the (normalized) Frobenius norm to bound it, since $\tr |M|\le d_M\|M\|_F$ for an operator $M$ of dimension $d_M$. In this way, the leading-order of the distance-based error bound in \cref{SMthm:main} can be further bounded by
\begin{equation}\label{eq:ErrD2App}
\begin{aligned}
&\left[\left(\sum_{j\neq j'} \|E_j^{\dag}E_{j'}\| \tr|\rho_{j,j'}-\mathbb{I}/d_{\rho_{j,j'}}|\right)^{1/2}+ \|E\|_F \right]\delta t^{p+1}\\
&\quad\leq \left[\max_{j}\|E_j\| \left(\sum_{j\neq j'} \sqrt{d_{\rho_{j,j'}}\tr(\rho_{j,j'}-\mathbb{I}/d_{\rho_{j,j'}})^2}\right)^{1/2}+ \|E\|_F \right]\delta t^{p+1}\\
&\quad\leq \left[\max_{j}\|E_j\| \left(\sum_{j\neq j'} \sqrt{d_{\rho_{j,j'}}\tr(\rho_{j,j'}^2)-1}\right)^{1/2}+ \|E\|_F \right]\delta t^{p+1}.
\end{aligned}
\end{equation}

Similarly, the von Neumann entanglement entropy is hard to measure, but we can lower bound it by the R{\'e}nyi-2 entropy, $S_2(\rho_{j,j'})=-\log_2\tr(\rho_{j,j'}^2)$. This gives the following modification of the entanglement-based error bound in \cref{SMthm:main}:
\begin{equation}\label{eq:ErrEntSM}
\begin{aligned}
&\left[\left(\sum_{j\neq j'} \|E_j^{\dag}E_{j'}\|\sqrt{\log (d_{\rho_{j,j'}})-S(\rho_{j,j'})}\right)^{1/2}+  \|E\|_F\right]\delta t^{p+1} \\
&\quad\leq \mathcal O\left[\max_{j}\|E_j\| \left(\sum_{j,j'} \sqrt{n_{j,j'}-S_2(\rho_{j,j'})}\right)^{1/2}+ \|E\|_F \right]\delta t^{p+1},
\end{aligned}
\end{equation}
where $n_{j,j'}=|\support(E_j^{\dag}E_{j'})|$ is the number of qubits in the corresponding subsystem. Observe that both modified bounds depend on the purities $\tr(\rho_{j,j'}^2)$ of the RDMs of the subsystems $\support(E_j^{\dag}E_{j'})$.

Finally, we show a more refined bound for the single-segment Trotter error. According to
\cref{EqSMTh1-1} in the proof of \cref{SMthm:main}, the leading term of the Trotter error is directly related to $\sum_{j} E_j$. Specifically, we have
\begin{equation}\label{SM:refineE}
\begin{aligned}
   \left\| \sum_{j} E_j \ket{\psi}\right \|  \delta t^{p+1}
   &=  \delta t^{p+1} \sqrt{\sum_{j,j'} \bra{\psi}E_j^{\dagger} E_{j'}\ket{\psi} } \\
   &=\delta t^{p+1}\sqrt{\sum_{j< j'} \bra{\psi} (E_j^{\dag}  E_{j'}+ \mathrm{h.c.}) \ket{\psi} +\sum_j \bra{\psi}E_j^{\dag}  E_{j}\ket{\psi}}.
\end{aligned}
\end{equation}
Unlike the previous bounds, which depend on purities, this bound
directly estimates the functions $O=\sum_w \bra{\psi_i}O_w\ket{\psi_i}$ with $O_w=E_j^{\dag}  E_{j'}+ \mathrm{h.c.}$ for $j<j'$ and $O_w=E_j^{\dag}  E_{j}$ for  $j=j'$, which are linear combinations of a few local Pauli operators. Compared to the purity estimation in the previous bounds, the refined bound in \cref{SM:refineE} reduces the measurement complexity for estimating the Trotter error to $M=\mc O(N^2)$ for a lattice Hamiltonian, as shown in Methods section of main text. The refined bound is then used in the measurement-assisted adaptive Hamiltonian simulation, and the numerical result is shown in Fig.~4(d) in the main text.

\section{Numerical results}

To demonstrate that highly entangled states are typical, we choose $10^4$ pure quantum states uniformly from a 12-qubit Hilbert space and calculate the entropy of a 4-qubit subsystem for each instance as shown in \cref{fig:typicalentanglement}. Most of the states have large entanglement entropy, close to the largest possible value of $4$.

\begin{figure}
  \includegraphics[width=0.5\columnwidth]{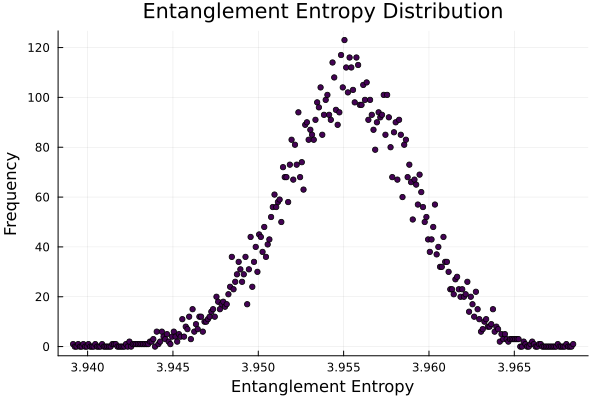}
    \caption{
    Entanglement entropy of 4-qubit RDMs from uniformly random 12-qubit states.}
    \label{fig:typicalentanglement}
\end{figure}

To numerically assess the performance of product-formula simulations as a function of entanglement, we consider the 1D quantum Ising spin system with mixed
fields (QIMF) \cite{cotler2021emergent},
\begin{align}
  H= h_x \sum_{j=1}^N X_j + h_y \sum_{j=1}^N Y_j + J \sum_{j=1}^{N-1} X_jX_{j+1},
\end{align}
with the initial state $\ket{0}^{\otimes n}$.
We consider the parameters $(h_x,h_y,J)= (0.8090, 0.9045, 1)$ and $(h_x,h_y,J)= (0, 0.9045, 1)$
as typical (i.e., satisfying ETH) \cite{kim2014testing} and atypical examples, respectively.

To apply product-formula simulation, we use an X-Y pattern, with
\begin{align}
  A = h_x \sum_{j=1}^N X_j  + J \sum_{j=1}^{N-1} X_jX_{j+1}, \quad
  B = h_y \sum_{j=1}^N Y_j.
\end{align}
The corresponding PF1 and PF2 commutators are
\begin{equation}
\begin{aligned} \relax
  [A, B]&= 2ih_xh_y  \sum_{j=1}^N Z_j  +2i Jh_y \sum_{j=1}^{N-1} (Z_jX_{j+1}+ X_jZ_{j+1}),\\
  [A, [A, B]]&= 4h_x^2h_y    \sum_{j=1}^N Y_j+
  4 J^2 h_y \sum_{j=1}^{N-1} Y_j+
   4 J^2 h_y \sum_{j=2}^{N} Y_j \\
   &\quad+
  8 Jh_xh_y \sum_{j=1}^{N-1} (Y_jX_{j+1}+ X_jY_{j+1}) + 8J^2 h_y \sum_{j=1}^{N-2} (X_jY_{j+1}X_{j+2}),\\
  [B, [A, B]]&= -4 h_xh_y^2  \sum_{j=1}^N X_j  + 8Jh_y^2 \sum_{j=1}^{N-1} ( Z_jZ_{j+1}-X_jX_{j+1}).
\end{aligned} \label{SMeq:commutators}
\end{equation}

For the theoretical worst- and average-case error bounds, we use the upper bounds for PF1 and PF2 in Refs.~\cite{childs2020theory,Zhaorandom22}. which we briefly summarrize as follows. For PF1, the  worst- and average-case error bounds are
\begin{align}
\|\mathscr{U}_1(\delta t)-U_0(\delta t)\|
&\le \frac{\delta t^2}{2}\|[A,B]\|,\\
\|\mathscr{U}_1(\delta t)-U_0(\delta t)\|_F
&\le \frac{\delta t^2}{2}\|[A,B]\|_F.
\end{align}
For the PF2, the  worst- and average-case error bounds are
\begin{align}
&\|\mathscr{U}_2(\delta t)-U_0(\delta t)\|\le \frac{\delta t^3}{12}
 \|\left[B,\left[B, A\right]\right] \| + \frac{\delta t^3}{24}
 \|\left[A,\left[A, B\right]\right] \|,\\
 &\|\mathscr{U}_2(\delta t)-U_0(\delta t)\|_F\le \frac{\delta t^3}{12}
 \|\left[B,\left[B, A\right]\right] \|_F+ \frac{\delta t^3}{24}
 \|\left[A,\left[A, B\right]\right] \|_F.
\end{align}
To bound the spectral norms $\|\cdot\|$ and the Frobenius norms $\|\cdot\|_F$ in these expressions,
we count the number of Pauli operators in \cref{SMeq:commutators} rather than evaluating their norms numerically. Concretely, we have
\begin{equation}
\begin{aligned}
  \|[A, B]\| &\le  2h_xh_yN+ 4h_yJ(N-1),\\
  \|[A, [A, B]]\|&\le 4h_x^2h_yN+8J^2h_y(N-1)+16Jh_yh_x(N-1)+8J^2h_y(N-2),\\
  \|[B, [A, B]]\|&\le  4h_xh_y^2N+16Jh_y^2(N-1)
\end{aligned} \label{eq:count_worst}
\end{equation}
and
\begin{equation}
\begin{aligned}
  \|[A, B]\|_F^2 &\le  4h_x^2h_y^2N+ 8h_yJ(N-1),\\
  \|[A, [A, B]]\|_F^2&\le 16(h_x^2+2J^2)^2h_y^2N+128 h_x^2h_y^2J^2(N-1)+128J^2h_y^2(N-2),\\
  \|[B, [A, B]]\|_F^2&\le  16h_x^2h_y^4N+128J^2h_y^2(N-1).\\
\end{aligned} \label{eq:count_average}
\end{equation}
To apply \cref{SMLemma:ABPF2}, we also need some more complicated nested commutators, namely
$\|[A,[A, [A, B]]]\|$, $\|[B,[A, [A, B]]]\|$, $\|[A, [B, [A, B]]]\|$, and $\|[B, [B, [A, B]]]\|$. We use the following upper bound:
\begin{align}
    \|[A,[A, [A, B]]]\|\le 2 \|A\| \, \|[A, [A, B]]\|.
\end{align}
For the other nested commutators $\|[B,[A, [A, B]]]\|$, $\|[A, [B, [A, B]]]\|$, and $\|[B, [B, [A, B]]]\|$, we apply the same method. We use $\|A\|\le h_xN+J(N-1)$ and $\|B\|\le h_yN$.

To apply our distance-based theoretical bounds forr a long-time evolution, we use the triangle inequality and obtain
\begin{align}
\|{\mathscr{U}_2(t)\ket{\psi}-U_0(t)\ket{\psi}}\|\le \sum_{j=0}^{r-1}\|{\mathscr{U}_2(t/r)\ket{\psi_j}-U_0(t/r)\ket{\psi_j}}\|,
\end{align}
where $\ket{\psi_j}=U_0(tj/r)\ket{\psi_0}$.
For each short-time evolution, we apply \cref{SMLemma:AB,SMLemma:ABPF2}, taking $\ket{\psi_j}$ as the input state. For all spectral and Frobenius norms of operators in \cref{SMLemma:AB,SMLemma:ABPF2}, we use the counting bounds presented in \cref{eq:count_worst,eq:count_average}. To estimate $\Delta_E$, we assume that we can obtain the distance $\tr|\rho_{j,j'}-\mathbb{I}_{\support(E_jE_{j'})}/d_{\support(E_jE_{j'})}|$
for different $E_jE_{j'}$. This property can be estimated numerically as discussed in \cref{sec:measured_error}. In this part, we only want to demonstrate the validity of our new theoretical bounds, so we do not consider a procedure for estimating the distance.

To determine the minimum number of Trotter steps $r$ that suffice to ensure error at most $\epsilon=10^{-5}$, as shown in Fig.~5(c) of main text, we need to keep track of the error for each segment. This would involve many calculations of the trace distance of the evolved state, which would be computationally intensive. Thus, rather than estimating $\Delta_E$ and calculating the Trotter error for each step, we assume that the entanglement remains stable for some time, so that we can use the estimated $\Delta_E$ for a few subsequent Trotter steps.
Specifically, we divide the total time duration $t$ into $C$ slices and make the algorithmic error in each slice at most $\varepsilon/C$.
We use the distance information for the evolved state $\ket{\psi(t_c)}$ for each $c\in\{0, 1, \dots, C-1\}$, with $t_c= ct/C$, and denote the corresponding distance-based upper bound of \cref{SMLemma:ABPF2} by
$\|\mathscr{U}_2(t/r_c)\ket{\psi(t_c)}-U_0(t/r_c)\ket{\psi(t_c)}\|^*$.
We calculate the minimum $r_c$ for the $c$th slice such that
\begin{align}
r_c \|\mathscr{U}_2(t/r_c)\ket{\psi(t_c)}-U_0(t/r_c)\ket{\psi(t_c)}\|^*\le \varepsilon/C.
\end{align}
Then the total number of Trotter steps is
$r^*=\sum_c r_c$.

This compromise results in a value of $r^*$ greater than the real theoretically predicted $r$ that would result by considering the state in each segment under the assumption that local entanglement entropy is non-decreasing.
However, the difference is not significant, especially for a large $C$.
We increase $C$ until the value $r^*$
changes by less than 1\% and use this as our estimate of $r$. In our numerics, $r^*$ converges before $C=20$.

\end{document}